\documentclass[a4paper,notitlepage,fleqn,12pt]{article}
\usepackage[tbtags]{amsmath}
\usepackage{amsfonts}
\usepackage{amsthm}
\usepackage{color}
\usepackage{graphicx,amssymb}
\usepackage{multirow}
\usepackage{rotating}
\usepackage{mathrsfs}
\usepackage{epstopdf}
\usepackage{enumerate}
\usepackage{pdflscape}
\usepackage{geometry}
\usepackage{natbib}
\usepackage{changepage}
\usepackage{dcolumn}
\usepackage{tabularx}
\usepackage{subfig}
\usepackage{mathabx}

\newcolumntype{d}[1]{D{.}{.}{#1}}
\newcolumntype{Y}{>{\raggedleft\arraybackslash}X}
\newcolumntype{Z}{>{\centering\arraybackslash}X}
\newtheorem{lemma}{Lemma}[section]

\newtheorem{cor}{Corollary}
\newtheorem{assum}{Assumption}
\newtheorem{thm}{Theorem}
\theoremstyle{remark}

\DeclareMathOperator*{\argmin}{argmin}

\DeclareMathOperator*{\sgn}{sgn}

\DeclareMathOperator*{\var}{var}

\numberwithin{equation}{section}
\input{tcilatex}
\begin{document}

\title{\vspace{-10mm}Hybrid Quantile Regression Estimation for Time Series
Models with Conditional Heteroscedasticity}
\author{Yao Zheng, Qianqian Zhu, Guodong Li and Zhijie Xiao \\
\textit{University of Hong Kong and Boston College} }

\date{}
\maketitle

\begin{abstract}
Estimating conditional quantiles of financial time series is essential for
risk management and many other applications in finance. It is well-known
that financial time series display conditional heteroscedasticity. Among the
large number of conditional heteroscedastic models, the generalized
autoregressive conditional heteroscedastic (GARCH) process is the most
popular and influential one. So far, feasible quantile regression methods
for this task have been confined to a variant of the GARCH model, the linear
GARCH model, owing to its tractable conditional quantile structure. This
paper considers the widely used GARCH model. An easy-to-implement hybrid
conditional quantile estimation procedure is developed based on a simple
albeit nontrivial transformation. Asymptotic properties of the proposed
estimator and statistics are derived, which facilitate corresponding
inferences. To approximate the asymptotic distribution of the quantile
regression estimator, we introduce a mixed bootstrapping procedure, where a
time-consuming optimization is replaced by a sample averaging. Moreover,
diagnostic tools based on the residual quantile autocorrelation function are
constructed to check the adequacy of the fitted conditional quantiles.
Simulation experiments are carried out to assess the finite-sample
performance of the proposed approach. The favorable performance of the
conditional quantile estimator and the usefulness of the inference tools are
further illustrated by an empirical application.
\end{abstract}


\newpage

\section{Introduction}

Time series models with conditional heteroscedasticity have become extremely
popular in financial applications since the appearance of %
\citeauthor{Engle1982}'s \citeyearpar{Engle1982} autoregressive conditional
heteroscedastic (ARCH) model and \citeauthor{Bollerslev1986}'s %
\citeyearpar{Bollerslev1986} generalized autoregressive conditional
heteroscedastic (GARCH) model; see also \cite{Francq_Zakoian2010}. These
models are widely used in the assessment and management of financial risk,
including the estimation of quantile-based measures such as the
Value-at-Risk (VaR) and the Expected Shortfall (ES). Spurred by the need of
various financial institutions and regulatory authorities, quantile-based
measures now play an important part in quantitative analysis and investment
decision making. For this reason, estimating conditional quantiles of
financial time series is crucial to both academic researchers and
professional practitioners in many areas of economics and finance. Furthermore, as
conditional quantiles can be directly estimated by quantile regression %
\citep{Koenker_Bassett1978}, it is especially appealing to study the
conditional quantile inference for conditional heteroscedastic models via
quantile regression.

Among the large number of conditional heteroscedastic models, arguably the
most popular and influential one is \citeauthor{Bollerslev1986}'s %
\citeyearpar{Bollerslev1986} GARCH model, since its specification is intuitive,
parsimonious and readily interpretable. It has proven highly
successful in capturing the volatility clustering of financial time series,
and therefore has been frequently integrated into the areas of asset
pricing, asset management and financial risk management.  The GARCH$(p,q)$ model can
be written as
\begin{equation}
x_{t}=\sqrt{h_{t}}\eta _{t},\hspace{5mm}h_{t}=\alpha_{0}+\sum_{i=1}^{q}%
\alpha _{i}x_{t-i}^{2}+\sum_{j=1}^{p}\beta _{j}h_{t-j},  \label{garch}
\end{equation}
where $\{\eta_t\}$ is a sequence of independent and identically distributed (%
$i.i.d.$) innovations with mean zero and variance one. Despite the
fast-growing interest in conditional quantile inference for time series
models \citep{Koenker2005}, the literature on quantile regression
 for the GARCH model is relatively sparse due to
technical difficulties in the estimation. Specifically, consider the conditional quantile of the GARCH process given by \eqref{garch},
\begin{equation}  \label{gquantile}
Q_{\tau}(x_t|\mathcal{F}_{t-1})= Q_{\tau, \eta}\sqrt{\alpha_0+\sum_{i=1}^q%
\alpha_ix_{t-i}^2+\sum_{j=1}^p\beta_jh_{t-j}},\quad 0<\tau<1,
\end{equation}
where $Q_{\tau, \eta}$ is the $\tau$th quantile of $\eta_t$, and $\mathcal{F}%
_{t}$ is the information set available at time $t$. The
\emph{square-root} function in \eqref{gquantile}, together with the
non-smooth loss function in quantile regression, $\rho _{\tau
}(x)=x[\tau -I(x<0)]$, leads to a non-smooth objective function which is
non-convex even for the ARCH special case.  It is this feature that causes the challenges in asymptotic derivation and numerical optimization, and the problem is even more complicated in view of the recursive structure of the conditional variances $\{h_{t}\}$.

On account of these difficulties, the previous literature considered quantile regression
estimation for \citeauthor{Taylor1986}'s \citeyearpar{Taylor1986} linear ARCH
(LARCH) or linear GARCH (LGARCH) models. In particular, an LGARCH($p,q$)
model has the following form,
\begin{equation}
y_{t}=\sigma _{t}\varepsilon _{t},\hspace{5mm}\sigma _{t}=\alpha
_{0}+\sum_{i=1}^{q}\alpha _{i}|y_{t-i}|+\sum_{j=1}^{p}\beta _{j}\sigma
_{t-j},  \label{lgarch}
\end{equation}%
where $\{\varepsilon _{t}\}$ is a sequence of $i.i.d.$ innovations with mean
zero. Its conditional quantile has a much simpler form,
\begin{equation}
Q_{\tau }(y_{t}|\mathcal{F}_{t-1})=\left( \alpha _{0}+\sum_{i=1}^{q}\alpha
_{i}|y_{t-i}|+\sum_{j=1}^{p}\beta _{j}\sigma _{t-j}\right) Q_{\tau
,\varepsilon },\quad 0<\tau <1,  \label{lquantile}
\end{equation}%
where $Q_{\tau ,\varepsilon }$ is the $\tau $th quantile of $\varepsilon
_{t} $. \cite{Koenker_Zhao1996} first considered the conditional quantile
estimation for the LARCH($q$) model, which, without any $\sigma _{t-j}$
involved in \eqref{lquantile}, reduces to a linear quantile
regression problem. Quantile regression for the LGARCH model, in contrast, is more troublesome due to the recursive structure of the conditional
scales $\{\sigma_t\}$. To tackle this, \cite{Xiao_Koenker2009}
proposed a two-stage scheme, where they replaced the unobservable $\sigma _{t-j}$'s
in \eqref{lquantile} with some initial estimates first, enabling a linear quantile regression at the second stage. Nevertheless, most practitioners and researchers still prefer \citeauthor{Bollerslev1986}'s \citeyearpar{Bollerslev1986} GARCH model in \eqref{garch}. For this reason, \cite{Lee_Noh2013} studied the asymptotic
properties of a quantile regression estimator for the GARCH model, without addressing the feasibility of the numerical optimization for this estimator. For a detailed discussion on the algorithmic issues in quantile regression, see \cite{Koenker_Park1996}.

The purpose of this paper is to develop an easy-to-implement approach to the
conditional quantile estimation and inference for \citeauthor{Bollerslev1986}%
's \citeyearpar{Bollerslev1986} original GARCH model given by \eqref{garch}.
To overcome the aforementioned difficulties, we design
the following transformation $T:\mathbb{R}\rightarrow \mathbb{R}$ for the
conditional quantile in \eqref{gquantile},
\begin{equation}
T(x)=x^{2}\sgn(x),  \label{transform}
\end{equation}%
where $\sgn(\cdot )$ is the sign function. Note that there are two
desirable properties of $T(\cdot)$:

\begin{itemize}
\item[(a)] it is the inverse of the square-root function \emph{in some sense};

\item[(b)] it is continuous and nondecreasing on $\mathbb{R}$.
\end{itemize}

Owing to this design of $T(\cdot)$, the conditional quantile of the
transformed sequence $\{T(x_{t})\}$ resembles that of the LGARCH process $%
\{y_{t}\} $ in \eqref{lquantile}, in that
\begin{equation}
Q_{\tau }[T(x_{t})|\mathcal{F}_{t-1}]=\left( \alpha
_{0}+\sum_{i=1}^{q}\alpha _{i}x_{t-i}^{2}+\sum_{j=1}^{p}\beta
_{j}h_{t-j}\right) T(Q_{\tau ,\eta }),  \label{lquantile2}
\end{equation}%
where $x_{t-i}^2=|T(x_{t-i})|$. This connects the conditional quantile
inference of the GARCH model directly to that of the LGARCH model. As
a result of this connection, we can estimate $Q_{\tau }(x_{t}|\mathcal{F}_{t-1})$  through estimating $Q_{\tau }[T(x_{t})|\mathcal{F}_{t-1}]$. Specifically, we can first estimate $Q_{\tau }[T(x_{t})|\mathcal{F}_{t-1}]$  via linear quantile regression with some initial estimates of $\{h_{t}\}$.  Then, by applying the inverse transformation $T^{-1}(\cdot )$ to the estimator of $Q_{\tau }[T(x_{t})|\mathcal{F}_{t-1}]$, we can obtain that of $Q_{\tau }(x_{t}|\mathcal{F}_{t-1})$,  owing to the monotonicity of the transformation.

The quantile regression based on \eqref{lquantile2} requires appropriate initial
estimates of the conditional variances $\{h_{t}\}$. In \cite%
{Xiao_Koenker2009}, the conditional scales $\{\sigma _{t}\}$ of the LGARCH
process \eqref{lgarch} are estimated based on a sieve approximation of $%
\sigma _{t}$ with an $m$th-order linear ARCH model: $\sigma _{t}=\gamma
_{0}+\sum_{j=1}^{m}\gamma _{j}|y_{t-j}|$,  with $m\rightarrow \infty$. A similar sieve approximation may
be used on the GARCH model (\ref{garch}) based on $h_{t}=\gamma
_{0}+\sum_{j=1}^{m}\gamma _{j}x_{t-j}^{2}$.
However, the tunning parameter $m$ heavily affects the numerical stability
of the procedure: e.g., larger $\alpha _{i}$ and $\beta _{j}$ would require
bigger $m$, but unnecessarily large $m$ can introduce too much noise into
the estimation; see the Monte Carlo evidence in Section 5.1.
On account of this, we estimate $\{h_{t}\}$ by the
Gaussian quasi-maximum likelihood estimator (QMLE) for the GARCH model. The
asymptotic normality of this estimator under mild technical conditions is
established by \cite{Francq_Zakoian2004}, and it is easier to implement as well as
numerically more stable than the sieve method.   Therefore, in this paper,  a hybrid conditional quantile estimator for the GARCH model  is constructed based on two
estimators of different nature: the Gaussian QMLE, which incorporates the global model structure, and the
quantile regression estimator, which approximates the conditional quantiles
locally.

We derive the asymptotic properties of the proposed estimator and statistics. These limiting results facilitate the statistical inference in this paper.
On the other hand, a sparsity/density function enters the asymptotic distribution of
the quantile regression estimator, and any feasible inference procedure
requires that the density is handled appropriately. Estimation of the
density function, although possible, is usually complicated and depends on
additional tunning parameters. The preliminary estimation of the density
function seriously affects the finite-sample performance of the inference
procedures. For this reason, we propose a bootstrap method to approximate
the distribution.

\cite{Jin_Ying_Wei2001} considered a bootstrap method by perturbing the
minimand of the objective function with random weights, which is especially
useful for time series models as the observations are ordered by time;
see also \cite{Rao_Zhao1992}, \cite{Li_Leng_Tsai2014} and \cite{Zhu2016}.
Applying this method to our context, we may conduct a randomly weighted QMLE
first, followed by a randomly weighted linear quantile regression. Nonetheless,  since the sparsity/density function is not involved in the asymptotic distribution of the QMLE, the first bootstrapping step is actually unnecessary. In view of this, we propose a mixed method: we suggest replacing the
first step with a sample averaging, so that the time-consuming optimization need only be
performed in the second bootstrapping step. A significant reduction in the
computation time hence results.

The asymptotic results and the proposed bootstrap method are useful for conditional quantile inference. For example, the bootstrapping procedure enables us to construct confidence intervals for
the fitted conditional quantiles, which may be especially interesting in
practice.  Furthermore, adopting Box-Jenkins' three-stage modeling strategy %
\citep{Box_Jenkins_Reinsel2008}, we consider diagnostic checking for the
fitted conditional quantiles. For conditional heteroscedastic models,
diagnostic tools based on the sample autocorrelation function (ACF)
of squared residuals \citep{Li_Mak1994} or absolute residuals %
\citep{Li_Li2005} are commonly used; see \cite{Li2004} for a review on
diagnostic checks of time series. In conditional quantile inference,
\cite{Li_Li_Tsai2015} proposed the quantile autocorrelation function (QACF)
and used it to develop goodness-of-fit tests for quantile autoregressive models %
\citep{Koenker_Xiao2006}.  Motivated by these, we construct diagnostic tools for the fitted conditional quantiles by introducing a suitable residual QACF in this paper.

The rest of the paper is organized as follows. Section 2 introduces the
hybrid conditional quantile estimator for GARCH models, and Section 3
proposes the mixed bootstrapping approximation procedure. Section 4
considers diagnostic checking for the fitted conditional quantiles.
Section 5 conducts extensive simulation experiments to assess the finite-sample
performance of the proposed inference tools; a comparison with existing
conditional quantile estimators is also provided. Section 6 presents
an empirical application, and Section 7 gives a short conclusion and
discussion. All technical details are relegated to the appendix. Throughout
the paper, $\rightarrow_d$ denotes the convergence in distribution, $o_p(1)$
denotes a sequence of random variables converging to zero in probability,
and the notation $o_p^*(1)$ corresponds to the bootstrapped probability
space.

\section{The Proposed Hybrid Conditional Quantile Estimation Procedure}

Let $\{x_t\}$ be a strictly stationary and ergodic time series generated by
the GARCH model in \eqref{garch}, where $\alpha_0>0$, $\alpha_i\geq 0$ for $%
1\leq i \leq q$, $\beta_j\geq 0$ for $1\leq j \leq p$; see \cite%
{Bollerslev1986}. The necessary and sufficient condition for the existence
of a unique strictly stationary and ergodic solution to this model is given
in \cite{Bougerol_Picard1992}.

Denote by $\mathcal{F}_{t}$ the $\sigma $-field generated by $%
\{x_{t},x_{t-1},\ldots \}$. Let $y_{t}=T(x_{t})$ where $T(\cdot )$ is
defined by \eqref{transform}, and denote $b_{\tau }=T(Q_{\tau ,\eta })$ with
$Q_{\tau ,\eta }$ being the $\tau $th quantile of $\eta _{t}$. From %
\eqref{lquantile2}, the $\tau$th quantile of the transformed variable $y_{t}$
conditional on $\mathcal{F}_{t-1}$ is
\begin{equation}
Q_{\tau }(y_{t}|\mathcal{F}_{t-1})=b_{\tau }\left(
\alpha_{0}+\sum_{i=1}^{q}\alpha
_{i}x_{t-i}^{2}+\sum_{j=1}^{p}\beta_{j}h_{t-j}\right) =\theta _{\tau
}^{\prime }z_{t},\hspace{5mm}0<\tau <1,  \label{model}
\end{equation}
where
\begin{equation*}
z_{t}=(1,x_{t-1}^{2},\dots,x_{t-q}^{2},h_{t-1},\dots,h_{t-p})^{\prime}\quad%
\text{and}\quad\theta _{\tau }=b_{\tau }(\alpha _{0},\alpha
_{1},\dots,\alpha_{q},\beta _{1},\dots,\beta _{p})^{\prime }.
\end{equation*}
If $\{h_{t}\}$ were known, then $Q_{\tau }(y_{t}|\mathcal{F}_{t-1})$ would
be linear in $\theta _{\tau }$, and one could estimate $Q_{\tau }(y_{t}|%
\mathcal{F}_{t-1})$ via a linear quantile regression on the transformed
model. In practice, this quantity can also be estimated with appropriate
initial estimates of $\{h_{t}\}$.

Denote by $\theta=(\alpha_0,\alpha_1,\dots, \alpha_q,
\beta_1,\dots,\beta_p)^{\prime}$ the parameter vector of model \eqref{garch}%
. Let $0<\underline{w}<\overline{w}$, $0<\rho_0<1$, $p\underline{w}<\rho_0$,
and define
\begin{align*}
\Theta=\{\theta: \beta_1+\cdots+\beta_p\leq \rho_0, \; \underline{w}%
&\leq\min(\alpha_0, \alpha_1, \dots, \alpha_q, \beta_1,\dots, \beta_p) \\
&\leq\max(\alpha_0, \alpha_1, \dots, \alpha_q, \beta_1,\dots, \beta_p)\leq%
\overline{w}\} \subset \mathbb{R}_+^{p+q+1},
\end{align*}
where $\mathbb{R}_+=(0,\infty)$; see \cite{Berkes_Horvath2004}. The true
value of $\theta$ is denoted by $\theta_{0}=(\alpha_{00},\alpha_{01},\dots,%
\alpha_{0q}, \beta_{01},\dots, \beta_{0p})^{\prime}$. Moreover, we define
the functions $h_t(\theta)$ recursively by
\begin{equation}  \label{aeq1}
h_t(\theta)=\alpha_0+\sum_{i=1}^q\alpha_ix_{t-i}^2+\sum_{j=1}^p%
\beta_jh_{t-j}(\theta).
\end{equation}
Note that $h_t(\theta_0)=h_t$. As \eqref{aeq1} depends on
infinite past observations, initial values for $\{x_0^2,\dots, x_{1-q}^2,
h_0, \dots, h_{1-p}\}$ are needed. This however does not affect our
asymptotic results. We set all initial values to $n^{-1}\sum_{t=1}^{n}x_t^2$
and denote the resulting $h_t(\theta)$ by $\widetilde{h}_t(\theta)$.

We propose the hybrid conditional quantile estimation procedure as follows.
\begin{itemize}
\item \textit{Step E1 (Initial estimation of the original model).} Perform
the Gaussian quasi-maximum likelihood estimation (QMLE) of model %
\eqref{garch},
\begin{equation}  \label{QMLE}
\widetilde{\theta}_n=\argmin_{\theta\in\Theta}\sum_{t=1}^{n}\widetilde{\ell}%
_t(\theta),
\end{equation}
where $\widetilde{\ell}_t(\theta)=x_t^2/\widetilde{h}_t(\theta)+\log%
\widetilde{h}_t(\theta)$; see \cite{Francq_Zakoian2004}. Then compute the initial
estimates of $\{h_t\}$ as $\widetilde{h}_t=\widetilde{h}_t(\widetilde{\theta}%
_n)$.

\item \textit{Step E2 (Quantile regression of the transformed model).}
Perform the weighted linear quantile regression of $y_{t}$ on $\widetilde{z}%
_{t}=(1,x_{t-1}^{2},\dots,x_{t-q}^{2},\widetilde{h}_{t-1},\ldots ,\widetilde{%
h}_{t-p})^{\prime }$ at a specified quantile level $\tau$,
\begin{equation}
\widehat{\theta }_{\tau n}=\argmin_{\theta _{\tau }}\sum_{t=1}^{n}\frac{1}{%
\widetilde{h}_{t}}\rho _{\tau }(y_{t}-\theta _{\tau }^{\prime}\widetilde{z}%
_{t}).  \label{criterion}
\end{equation}
Thus, the $\tau$th conditional quantile of $y_{t}$ can be estimated by $%
\widehat{Q}_{\tau }(y_{t}|\mathcal{F}_{t-1})=\widehat{\theta }_{\tau
n}^{\prime}\widetilde{z}_{t}$.

\item \textit{Step E3 (Conditional quantile estimation for the original time
series).} Estimate the $\tau$th conditional quantile of $x_{t}$ by $\widehat{%
Q}_{\tau }(x_{t}|\mathcal{F}_{t-1})=T^{-1}(\widehat{\theta }_{\tau
n}^{\prime }\widetilde{z}_{t})$, where $T^{-1}(x)=\sqrt{|x|}\sgn(x)$ is the
inverse function of $T(\cdot)$.
\end{itemize}

\smallskip For convenience of the asymptotic analysis,
we make the following assumptions.

\begin{assum}
\label{assum1} (i) $\theta_0$ is in the interior of $\Theta$; (ii) $\eta_t^2$
has a non-degenerate distribution with $E\eta_t^2=1$; (iii) The polynomials $%
\sum_{i=1}^q\alpha_i x^i$ and $1-\sum_{j=1}^p\beta_j x^j$ have no common
root; (iv) $E\eta_t^{4}<\infty$.
\end{assum}

Assumption \ref{assum1} is used by \cite{Francq_Zakoian2004} to ensure the
consistency and asymptotic normality of the Gaussian QMLE $\widetilde{\theta}%
_n$, and is known as the sharpest result. It implies only a finite fractional moment of $x_t$, i.e., $E|x_t|^{2\delta_0}<\infty$ for some $\delta_0>0$ %
\citep{Berkes_Istvan_Horvath2003, Francq_Zakoian2004}. For the GARCH model,
imposing a higher-order moment condition on $x_t$ would reduce the available parameter space; see \citet[Chapter~2.4.1]{Francq_Zakoian2010}.

\begin{assum}
\label{assum2} The density $f(\cdot)$ of $\varepsilon_t=T(\eta_t)$ is
positive and differentiable almost everywhere on $\mathbb{R}$, with its
derivative $\dot{f}$ satisfying that $\sup_{x\in \mathbb{R}}|\dot{f}%
(x)|<\infty$.
\end{assum}

Assumption \ref{assum2} is made for  brevity of the technical proofs, while it is
sufficient to restrict the positiveness of $f(\cdot)$ and the boundedness of
$|\dot{f}(\cdot)|$ in a small and fixed interval $[b_{\tau}-r, b_{\tau}+r]$
for some $r>0$.

Let $\kappa_1=E[\eta_t^2I(\eta_t< Q_{\tau, \eta})]-\tau$ and $%
\kappa_2=E\eta_t^4-1$. Define the following $(p+q+1)\times(p+q+1)$ matrices:
\begin{equation*}
J=E\left[\frac{1}{h_t^2}\frac{\partial h_t(\theta_0)}{\partial\theta}\frac{%
\partial h_t(\theta_0)}{\partial\theta^\prime}\right],\quad
\Omega_0=E(z_tz_t^\prime),
\end{equation*}
and for $i=1$ and 2,
\begin{equation*}
\Omega_i=E\left (\frac{z_tz_t^\prime}{h_t^{i}}\right ), \quad H_i=E\left [%
\frac{z_t}{h_t^i}\frac{\partial h_t(\theta_0)}{\partial\theta^\prime}\right]%
, \quad\text{and}\quad \Gamma_i=E\left [\frac{z_t}{h_t^i}\sum_{j=1}^{p}%
\beta_{0j}\frac{\partial h_{t-j}(\theta_0)}{\partial\theta^\prime}\right].
\end{equation*}
The asymptotic distribution of the quantile regression estimator $\widehat{%
\theta }_{\tau n}$ is given as follows.

\begin{thm}
\label{thm1} Under Assumptions \ref{assum1} and \ref{assum2},
\begin{equation*}  \label{estimator}
\sqrt{n}(\widehat{\theta}_{\tau n} -{\theta}_{\tau 0}) \rightarrow_d
N(0,\Sigma_1),
\end{equation*}
where $\theta_{\tau 0}=b_{\tau}\theta_0$ and
\begin{equation*}
\Sigma_1=\Omega_2^{-1}\left [\frac{\tau-\tau^2}{f^2(b_{\tau})}\Omega_2+
\frac{\kappa_1 b_{\tau}}{f(b_{\tau})}(\Gamma_2
J^{-1}H_2^\prime+H_2J^{-1}\Gamma_2^\prime)+\kappa_2 b_{\tau}^2 \Gamma_2
J^{-1}\Gamma_2^\prime\right ]\Omega_2^{-1}.
\end{equation*}
\end{thm}

\smallskip We have used the weighted quantile regression %
\eqref{criterion} for the sake of efficiency, since $%
y_{t}-Q_{\tau}(y_{t}|\mathcal{F}_{t-1})=h_{t}(\varepsilon _{t}-b_{\tau })$.
Alternatively, the following unweighted quantile regression may be
considered in Step E2,
\begin{equation*}
\widecheck{\theta }_{\tau n}=\argmin_{\theta _{\tau }}\sum_{t=1}^{n}\rho
_{\tau }(y_{t}-\theta _{\tau }^{\prime }\widetilde{z}_{t});
\end{equation*}
see also \cite{Xiao_Koenker2009}. The following corollary provides the asymptotic distribution of the
unweighted quantile regression estimator $\widecheck{\theta }_{\tau n}$.

\begin{cor}
\label{cor1} If $E|x_t|^{4+\iota_0}<\infty$ for some $\iota_0>0$, and
Assumptions \ref{assum1} and \ref{assum2} hold, then
\begin{equation*}  \label{estimator_unw}
\sqrt{n}(\widecheck{\theta}_{\tau n} -{\theta}_{\tau 0}) \rightarrow_d
N(0,\Sigma_2),
\end{equation*}
where
\begin{equation*}
\Sigma_2=\Omega_1^{-1}\left [\frac{\tau-\tau^2}{f^2(b_{\tau})}\Omega_0+
\frac{\kappa_1 b_{\tau}}{f(b_{\tau})}(\Gamma_1
J^{-1}H_1^\prime+H_1J^{-1}\Gamma_1^\prime)+\kappa_2 b_{\tau}^2 \Gamma_1
J^{-1}\Gamma_1^\prime\right ]\Omega_1^{-1}.
\end{equation*}
\end{cor}

\smallskip In contrast to Theorem \ref{thm1}, Corollary \ref{cor1} requires $E|x_t|^{4+\iota_0}<\infty$
which entails a smaller available parameter space $\Theta$. Moreover, in the
ARCH case, the asymptotic covariance matrices $\Sigma_1$ and $\Sigma_2$
reduce to $(\tau-\tau^2)\Omega_2^{-1}/f^2(b_{\tau})$ and $%
(\tau-\tau^2)\Omega_1^{-1}\Omega_0\Omega_1^{-1}/f^2(b_{\tau})$, respectively, where it can be verified that $\Sigma_2-\Sigma_1$ is nonnegative definite, i.e., $%
\widehat{\theta}_{\tau n}$ is asymptotically more efficient than $%
\widecheck{\theta}_{\tau n}$. For the GARCH case, a theoretical comparison
becomes much more difficult, but our Monte Carlo evidence in
Section 5.2 demonstrates that the weighted estimator $%
\widehat{\theta }_{\tau n}$ is generally superior in finite samples. For
this reason, we focus on the weighted estimator $\widehat{\theta }_{\tau n}$
in our later discussions.

The asymptotic result for the $\tau$th
conditional quantile estimator of $y_{n+1}$ is given in the next corollary.

\begin{cor}
\label{cor2} Under the conditions in Theorem \ref{thm1}, it holds that
\begin{equation*}
\widehat{Q}_{\tau}(y_{n+1}|\mathcal{F}_{n})-Q_{\tau}(y_{n+1}|\mathcal{F}%
_{n})=u_{n+1}^{\prime }(\widetilde{\theta }_{n}-\theta
_{0})+z_{n+1}^{\prime}(\widehat{\theta }_{\tau n}-\theta _{\tau
0})+o_{p}(n^{-1/2}),
\end{equation*}
where $u_{n+1}=b_{\tau }\sum_{j=1}^{p}\beta _{0j}{\partial
h_{n+1-j}(\theta_{0})}/{\partial \theta }$.
\end{cor}

When $b_{\tau }\neq 0$, we further have the result for the $\tau $%
th conditional quantile estimator of $x_{n+1}$ as follows,
\begin{equation}
\widehat{Q}_{\tau }(x_{n+1}|\mathcal{F}_{n})-Q_{\tau }(x_{n+1}|\mathcal{F}%
_{n})=\frac{u_{n+1}^{\prime }(\widetilde{\theta }_{n}-\theta
_{0})+z_{n+1}^{\prime }(\widehat{\theta }_{\tau n}-\theta _{\tau 0})}{2\sqrt{%
|b_{\tau }h_{n+1}|}}+o_{p}(n^{-1/2}).  \label{quantx}
\end{equation}

\section{A Mixed Bootstrapping Procedure}

The asymptotic results in Section 2 facilitate statistical
inference based on the conditional quantile estimation. However, the limiting covariance matrix $\Sigma _{1}$ in Theorem \ref{thm1} depends on the sparsity
function $1/f(b_{\tau })$, whose estimation  is
complicated and sensitive to additional tuning parameters. In this section, we
propose a mixed bootstrapping procedure for approximating the asymptotic
distribution of $\widehat{\theta }_{\tau n}$, and
further construct confidence intervals for the conditional quantiles.

We first consider the random-weighting bootstrap method. Notice that the
QMLE $\widetilde{\theta }_{n}$ contributes to the asymptotic distribution of
the quantile regression estimator $\widehat{\theta }_{\tau n}$
in the way that
\begin{equation*}
\sqrt{n}(\widehat{\theta }_{\tau n}-\theta _{\tau 0})=\frac{\Omega _{2}^{-1}%
}{f(b_{\tau })}T_{1n}-b_{\tau }\Omega _{2}^{-1}\Gamma _{2}\sqrt{n}(%
\widetilde{\theta }_{n}-\theta _{0})+o_{p}(1),
\end{equation*}%
where $T_{1n}=n^{-1/2}\sum_{t=1}^{n}\psi _{\tau }(\varepsilon _{t}-b_{\tau })%
{z_{t}}/{h_{t}}$, as implied by the proof of Theorem \ref{thm1}. This suggests that the random-weighting bootstrap
needs to be employed for both $\widetilde{\theta }_{n}$ and $\widehat{\theta
}_{\tau n}$, and hence leads to the following bootstrapping procedure:
\begin{itemize}
\item \textit{Step B1.} In parallel with Step E1, perform the randomly
weighted QMLE,
\begin{equation}  \label{wQMLE}
\widetilde{\theta}_n^*=\argmin_{\theta\in\Theta}\sum_{t=1}^{n}\omega_t%
\widetilde{\ell}_t(\theta),
\end{equation}
where $\{\omega_t\}$ are $i.i.d.$ non-negative random weights with mean and
variance both equal to one, and then compute the initial estimates of $\{h_t\}$
as $\widetilde{h}_t^*=\widetilde{h}_t(\widetilde{\theta}_n^*)$.

\item \textit{Step B2.} Resembling Step E2, perform the randomly weighted
quantile regression,
\begin{equation}  \label{wcriterion}
\widehat{\theta}_{\tau n}^* =\argmin_{\theta_{\tau}}\sum_{t=1}^n\frac{%
\omega_t}{\widetilde{h}_t}\rho_{\tau}(y_t-\theta_{\tau}^{\prime}\widetilde{z}%
_t^*),
\end{equation}
where $\widetilde{z}_t^*=(1,x_{t-1}^2,\dots,x_{t-q}^2, \widetilde{h}%
_{t-1}^*,\dots,\widetilde{h}_{t-p}^*)^{\prime}$.

\item \textit{Step B3.} Analogous to Step E3, calculate the $\tau$th
conditional quantile estimate $\widehat{Q}_{\tau }^*(x_{t}|\mathcal{F}%
_{t-1})=T^{-1}(\widehat{\theta }_{\tau n}^{*\prime }\widetilde{z}_{t}^*)$.
\end{itemize}

\smallskip As a result,  the distribution of $(\widehat{%
\theta }_{\tau n}-\theta _{\tau 0})$ can be approximated by that of $(\widehat{\theta }_{\tau
n}^{\ast }-\widehat{\theta }_{\tau n})$. However, the numerical optimization \eqref{wQMLE} is in fact unnecessary, and can be time-consuming given the large
number of bootstrap replications. Instead of adopting the above procedure,
we next consider a mixed bootstrap method.

The randomly weighted QMLE $\widetilde{\theta}_n^*$ in \eqref{wQMLE}  is calculated for the purpose of approximating the
asymptotic distribution of $\widetilde{\theta}_n$. This is because it can be verified
that
\begin{equation*}
\sqrt{n}(\widetilde{\theta}_n^*-\widetilde{\theta}_n)=-\frac{J^{-1}}{\sqrt{n}%
}\sum_{t=1}^{n}(\omega_t-1)\left (1-\frac{|y_t|}{h_t}\right )\frac{1}{h_t}%
\frac{\partial h_t(\theta_0)}{\partial\theta}+o_p^*(1),
\end{equation*}
which is comparable to the result from \cite{Francq_Zakoian2004} that
\begin{equation*}
\sqrt{n}(\widetilde{\theta}_{n}-\theta_{0})=-\frac{J^{-1}}{\sqrt{n}}%
\sum_{t=1}^{n}\left (1-\frac{|y_t|}{h_t}\right )\frac{1}{h_t}\frac{\partial
h_t(\theta_0)}{\partial\theta}+o_p(1).
\end{equation*}
Note that the density $f(\cdot)$ is not involved in the above representations.
On the other hand, the matrix $J=E\{h_t^{-2}[\partial
h_t(\theta_0)/\partial\theta][\partial
h_t(\theta_0)/\partial\theta^\prime]\} $ can be estimated consistently by $%
\widetilde{J}=n^{-1}\sum_{t=1}^{n}\widetilde{h}_t^{-2}[\partial \widetilde{h}%
_t(\widetilde{\theta}_n)/\partial\theta] [\partial\widetilde{h}_t(\widetilde{%
\theta}_n)/\partial\theta^\prime]$. These indicate that the minimization \eqref{wQMLE} in
Step B1 can be simply replaced by a sample averaging:

\begin{itemize}
\item \textit{Step B1$^\prime$.} Calculate the estimator $\widetilde{\theta}%
_n^*$ by
\begin{equation}  \label{boots}
\widetilde{\theta}_n^*=\widetilde{\theta}_n-\frac{\widetilde{J}^{-1}}{n}%
\sum_{t=1}^{n}(\omega_t-1)\left(1-\frac{|y_t|}{\widetilde{h}_t}\right)\frac{1%
}{\widetilde{h}_t}\frac{\partial\widetilde{h}_t(\widetilde{\theta}_n)}{%
\partial\theta}.
\end{equation}
\end{itemize}

Combining Steps B1$^{\prime }$, B2
and B3, we propose a mixed bootstrapping procedure.  Its theoretical
justification is provided as follows.

\begin{assum}
\label{assum3} The random weights $\{\omega_t\}$ are $i.i.d.$ non-negative
random variables with mean and variance both equal to one, satisfying $%
E|\omega_t|^{2+\kappa_0}<\infty$ for some $\kappa_0>0$.
\end{assum}

\begin{thm}
\label{thm3} Suppose that $E|\eta_t|^{4+2\nu_0}<\infty$ for some $\nu_0>0$
and Assumptions \ref{assum1}-\ref{assum3} hold. Then, conditional on $%
\mathcal{F}_{n}$, $\sqrt{n}(\widehat{\theta}_{\tau n}^{*}-\widehat{\theta}%
_{\tau n}) \rightarrow_d N(0,\Sigma_1)$ in probability as $%
n\rightarrow\infty $, where $\Sigma_1$ is defined as in Theorem \ref{thm1}.
\end{thm}

\begin{cor}
\label{cor3} Under the conditions of Theorem \ref{thm3}, it holds that
\begin{equation*}
\widehat{Q}_{\tau }^{\ast}(y_{n+1}|\mathcal{F}_{n})-\widehat{Q}%
_{\tau}(y_{n+1}|\mathcal{F}_{n})=u_{n+1}^{\prime }(\widetilde{\theta }%
_{n}^{\ast }-\widetilde{\theta }_{n})+z_{n+1}^{\prime }(\widehat{\theta }%
_{\tau n}^{\ast}-\widehat{\theta }_{\tau n})+o_{p}^{\ast }(n^{-1/2}),
\end{equation*}
where $u_{n+1}$ is defined as in Corollary \ref{cor2}.
\end{cor}

This, together with Corollary \ref{cor2} and the asymptotic results for $%
\widetilde{\theta }_{n}^{\ast }$ and $\widehat{\theta }_{\tau n}^{\ast}$ in
the proof of Theorem \ref{thm3}, indicates that confidence
intervals for the conditional quantile $Q_{\tau }(y_{n+1}|\mathcal{F}_{n})$ can be constructed using the bootstrap
sample $\{\widehat{Q}_{\tau }^{\ast }(y_{n+1}|\mathcal{F}_{n})\}$.
As a consequence, applying the monotonicity of $T^{-1}(\cdot )$, the corresponding
confidence intervals for $Q_{\tau}(x_{n+1}|\mathcal{F}_{n})=T^{-1}[Q_{\tau
}(y_{n+1}|\mathcal{F}_{n})]$ can be constructed based on the empirical
quantiles of $\widehat{Q}_{\tau }^{\ast }(x_{n+1}|\mathcal{F}%
_{n})=T^{-1}[\widehat{Q}_{\tau }^{\ast }(y_{n+1}|\mathcal{F}_{n})]$, irrespective of
the value of $b_{\tau}$.

We summarize the proposed bootstrapping procedure as follows.
\begin{itemize}
\item \textit{Step 1.} Generate \emph{i.i.d.} random weights $\{\omega_t\}$
from a non-negative distribution with mean and variance both equal to one.

\item \textit{Step 2.} Calculate $\widetilde{\theta}_{n}^*$ by \eqref{boots}, and subsequently perform the randomly weighted linear quantile regression in %
\eqref{wcriterion} to obtain $\widehat{\theta}_{\tau n}^*$.

\item \textit{Step 3.} Calculate $E^{(1)}=\sqrt{n}(\widehat{\theta}_{\tau
n}^{*}-\widehat{\theta}_{\tau n})$ and $Q^{(1)}=T^{-1}(\widehat{\theta}%
_{\tau n}^{*\prime}\widetilde{z}_{n+1}^{*})$. Then, repeat Steps 1-2 for $B-1$
times to obtain $\{E^{(1)}, \dots, E^{(B)}\}$ and $\{Q^{(1)}, \dots,
Q^{(B)}\}$.  The empirical distribution of $\{E^{(i)}\}_{i=1}^B$ can be
used to approximate the asymptotic distribution of $\sqrt{n}( \widehat{\theta%
}_{\tau n}-\theta_{\tau 0})$, and the empirical quantiles of $%
\{Q^{(i)}\}_{i=1}^B$ can be used to construct confidence intervals for $%
Q_{\tau}(x_{n+1}|\mathcal{F}_{n})$.
\end{itemize}

\section{Diagnostic Checking for Conditional Quantiles}

To further illustrate the potential applicability of the results in previous sections, we consider diagnostic checking for the fitted conditional
quantiles. To construct this test, we first introduce the
following weighted residuals:
\begin{equation}
\varepsilon _{t,\tau }=\frac{y_{t}-Q_{\tau }(y_{t}|\mathcal{F}_{t-1})}{h_{t}}%
=\varepsilon _{t}-b_{\tau },\quad t\in \mathbb{Z},  \label{residuals}
\end{equation}%
where $y_{t}=T(x_{t})$ and $\varepsilon _{t}=T(\eta _{t})$. If the conditional quantile $Q_{\tau }(x_{t}|\mathcal{F}_{t-1})$, and hence $Q_{\tau }(y_{t}|\mathcal{F}_{t-1})$, is correctly specified by (%
\ref{gquantile}) at quantile level $\tau$, then it follows that E$\left[ \psi
_{\tau }(\varepsilon _{t,\tau })|\mathcal{F}_{t-1}\right] =0$. Motivated by
the quantile autocorrelation function (QACF) in \cite{Li_Li_Tsai2015} and
the absolute residual ACF in \cite{Li_Li2005}, we define the QACF of $%
\{\varepsilon _{t,\tau }\}$ at lag $k$ as
\begin{equation*}
\rho _{k,\tau }=\text{qcor}_{\tau }\big\{\varepsilon _{t,\tau },|\varepsilon
_{t-k,\tau }|\big\}=\frac{E\big\{\psi _{\tau }(\varepsilon _{t,\tau
})|\varepsilon _{t-k,\tau }|\big\}}{\sqrt{(\tau -\tau ^{2})\sigma _{a,\tau
}^{2}}},\quad k=1,2,\dots ,
\end{equation*}%
where $\sigma _{a,\tau }^{2}=\var(|\varepsilon _{t,\tau }|)=E(|\varepsilon
_{t,\tau }|-\mu _{a,\tau })^{2}$, with $\mu _{a,\tau }=E|\varepsilon
_{t,\tau }|$. Thus, under the null hypothesis that $Q_{\tau }(y_{t}|\mathcal{F}_{t-1})$ is correctly
specified, it holds that $\rho _{k,\tau }=0$ for all $k$. We shall base our test on
this residual QACF.

For a given $\tau \in (0,1)$, let $\widehat{\theta }_{\tau n}$ be the quantile
regression estimate obtained in \eqref{criterion}, and $\widetilde{h}%
_{t}$ and $\widetilde{z}_{t}$ be the associated volatility and regressors used in the
quantile regression. We construct the following sample counterpart of the
weighted residuals in \eqref{residuals}:
\begin{equation*}
\widehat{\varepsilon }_{t,\tau }=\frac{y_{t}-\widehat{\theta }_{\tau
n}^{\prime }\widetilde{z}_{t}}{\widetilde{h}_{t}},\quad t=1,\dots ,n,
\end{equation*}%
Then, the corresponding residual QACF at lag $k$ is
\begin{equation*}
r_{k,\tau }=\frac{1}{\sqrt{(\tau -\tau ^{2})\widehat{\sigma }_{a,\tau }^{2}}}%
\cdot \frac{1}{n}\sum_{t=k+1}^{n}\psi _{\tau }(\widehat{\varepsilon }%
_{t,\tau })|\widehat{\varepsilon }_{t-k,\tau }|,
\end{equation*}%
where $\widehat{\sigma }_{a,\tau }^{2}=n^{-1}\sum_{t=1}^{n}(|\widehat{
\varepsilon }_{t,\tau }|-\widehat{\mu }_{a,\tau })^{2}$, with $\widehat{\mu }%
_{a,\tau }=n^{-1}\sum_{t=1}^{n}|\widehat{\varepsilon }_{t,\tau }|$.

Let $K$ be a predetermined positive integer, and denote $R=(r_{1,\tau
},\dots ,r_{K,\tau })^{\prime }$.
Under the null hypothesis, $R$ will be close to zero (in the sense that it
is a zero-mean random vector). If the null hypothesis is false, $R$ will
deviate from zero. A test statistic can be constructed upon appropriate
standardizations and transformations on $R$.

We first derive the asymptotic distribution of $R=(r_{1,\tau
},\dots,r_{K,\tau })^{\prime }$, which provides guidance for the
construction of the test. Let $\epsilon_{t}=(|\varepsilon_{t,\tau}|,
|\varepsilon_{t-1,\tau}|, \dots,|\varepsilon_{t-K+1,\tau}|)^\prime$ and $%
\Xi=E(\epsilon_{t}\epsilon_{t}^\prime)$. Define the following $K\times
(p+q+1)$ matrices:
\begin{equation*}
D_1=E\left (\frac{\epsilon_{t-1}z_t^\prime}{h_t}\right ),\quad D_2=E\left [%
\frac{\epsilon_{t-1}}{h_t}\sum_{j=1}^{p}\beta_{0j}\frac{\partial
h_{t-j}(\theta_0)}{\partial\theta^\prime}\right ], \quad\text{and}\quad
D_3=E\left [\frac{\epsilon_{t-1}}{h_t}\frac{\partial h_{t}(\theta_0)}{%
\partial\theta^\prime}\right ].
\end{equation*}
For simplicity, denote $P=D_2-D_1\Omega_2^{-1}\Gamma_2$, $%
Q=D_3-D_1\Omega_2^{-1}H_2$, and $\Omega_3=D_1\Omega_2^{-1}D_1^\prime$.

\begin{thm}
\label{thm2} If $E|\eta_t|^{4+2\nu_0}<\infty$ for some $\nu_0>0$ and
Assumptions \ref{assum1} and \ref{assum2} hold, then
\begin{equation*}
\sqrt{n}R\rightarrow_d N(0,\Sigma_3),
\end{equation*}
where
\begin{equation*}
\Sigma_3= \frac{\Xi-\Omega_3}{\sigma_{a,\tau}^2}+ \frac{\kappa_1
b_{\tau}f(b_{\tau})}{(\tau-\tau^2)\sigma_{a,\tau}^2}(QJ^{-1}P^%
\prime+PJ^{-1}Q^\prime)+\frac{\kappa_2 b_{\tau}^2f^2(b_{\tau})}{%
(\tau-\tau^2)\sigma_{a,\tau}^2}PJ^{-1}P^\prime.
\end{equation*}
\end{thm}

Suppose that $\widehat{\Sigma }_{3} $ is a consistent estimator of $\Sigma
_{3}$. Then Theorem \ref{thm2} implies that the following test statistic,
\begin{equation}  \label{test}
Q(K)=nR^{\prime}\widehat{\Sigma}_3^{-1}R,
\end{equation}
converges to the $\chi ^{2}$ distribution with $K$ degrees of freedom as $%
n\rightarrow \infty$. Nevertheless, in practice, it is difficult to estimate the
asymptotic covariance matrix $\Sigma_{3}$ as it involves $f(b_{\tau })$.
We propose a bootstrap method in a similar way to the previous section.

Let $\widehat{\varepsilon}_{t,\tau}^*=(y_t-\widehat{\theta}_{\tau
n}^{*\prime}\widetilde{z}_t^*)/\widetilde{h}_t$. To approximate the
asymptotic distribution of $R$ in Theorem \ref{thm2}, we calculate the
randomly weighted residual QACF at lag $k$ by
\begin{equation}  \label{r*}
r_{k,\tau}^*=\frac{1}{\sqrt{(\tau-\tau^2)\widehat{\sigma}_{a,\tau}^{2}}}\cdot%
\frac{1}{n}\sum_{t=k+1}^{n}\omega_t\psi_{\tau}(\widehat{\varepsilon}%
_{t,\tau}^*)|\widehat{\varepsilon}_{t-k,\tau}^*|.
\end{equation}
Let $R^*=(r_{1, \tau}^*, \dots, r_{K, \tau}^*)^\prime$. The bootstrapping
test follows from the next theorem.

\begin{thm}
\label{thm4} Suppose that the conditions in Theorem \ref{thm3} hold. Then,
conditional on $\mathcal{F}_{n}$, $\sqrt{n}(R^*-R) \rightarrow_d
N(0,\Sigma_3)$ in probability as $n\rightarrow\infty$, where $\Sigma_3$ is
defined as in Theorem \ref{thm2}.
\end{thm}

The bootstrapping test can be incorporated into the bootstrapping procedure proposed in Section 3. Specifically, in Step 3 of the procedure summarized therein, we can further calculate the vector $R^*$ by \eqref{r*} and subsequently obtain $T^{(1)}=\sqrt{n}(R^*-R)$. Then, by repeating Steps
1-2 for $B-1$ times, we can obtain $\{T^{(1)}, \dots, T^{(B)}\}$. As a result, we can approximate $\Sigma_3$ by the sample covariance matrix $\Sigma_3^*$ of $\{T^{(i)}\}_{i=1}^B$, and calculate the bootstrapping test
statistic $Q(K)$ accordingly.

If the value of $Q(K)$ exceeds the 95th theoretical percentile of $\chi_{K}^2$, then the null hypothesis that $r_{k,\tau}$ with $1\leq k\leq K$ are jointly insignificant is rejected.  We can also examine the
significance of the individual $r_{k,\tau}$'s, by checking if $\sqrt{n}r_{k,\tau}$ lies between the 2.5th and 97.5th empirical percentiles
of $\{T^{(i)}_k\}_{i=1}^B$, where $T^{(i)}_k$ is the $k$th element of $%
T^{(i)}$.

\section{Simulation Studies}

\subsection{Comparison with Existing Conditional Quantile Estimators}

We conduct Monte Carlo simulations to evaluate the finite-sample performance of the proposed estimation method and inference tools. This subsection focuses on the forecasting performance in comparison with existing condition quantile estimation methods for time series. The data $\{x_t\}_{t=1}^n$ are generated from the GARCH($1,1$) model,
\begin{equation}  \label{garch11}
x_t= \sqrt{h_t} \eta_t, \hspace{5mm} h_t=\alpha_0+\alpha_1x_{t-1}^2+%
\beta_1h_{t-1},
\end{equation}
where $\{\eta_t\}$ are $i.i.d.$ standard normal or standardized Student's $%
t_5$ with variance one. Two sets of parameters are
considered: $(\alpha_0, \alpha_1, \beta_1)=(0.1,0.8,0.15)$ (Model 1) and $%
(\alpha_0, \alpha_1, \beta_1)=(0.1,0.15,0.8)$ (Model 2). Note that Model 1 has
larger volatility, whereas the effect of shocks on
the volatility is more persistent in Model 2. We estimate the conditional quantiles at $\tau=0.05$ using various methods. The estimates of $Q_{\tau}(x_t|\mathcal{F}%
_{t-1})$ with $1\leq t\leq n$ are called the in-sample forecasts, while that of $%
Q_{\tau}(x_{n+1}|\mathcal{F}_{n})$ the out-of-sample forecast.

Particularly, as an alternative to Step E1,
we can adapt \citeauthor{Xiao_Koenker2009}'s \citeyearpar{Xiao_Koenker2009}
method to estimate the conditional variances $\{h_{t}\}$ by a sieve
approximation:
\begin{equation*}
h_{t}=\rho _{0}+\sum_{j=1}^{m}\rho _{j}x_{t-j}^{2}.
\end{equation*}%
Subsequently, estimates of $Q_{\tau }(x_{t}|\mathcal{F}_{t-1})$ can be obtained by
applying the transformation $T^{-1}(\cdot )$ to those of $Q_{\tau }(y_{t}|%
\mathcal{F}_{t-1})$, where $y_t=T(x_t)$; this method is denoted as $\text{%
QGARCH}_{1} $ and $\text{QGARCH}_{2}$ below. Following \cite%
{Xiao_Koenker2009}, we set the order of the sieve approximation to $%
m=3n^{1/4}$. In summary, we compare the following five methods:

\begin{itemize}
\item[a.] Hybrid: The hybrid conditional quantile estimator proposed in
Section 2, with weighted quantile regression in Step E2.

\item[b.] $\text{QGARCH}_{1}$: Estimation based on a sieve approximation at
the specific quantile level $\tau $ for the initial estimation of $\{h_{t}\}$, similar to the \textquotedblleft QGARCH1\textquotedblright\ method
in \cite{Xiao_Koenker2009}.

\item[c.] $\text{QGARCH}_{2}$: Estimation based on a sieve approximation
over multiple quantile levels $\tau _{i}=i/20$ for $i=1,2,\ldots ,19$, which are
combined via the minimum distance estimation, for the initial estimation of $%
\{h_{t}\}$, similar to the \textquotedblleft QGARCH2\textquotedblright\
method in \cite{Xiao_Koenker2009}.

\item[d.] CAViaR: The indirect GARCH($1,1$)-based CAViaR method in \cite%
{Engle_Manganelli2004}, using the Matlab code of grid-seaching from these
authors and the same settings of initial values for the optimization as in
their paper.

\item[e.] RiskM: The conventional RiskMetrics method, which assumes that the
data follow the Integrated GARCH(1,1) model: $x_{t}=\sqrt{h_{t}}\eta _{t}$, $%
h_{t}=0.06x_{t-1}^{2}+0.94h_{t-1}$, where $\{\eta _{t}\}$ are $i.i.d.$
standard normal; see, e.g, \cite{Morgan1996} and \cite{Tsay2010}.
\end{itemize}

We examine the in-sample and out-of-sample performance separately. Three
sample sizes, $n=200$, 500 and 1000, are considered, and 1000 replications
are generated for each sample size. For each setting, we compute the biases
and mean squared errors (MSEs) of the estimates by averaging individual
values over all time points and all samples. The results for Models 1 and 2
are reported in Tables \ref{table1a} and \ref{table1b} respectively.

Four findings from the tables are summarized as follows. Firstly, smaller
in-sample biases (or MSEs) are usually associated with smaller out-of-sample
biases (or MSEs). Secondly, the method with the smallest
absolute value of bias is the hybrid estimator when the innovations are
Gaussian, yet is the CAViaR estimator in the Student's $t_{5}$
cases. Interestingly, the in-sample bias (with sign) for cases with Student's $t_{5}$ distributed innovations is generally smaller than the corresponding number for the Gaussian cases,
possibly due to their heavy tails. Thirdly, for the MSE, the hybrid estimator
is clearly the best among all methods for Model 1, whereas CAViaR seems the
most favorable method for Model 2; however, these two methods are comparable as $n$ increases to 1000. Note that the
indirect CAViaR estimator of \cite{Engle_Manganelli2004} is essentially the
quantile regression estimator for GARCH models in \cite{Lee_Noh2013}.
Compared with CAViaR, the hybrid estimator relies on an initial estimation
that reduces efficiency, yet uses weights to improve efficiency. As
a result of these two effects, the efficiency gains from the weights will be
more pronounced when there are larger variations in the conditional variances $\{h_{t}\}$, namely the case of Model 1.

Lastly, for all methods except RiskM, the absolute value of in-sample bias
and in-sample MSE generally decrease as $n$ increases, while the
out-of-sample performance can be less stable. For the hybrid and CAViaR
estimators, the out-of-sample bias and MSE mostly decrease as $n$ increases,
with only a few exceptions in Student's $t_5$ cases. Nevertheless, the out-of-sample performance of the
QGARCH estimators can become very bad as $n$ increases to 1000 for both models with heavy-tailed innovations; e.g., the
out-of-sample MSEs of the QGARCH estimators can be as large as 10, as shown
in Table \ref{table1a}. This is due to only a few replications where the initial
estimates of $\{h_t\}$ are rather poor: the sieve approximation uses
unnecessarily large order $m$ that introduces too much noise. Note that $m$
increases with $n$, while smaller $\alpha_1$ and $\beta_1$ favor smaller $m$%
. Since the magnitude of $\beta_1$ has a greater impact on the choice of $m$
than $\alpha_1$, the problem is more severe in Model 1.

Overall, for the models we considered, the hybrid and grid-searching based
CAViaR estimators have the best forecasting performance, and they both outperform the sieve-based QGARCH estimators, while the RiskM estimator is the worst in most cases and is
especially unsatisfactory for Model 1. Finally, it is worth noting
that the hybrid estimator takes much less computation time than the
CAViaR estimator. For instance, for our 1000 replications of Model 1 with
Gaussian innovations and sample size $n=$1000, the CAViaR estimator
takes 15.6 minutes, while the proposed hybrid estimator takes only 2.8
minutes.

\subsection{Finite-Sample Performance of the Proposed Inference Tools}

This subsection consists of  three simulation experiments for evaluating the finite-sample performance of the proposed inference tools.

The first experiment compares the efficiency of the weighted quantile
regression estimator $\widehat{\theta}_{\tau n}$ and its unweighted
counterpart $\widecheck{\theta}_{\tau n}$. We generate the data  from the GARCH($1,1$) model in \eqref{garch11} with standard normal or standardized Student's $t_5$ distributed innovations $\{\eta_t\}$, using two sets of
parameters: $(\alpha_0, \alpha_1, \beta_1)=(0.4,0.2,0.2)$ (Model (a)) and $%
(\alpha_0, \alpha_1, \beta_1)=(0.4,0.2,0.6)$ (Model (b)). These settings ensure strict stationarity of $\{x_t\}$ with $%
E|x_t|^{4+\iota_0}<\infty$ for some $\iota_0>0$, as required in Corollary %
\ref{cor1}; see \cite{Ling_McAleer2002} for the existence of moments of
GARCH models. Particularly, the available parameter space of $\alpha_1$ is
severely restricted in the Student's $t_5$ case. The sample size is $n=2000$%
, and two quantile levels, $\tau=0.1$ and 0.25, are considered. Figure \ref%
{fig3} provides the box plots for the two estimators based on 1000
replications. It shows that the interquartile range of the weighted
estimator $\widehat{\theta}_{\tau n}$ is smaller than that of the unweighted
$\widecheck{\theta}_{\tau n}$ under all settings; the latter also suffers
from more severe outliers. The efficiency gains from the weights seem larger
for the Student's $t_5$ cases. Moreover,  for the unweighted estimator $\widecheck{\theta}_{\tau n}$, the sample median slightly deviates from the true value $\theta_{\tau 0}$ especially when the innovations are Student's $t_5$
distributed.  The results suggest that the weighted estimator is more efficient in finite samples.

The second experiment examines the finite-sample performance of the estimator $\widehat{\theta}_{\tau n}$ further, as well as that of the residual QACF $r_{k,\tau}$ and the bootstrapping approximations. The data are generated from the GARCH($1,1$) model in \eqref{garch11} with $(\alpha_0, \alpha_1, \beta_1)=(0.4,0.4,0.4)$ and the same settings for the innovations. The sample sizes
are $n=500$, $1000$ and $2000$, with $1000$ replications generated for each
sample size. Three distributions for the random weights $\{\omega_t\} $ in the bootstrapping procedure are considered: the standard
exponential distribution ($W_1$); the Rademacher distribution ($W_2$), which
takes the values 0 or 2, each with probability 0.5 \citep{Li_Leng_Tsai2014};
and Mammen's two-point distribution ($W_3$), which takes the value $(-\sqrt{5%
}+3)/2$ with probability $(\sqrt{5}+1)/2\sqrt{5}$ and the value $(\sqrt{5}%
+3)/2$ with probability $1-(\sqrt{5}+1)/2\sqrt{5}$ \citep{Mammen1993}. As in the previous experiment, we examine two quantile levels, $\tau=0.1$ and 0.25.

The biases, empirical standard deviations (ESDs) and asymptotic standard
deviations (ASDs) for $\widehat{\theta}_{\tau n}$ are reported in Table \ref%
{table2}, and those for the residual QACF $r_{k,\tau}$ at lags $k=2, 4$ and 6
are given in Table \ref{table3}. All ASDs are computed according to the proposed
bootstrapping procedure. From both tables, we have the following results: (1) the biases are small, and the ESDs and
ASDs are fairly close to each other as $n$ increases to 1000; (2) as $n$
increases, the biases and standard deviations decrease, and the ESDs get
closer to the corresponding ASDs; (3) the choice of random weights have
little influence on the bootstrapping approximations; (4) the performance of the bootstrapping approximations are similar for both quantile levels as $n$ increases to 1000.

The third experiment evaluates the empirical size and power of the proposed
bootstrapping test statistic $Q(K)$. The data generating
processes are
\begin{equation*}  \label{sim2}
x_t=\sqrt{h_t}\eta_t, \hspace{5mm}
h_t=0.4+0.2x_{t-1}^2+dx_{t-4}^2+0.2h_{t-1},
\end{equation*}
where $\{\eta_t\}$ follow the same distribution as in the previous experiment, and we consider
departure $d=0$, 0.3 or 0.6. We conduct the estimation assuming a GARCH($1,1$) model; thus, $d=0$ corresponds to the size of
the test, and $d\neq 0$ corresponds to the power. All other settings are
preserved from the previous experiment.

Table \ref{table4} reports the rejection rates of $Q(K)$ at the maximum lag $K=6$.
It shows a satisfactory performance of both the
size and the power. The sizes are close to the nominal rate $5\%$ for $n$ as small as 500, and the powers increase as either $n$ or
the departure $d$ increases. Interestingly, the powers at $\tau=0.1$ and 0.25 are close when $\{\eta_t\}$ are Gaussian, yet differ
notably when they are Student's $t_5$ distributed. Moreover, it is worth noting that when the lower quantile $\tau$ is smaller, the actual departure in the quantile
regression, namely $|b_{\tau}d|$, increases, yet meanwhile the
density $f(b_{\tau})$ decreases, i.e., there are fewer data points
around $b_{\tau}$. Hence, the effect of $\tau$ on the power is
mixed, and the simulation result suggests that the overall impact of $\tau$ varies with the innovation distribution. Lastly, the different random weights
distributions perform similarly, as in the previous experiment.

\section{Empirical Analysis}

This section demonstrates the usefulness of the proposed approach through
analyzing daily log returns of three stock market indices: the S\&P 500
index, the Dow 30 index, and the Hang Seng Index (HSI). The data are from
January 2, 2008 to June 30, 2016.

To begin with, we illustrate the estimation of the conditional quantiles of the daily log returns of the S\&P 500 index. The sequence of log
returns, denoted by $\{x_t\}$, has a sample size of 2139, and its time plot is shown in Figure \ref{fig1}. For the time being, we focus on $\tau=0.05$, which corresponds to the $5\%$ VaR. Throughout this section, the standard
exponential random weights are used in the bootstrapping approximations.

We first consider an ARCH(1) model for $\{x_t\}$. The Gaussian
QMLE gives following estimation result,
\begin{equation}  \label{archreal}
\widetilde{h}_t=2.608\times10^{-7}_{6.222%
\times10^{-6}}+0.864_{5.159}x_{t-1}^2,
\end{equation}
where the standard errors of the parameter estimates are written as the
subscripts. Note that both the intercept and the ARCH coefficient are
insignificant. We compute the initial estimates of $\{h_t\}$ by \eqref{archreal}. Subsequently, using the proposed estimation method, we have that the fitted conditional quantile of $\{y_t\}$,
with $y_t=T(x_t)$, is
\begin{equation*}
\widehat{Q}_{0.05}(y_t|\mathcal{F}_{t-1})
=-2.717\times10^{-4}_{7.554\times10^{-5}}-1.659_{1.499}x_{t-1}^2,
\end{equation*}
where the coefficient of $x_{t-1}^2$ is insignificant, while the intercept
is significant. We next check
the adequacy of the fitted conditional quantile. The left panel of Figure %
\ref{fig2} shows that the individual residual QACF exceeds the 95\% confidence
bounds at, e.g., lags 1, 6 and 14, by a relatively large margin. Moreover, the $p$-value of the bootstrapping test statistic $Q(K)$ with the maximum lag $K=6$ is as small as 0.039,
and those of $Q(K)$ with $K=12$, 18, 24 and 30 are 0.228, 0.174, 0.127 and
0.088, respectively. The diagnostic tests indicate that there may be
higher-order ARCH effects which are not captured by the fitted model.

In view of this, we next consider a GARCH($1,1$) model and estimate the conditional quantiles using the proposed estimation procedure.  As a result, the initial estimates of $\{h_t\}$ are calculated from
\begin{equation}  \label{garchreal}
\widetilde{h}_t=2.646\times10^{-6}_{7.793%
\times10^{-7}}+0.126_{0.018}x_{t-1}^2+0.858_{0.019}\widetilde{h}_{t-1},
\end{equation}
and the fitted conditional quantile function is
\begin{equation}  \label{qgarchreal}
\widehat{Q}_{0.05}(y_t|\mathcal{F}_{t-1})= -4.713\times 10^{-7}_{3.199\times
10^{-5}}-0.124_{0.261}x_{t-1}^2-3.007_{0.521}\widetilde{h}_{t-1}.
\end{equation}
Interestingly, while all parameter estimates in \eqref{garchreal} are significant, only the coefficient of $\widetilde{h}_{t-1}$ is significant in \eqref{qgarchreal}.
Note that unlike the ARCH model, the conditional variances $\{h_t\}$ are defined recursively in  the GARCH model.  Therefore, the coefficient of $x_{t-1}^2$ in \eqref{garchreal} incorporates ARCH effects of all orders, and its significance confirms that ARCH
effects exist. That is, $h_t$ is not constant, even though the coefficient of $x_{t-1}^2$ in \eqref{archreal} is insignificant. On the
contrary, the coefficient of $x_{t-1}^2$ in \eqref{qgarchreal} contains only the effect of $x_{t-1}^2$ on the conditional quantile, as the higher order ARCH effects are already incorporated into the initial estimates $\widetilde{h}_{t-1}$. The insignificant coefficient of $x_{t-1}^2$ in \eqref{qgarchreal} suggests that $x_{t-1}^2$ itself may have no contribution to the conditional quantile at $\tau=0.05$.

We next conduct diagnostic checking for the fitted conditional quantiles given by \eqref{qgarchreal}.  As shown in the right panel of Figure \ref{fig2}, the residual QACF only slightly stands out of the 95\%
confidence bounds at lags 3, 21 and 24, and falls within them at all the
other lags. Furthermore, the $p$-values of $Q(K)$ are always larger than 0.257 for $K=6$, 12, 18, 24 and 30, indicating the adequacy of the
fitted conditional quantiles.

Moreover, to examine the forecasting performance of the proposed approach, we consider the rolling forecast of conditional quantiles at $%
\tau=0.05$, i.e., the negative 5\% VaR over a one-day horizon. We first
conduct the estimation for the first two years' data (January 2, 2008 to
December 31, 2009) under the GARCH($1,1$) model
assumption, and compute the conditional quantile forecast for the next
trading day, i.e., the forecast of $Q_{\tau}(x_{n+1}|\mathcal{F}_n)$. Then
we advance the forecasting origin by one, and, with one more observation
included in the estimation subsample, repeat the foregoing procedure until we
reach the end of the data set. For each rolling step, we use the proposed
bootstrap method to construct the 95\% confidence interval for the conditional quantile. The rolling forecasts and the corresponding confidence intervals are provided in Figure \ref{fig1}. It shows that $x_t$ falls below the conditional quantile forecast only occasionally, which supports the reliable performance of the proposed estimation and bootstrap method.

Finally, to better evaluate the forecasting performance of the proposed method, we conduct a more extensive analysis using the daily log returns of all the three stock market indices, and provide a comparison with the various conditional quantile estimation methods employed in the simulation experiment in Section 5.1.  We examine two quantile levels: $\tau=0.05$ and 0.01, which correspond to the one-day 5\% and 1\% VaR. The foregoing rolling procedure is adopted, and forecasting performance is measured by the empirical coverage rate (ECR), namely the proportion of observations that fall below the corresponding conditional quantile forecasts. The sample sizes for the Dow 30 index and the HSI are 2139 and 2130 respectively.

Table \ref{table7} reports the ECRs for the whole forecasting period as well
as those for four separate subperiods: (1) January 4, 2010 to December 30, 2011; (2) January 3, 2012 to
December 31, 2013; (3) January 2, 2014 to December 31, 2015; and (4) January 4, 2016
to June 30, 2016. For the ECRs of the whole period, the proposed
hybrid method always gives the ECR closest to the nominal rate among all methods, for both 1\% and 5\%
VaR and all market indices, except the case of 1\% VaR for the HSI. Although the results for the subperiods
are more mixed, it seems that the hybrid method is still most likely to give
the best ECR. Specifically, the number of times a method gives the best ECR
(including ties) among all methods for any subperiod and any quantile
level are: 16 for Hybrid, 4 for $\text{QGARCH}%
_1$, 7 for $\text{QGARCH}_2$, 6 for CAViaR, and 7 for RiskM. The forecasting
performance therefore corroborates the usefulness of the proposed method.

\section{Conclusion and Discussion}

Although quantile regression by nature is highly relevant to the conditional
quantile estimation for time series models with conditional
heteroscedasticity, it has been a very challenging task  for the most important member of this family, i.e., \citeauthor{Bollerslev1986}'s \citeyearpar{Bollerslev1986} GARCH model.
The major technical difficulties are due to the presence of latent variables and the square-root form of the conditional quantile function of this model.  In this paper, we propose an easy-to-implement quantile regression method for this important model. Our approach rests upon a monotone transformation which directly links the conditional quantile of the GARCH model to that of the linear GARCH model whose structure is much more tractable. As a result, quantile regression estimation for
the GARCH model is made easy. Meanwhile, the original GARCH form enables
reliable initial estimation for the conditional variances via the Gaussian
QMLE.

Inference about the conditional quantile is conducted, including
construction of confidence intervals and diagnostic checks.  To approximate the asymptotic distributions of the proposed estimator and statistics, we introduce a new bootstrap method, and through replacing an optimization step with a sample averaging, we speed up the bootstrapping procedure  significantly.  To sum up, a complete approach to the conditional quantile inference for the widely used GARCH model is provided in this paper.

The proposed approach can be extended in several directions. First, it is
well known that financial time series can be so heavy-tailed that $E(\eta
_{t}^{4})=\infty $ \citep{Hall2003, Mittnik2003, Mikosch2000}. For such
cases, we may alternatively consider methods that are more robust
than the Gaussian QMLE for the initial estimation of the conditional
variances, including the least absolute deviation estimator of \cite{Peng_Yao2003} and the rank-based
estimator of \cite{Andrews2012} among others.
Second, our approach can be applied to
the conditional quantile estimation for other conditional heteroscedastic
models, including the asymmetric GJR-GARCH model \citep{GJR_1993}. Third,
although the multivariate GARCH model has been widely used for the
volatility modeling of multiple asset returns \citep{BEKK_1995}, the
conditional quantile estimation for the corresponding portfolio returns is
still an open problem. The proposed hybrid conditional quantile inference
procedure offers some preliminary ideas on this, and we will leave it for
future research.

\renewcommand{\thesection}{A} 
\setcounter{equation}{0} 

\section*{Appendix: Technical Details}

This appendix gives the proofs of Theorems \ref{thm1}-\ref{thm4},
Corollaries \ref{cor1}-\ref{cor3} and Equation \eqref{quantx}. Lemmas \ref%
{lem2} and \ref{lem3} contain some preliminary results for GARCH models
which will be repeatedly used later.

Throughout the appendix, $C$ is a generic positive constant which may take
different values at its different occurrences, and $C(M)$ is such a constant
whose value depends on $M$. We denote by $\|\cdot\|$ the norm of a matrix or
column vector, defined as $\|A\|=\sqrt{\text{tr}(AA^\prime)}=\sqrt{%
\sum_{i,j}|a_{ij}|^2}$. In addition, let $z_t(\theta)=(1,x_{t-1}^2,%
\dots,x_{t-q}^2, h_{t-1}(\theta),\dots,h_{t-p}(\theta))^{\prime}$, $%
\widetilde{z}_t(\theta)=(1,x_{t-1}^2,\dots,x_{t-q}^2, \widetilde{h}%
_{t-1}(\theta),\dots, \widetilde{h}_{t-p}(\theta))^{\prime}$, and, for
simplicity, write $z_t=z_t(\theta_{0})$, $\breve{z}_t=\widetilde{z}%
_t(\theta_{0})$, and $\widetilde{z}_t=\widetilde{z}_t(\widetilde{\theta}%
_{n}) $, where $\widetilde{\theta}_{n}$ is the Gaussian QMLE of model %
\eqref{garch}. In the proofs of Theorems \ref{thm3} and \ref{thm4}, the
notations $E^*$, $O_p^*(1)$ and $o_{p}^*(1)$ correspond to the bootstrapped
probability space.

\begin{lemma}
\label{lem2} Under Assumption \ref{assum1}, for any $\kappa>0$, there exists
a constant $c>0$ such that
\begin{align*}
& (i)\; E\left(\sup\left\{ \frac{h_t(\theta_2)}{h_t(\theta_1)}:
\|\theta_1-\theta_2\|\leq c, \; \theta_1, \theta_2\in \Theta \right \}\right
)^{\kappa}<\infty, \\
& (ii)\; E\sup\left \{ \left\|\frac{1}{h_t(\theta_1)} \frac{\partial
h_t(\theta_2)}{\partial\theta} \right\|^{\kappa} : \|\theta_1-\theta_2\|\leq
c, \; \theta_1, \theta_2\in \Theta \right \} < \infty, \\
& (iii) \; E\sup\left \{ \left\|\frac{1}{h_t(\theta_1)} \frac{\partial^2
h_t(\theta_2)}{\partial \theta\partial \theta^\prime} \right\|^{\kappa} :
\|\theta_1-\theta_2\|\leq c, \;\theta_1, \theta_2\in \Theta \right \} <
\infty
\end{align*}
and
\begin{equation*}
(iv) \; E\sup \left \{ \left|\frac{1}{h_t(\theta_1)} \frac{\partial^3
h_t(\theta_2)}{\partial \theta_i\partial \theta_k\partial \theta_\ell}
\right|^{\kappa} : \|\theta_1-\theta_2\|\leq c, \;\theta_1, \theta_2\in
\Theta\right \} < \infty,
\end{equation*}
for all $1\leq i, k, \ell \leq p+q+1$.
\end{lemma}

\begin{proof}[Proof of Lemma \ref{lem2}]
	We first prove (i). For any $\theta=(\alpha_0, \alpha_1,\dots, \alpha_q, \beta_1, \dots, \beta_p)^{\prime}\in\Theta$ and $\gamma>1$, define
	\[ U(\gamma, \theta) = \{\theta^*=(\alpha_0^*, \alpha_1^*,\dots, \alpha_q^*, \beta_1^*, \dots, \beta_p^*)^{\prime}\in \Theta: \max_{1\leq j\leq p} \frac{\beta_j^*}{\beta_j} \leq \gamma\}. \]	
	Claim (i) follows from a more general result: for any $\kappa>0$, there is $ \gamma>1 $ such that
	\begin{equation}\label{lem1eq1}
	E\left [\sup_{\theta\in\Theta}\sup_{\theta^*\in U(\gamma, \theta)} \frac{h_t(\theta^*)}{h_t(\theta)}\right]^{\kappa} < \infty.
	\end{equation}
	Notice that for any $\theta$, the set $U(\gamma, \theta)$ only imposes an upper bound on the $\beta_j^*$'s, while the condition $\|\theta_1-\theta_2\|\leq c$
	restricts the distance between $\theta_1$ and $\theta_2$.
	
	We shall prove \eqref{lem1eq1}. Note that the functions $h_t(\theta)$, as defined recursively in \eqref{aeq1}, can be written in the form of
	\[h_t(\theta)=c_0(\theta)+\sum_{j=1}^{\infty} c_j(\theta)x_{t-j}^2,\]
	and the series converges with probability one for all $\theta\in\Theta$; see, e.g., \cite{Berkes_Istvan_Horvath2003}. Moreover, $c_0(\theta)=\alpha_0/(1-\beta_1-\cdots-\beta_p)\geq C_1=\underline{w}/(1-p\underline{w})>0$ for all $\theta\in\Theta$, and from Lemma 3.1 in \cite{Berkes_Istvan_Horvath2003}, it holds that
	\begin{equation}\label{lem1eq3}
	\sup_{\theta\in\Theta}c_j(\theta) \leq C_2 \rho_1^{j}, \hspace{5mm}j\geq 0,
	\end{equation}
	where $\rho_1=\rho_0^{1/p}\in (0,1)$,  and
	\begin{equation}\label{lem1eq4}
	\sup_{\theta\in\Theta} \sup_{\theta^*\in U(\gamma, \theta)}\frac{c_j(\theta^*)}{c_j(\theta)} \leq C_3 \gamma^j, \hspace{5mm}j\geq 0,
	\end{equation}
	for some constants $C_2, C_3>0$.  Using \eqref{lem1eq4}, we have
	\begin{equation*}\
	\sup_{\theta\in\Theta}\sup_{\theta^*\in U(\gamma, \theta)} \frac{h_t(\theta^*)}{h_t(\theta)}\leq \frac{C_2}{C_1}+C_3\sup_{\theta\in\Theta}\dfrac{\sum_{j=1}^{\infty}\gamma^jc_j(\theta)x_{t-j}^2}{C_1+\sum_{j=1}^{\infty}c_j(\theta)x_{t-j}^2},
	\end{equation*}
	and then it suffices to show that for any $\kappa\geq1$,
	\begin{equation*}
	 \left\|\sup_{\theta\in\Theta}\dfrac{\sum_{j=1}^{\infty}\gamma^jc_j(\theta)x_{t-j}^2}{C_1+\sum_{j=1}^{\infty}c_j(\theta)x_{t-j}^2} \right\|_{\kappa} <\infty,
	\end{equation*}
	where $\|\cdot\|_{\kappa}$ denotes the $L_{\kappa}$ norm, i.e., $\|X\|_{\kappa}=(E|X|^{\kappa})^{1/\kappa}$. Note that there is $\delta_0>0$ such that $E|x_0^2|^{\delta_0}<\infty$.  Thus,  for any $\kappa\geq 1$ and $\delta_1\in(1-\delta_0/\kappa, 1)$, by \eqref{lem1eq3} and the Minkowski inequality, we have
	\begin{align*}
	\left\| \sup_{\theta\in\Theta}\dfrac{\sum_{j=1}^{\infty}\gamma^jc_j(\theta)x_{t-j}^2}{C_1+\sum_{j=1}^{\infty}c_j(\theta)x_{t-j}^2}\right\|_{\kappa}&\leq \left\| \sup_{\theta\in\Theta}\sum_{j=1}^{\infty}\dfrac{\gamma^jc_j(\theta)x_{t-j}^2}{C_1^{1-{\delta_1}}[c_j(\theta)x_{t-j}^2]^{{\delta_1}}}\right\|_{\kappa}\\
	&\leq C_1^{-(1-{\delta_1})}\left \| \sum_{j=1}^{\infty}\gamma^j  (C_2\rho_1^jx_{t-j}^2)^{1-{\delta_1}}\right \|_{\kappa}\\
	& \leq C \sum_{j=1}^{\infty} (\gamma\rho_1^{1-{\delta_1}})^j\left [E|x_0^2|^{(1-{\delta_1})\kappa}\right ]^{1/\kappa}<\infty,
	\end{align*}
	if $\gamma$ is close enough to 1.  Therefore, \eqref{lem1eq1} holds, and so does (i).
	
	From the proof of Theorem 2.2 in \cite{Francq_Zakoian2004}, under Assumption \ref{assum1},  for any $\kappa>0$,
	\begin{equation*}
	E\sup_{\theta\in\Theta} \left\|\frac{1}{h_t(\theta)} \frac{\partial h_t(\theta)}{\partial\theta}  \right\|^{\kappa}  < \infty, \hspace{5mm}
	E\sup_{\theta\in\Theta} \left\|\frac{1}{h_t(\theta)} \frac{\partial^2 h_t(\theta)}{\partial \theta\partial \theta^{\prime}}  \right\|^{\kappa} < \infty \hspace{5mm}\text{and}\hspace{5mm}
	\end{equation*}
	\begin{equation*}
	E\sup_{\theta\in\Theta} \left|\frac{1}{h_t(\theta)} \frac{\partial^3 h_t(\theta)}{\partial \theta_{i}\partial \theta_{k}\partial \theta_{\ell}}  \right|^{\kappa} < \infty,
	\end{equation*}
	where $1\leq i, k, \ell \leq p+q+1$; see also Lemma 3.6 in \cite{Berkes_Horvath2004}. Combining these with (i), we immediately obtain (ii)-(iv).
\end{proof}

\begin{lemma}
\label{lem3} Under Assumption \ref{assum1},
\begin{equation*}
\sup_{\theta \in \Theta} |\widetilde{h}_t(\theta) - h_t(\theta)| \leq
C\rho^t \zeta \hspace{5mm}\text{and}\hspace{5mm} \sup_{\theta \in \Theta}
\left \|\frac{\partial\widetilde{h}_t(\theta)}{\partial\theta} - \frac{%
\partial h_t(\theta)}{\partial\theta}\right \| \leq C\rho^t \zeta,
\end{equation*}
where $C>0$ and $0<\rho<1$ are constants, and $\zeta$ is a random variable
independent of $t$ with $E|\zeta|^{\delta_0}<\infty$ for some $\delta_0>0$.
\end{lemma}

\begin{proof}[Proof of Lemma \ref{lem3}]
	The lemma can be proved by a method similar to that for Equations (6) and (7) in the proof of Theorem 1 in \cite{Zheng_Li_Li2016}.
\end{proof}

\begin{proof}[Proof of Theorem \ref{thm1}]	
	Let $L_n(\theta)=\sum_{t=1}^{n}\widetilde h_{t}^{-1}\rho_{\tau}(y_t-\theta^{\prime}\widetilde{z}_t)$ and
	$\breve{L}_n(\theta)=\sum_{t=1}^{n}\widetilde h_{t}^{-1}\rho_{\tau}(y_t-\theta^{\prime}\breve{z}_t)$.
	Notice that  for $x\neq 0$,
	\begin{equation}\label{identity}
		\rho_{\tau}(x-y)-\rho_{\tau}(x)=-y\psi_{\tau}(x)+\int_{0}^{y}[I(x\leq s)-I(x\leq 0)]ds,
	\end{equation}
	where $\psi_{\tau}(x)=\tau-I(x<0)$; see \cite{Knight1998}.
	Then, for any fixed $u\in \mathbb{R}^{p+q+1}$,
	\begin{equation}\label{L1_L2}
		L_n(\theta_{\tau 0}+n^{-1/2}u)-\breve{L}_n(\theta_{\tau 0})=-L_{1n}(u)+L_{2n}(u),
	\end{equation}
	where
	\begin{align*}
		L_{1n}(u)&= \sum_{t=1}^{n}\psi_{\tau}(\breve{e}_{t,\tau})\widetilde{h}^{-1}_t \left[(\theta_{\tau 0}+n^{-1/2}u)^{\prime}\widetilde{z}_t-\theta_{\tau 0}^{\prime}\breve{z}_t\right], \\
		L_{2n}(u)&=\sum_{t=1}^{n}\widetilde{h}^{-1}_t\int_{0}^{(\theta_{\tau 0}+n^{-1/2}u)^{\prime}\widetilde{z}_t-\theta_{\tau 0}^{\prime}\breve{z}_t}\left[I(\breve{e}_{t,\tau}\leq s)-I(\breve{e}_{t,\tau}\leq 0)\right]ds,
	\end{align*}
	and $\breve{e}_{t,\tau}= y_t-\theta_{\tau 0}^{\prime}\breve{z}_t$. Let  $u^{(j)}$ be the $(j+q+1)$-th element of $u$, and denote $\beta_{\tau0}^{(j)}=b_{\tau}\beta_{0j}$, for $j=1,\dots, p$. It can be verified that
	\begin{equation}\label{beq1}
		(\theta_{\tau 0}+n^{-1/2}u)^{\prime}\widetilde{z}_t-\theta_{\tau 0}^{\prime}\breve{z}_t =\xi_{1nt}(\widetilde{\theta}_n)+\xi_{2nt}(\widetilde{\theta}_n)+\xi_{3nt}(\widetilde{\theta}_n),
	\end{equation}
	where
	\begin{align*}
	\begin{split}
		\xi_{1nt}(\theta)&=n^{-1/2}u^\prime z_t+\sum_{j=1}^{p}\beta_{\tau0}^{(j)}\frac{\partial h_{t-j}(\theta_0)}{\partial\theta^\prime}(\theta-\theta_0),\\
		 \xi_{2nt}(\theta)&=\frac{1}{\sqrt{n}}\sum_{j=1}^{p}u^{(j)}[h_{t-j}(\theta)-h_{t-j}]+\sum_{j=1}^{p}\beta_{\tau0}^{(j)}\bigg[h_{t-j}(\theta)-h_{t-j}-\frac{\partial h_{t-j}(\theta_0)}{\partial\theta^\prime}(\theta-\theta_0)\bigg],\\
		 \xi_{3nt}(\theta) &=\frac{1}{\sqrt{n}}\sum_{j=1}^{p}u^{(j)}[\widetilde{h}_{t-j}(\theta)-h_{t-j}(\theta)] \\ &\hspace{5mm}+\sum_{j=1}^{p}\beta_{\tau0}^{(j)}\left \{[\widetilde{h}_{t-j}(\theta)-h_{t-j}(\theta)]-[\widetilde{h}_{t-j}(\theta_0)-h_{t-j}]\right \}.
		\end{split}
	\end{align*}
	For any $M>0$, denote $\Theta_n=\Theta_n(M)=\{\theta\in\Theta: \|\theta-\theta_0\|\leq n^{-1/2}M\}$.
	Using the Taylor expansion, it holds that
	\begin{equation}\label{beq11}
		\sup_{\theta\in\Theta_n}|\xi_{2nt}(\theta)|
		\leq \frac{M}{n}\sum_{j=1}^{p}|u^{(j)}| \sup_{\theta \in \Theta_n} \left \|\frac{\partial h_{t-j}(\theta)}{\partial\theta} \right\|+\frac{M^2}{2n}\sum_{j=1}^{p}
		|\beta_{\tau0}^{(j)}| \sup_{\theta \in \Theta_n} \left \|\frac{\partial^2h_{t-j}(\theta)}{\partial\theta\partial\theta^{\prime}} \right\|,
	\end{equation}
	and by Lemma \ref{lem3},
	\begin{align}
		\begin{split}\label{beq12}
			&\sup_{\theta \in \Theta_n}|\xi_{3nt}(\theta)|\\
			&\hspace{5mm}\leq \frac{1}{\sqrt{n}}\sum_{j=1}^{p}
			\bigg[|u^{(j)}|\sup_{\theta \in \Theta} |\widetilde{h}_{t-j}(\theta) - h_{t-j}(\theta)|+ M|\beta_{\tau0}^{(j)}| \sup_{\theta \in \Theta} \left \|\frac{\partial\widetilde{h}_{t-j}(\theta)}{\partial\theta} - \frac{\partial h_{t-j}(\theta)}{\partial\theta}\right \|\bigg]\\
			&\hspace{5mm}\leq n^{-1/2}C(M)\rho^t\zeta.
		\end{split}
	\end{align}
	 Moreover,
	\begin{equation}\label{beq5}
		\breve{e}_{t,\tau}=(\varepsilon_t-b_{\tau})h_t+a_t, \hspace{5mm}\text{where}\hspace{5mm}
		a_t=\sum_{j=1}^{p}\beta_{\tau0}^{(j)}[h_{t-j}-\widetilde{h}_{t-j}(\theta_0)] \in \mathcal{F}_0.
	\end{equation}
	
	We first consider $L_{1n}(u)$, which can be decomposed into four parts,
	\begin{equation}\label{L1a}
		 L_{1n}(u)=\sum_{t=1}^{n}A_{1nt}(\widetilde{\theta}_n)+\sum_{t=1}^{n}A_{2nt}(\widetilde{\theta}_n)+\sum_{t=1}^{n}A_{3nt}(\widetilde{\theta}_n)+\sum_{t=1}^{n}A_{4nt}(\widetilde{\theta}_n),
	\end{equation}
	where
\[A_{1nt}(\theta)=\psi_{\tau}(\breve{e}_{t,\tau})\widetilde{h}_t^{-1}(\theta)\xi_{3nt}(\theta)+ \psi_{\tau}(\breve{e}_{t,\tau}) [\widetilde{h}_t^{-1}(\theta)-h_t^{-1}(\theta)][\xi_{1nt}(\theta)+\xi_{2nt}(\theta)],\]
\[A_{2nt}(\theta)= [\psi_{\tau}(\breve{e}_{t,\tau})-\psi_{\tau}(\varepsilon_t-b_{\tau})] h_t^{-1}(\theta)[\xi_{1nt}(\theta)+\xi_{2nt}(\theta)],\]
\[A_{3nt}(\theta)=\psi_{\tau}(\varepsilon_t-b_{\tau})h_t^{-1}(\theta) \xi_{2nt}(\theta), \hspace{5mm}\text{and}\hspace{5mm} A_{4nt}(\theta)=\psi_{\tau}(\varepsilon_t-b_{\tau})h_t^{-1}(\theta)\xi_{1nt}(\theta).
\]
	Note that $\inf_{\theta\in\Theta}h_t(\theta)\geq \underline{w}$ and $\inf_{\theta\in\Theta}\widetilde{h}_t(\theta)\geq \underline{w}$. By Lemma \ref{lem3}, \eqref{beq11} and \eqref{beq12}, we can show that
	\begin{equation}\label{A1nt}
	\begin{split}
		\sup_{\theta\in\Theta_n}\left |\sum_{t=1}^{n}A_{1nt}(\theta)\right |&\leq \frac{1}{\underline{w}}\sum_{t=1}^{n}\sup_{\theta\in\Theta_n}|\xi_{3nt}(\theta)| +\frac{C\zeta}{\underline{w}^2}\sum_{t=1}^{n}\rho^t\sup_{\theta\in\Theta_n}(|\xi_{1nt}(\theta)| +|\xi_{2nt}(\theta)|)\\
		&=o_p(1),
	\end{split}
	\end{equation}
	which, together with the fact that $\sqrt{n}(\widetilde{\theta}_n-\theta_0)=O_p(1)$,  implies that
	\begin{equation}\label{A1}
		\sum_{t=1}^{n}A_{1nt}(\widetilde{\theta}_n)=o_p(1).
	\end{equation}
	
	Note that by Lemma \ref{lem3} and Assumption \ref{assum2}, we have
	\begin{equation*}
		|F(b_{\tau})-F(b_{\tau}-h_t^{-1}a_t)|\leq \sup_{x\in\mathbb{R}}f(x)\sum_{j=1}^{p}\frac{|\beta_{\tau0}^{(j)}|}{\underline{w}}|h_{t-j}(\theta_0)-\widetilde{h}_{t-j}(\theta_0) |\leq C\rho^t\zeta.
	\end{equation*}
	It then follows from \eqref{beq5} that
	\begin{align}\label{beq10}
		 E|\psi_{\tau}(\breve{e}_{t,\tau})-\psi_{\tau}(\varepsilon_t-b_{\tau})|&=E|F(b_{\tau})-F(b_{\tau}-h_t^{-1}a_t)|\notag\\
		&=E[|F(b_{\tau})-F(b_{\tau}-h_t^{-1}a_t)|I(C\rho^t\zeta\leq \rho^{t/2})]\notag\\
		&\hspace{5mm}+E[|F(b_{\tau})-F(b_{\tau}-h_t^{-1}a_t)|I(C\rho^t\zeta > \rho^{t/2})]\notag\\
		&\leq \rho^{t/2}+ \mathrm{Pr}(C\rho^t\zeta > \rho^{t/2})\leq \rho^{t/2}+C\rho^{{\delta_0} t/2},
	\end{align}
	where we used the Markov inequality and the fact that $E|\zeta|^{\delta_0}<\infty$.
	Moreover,
	\begin{equation}\label{beq2}
		\|h^{-1}_tz_t\|\leq\frac{\sqrt{p+q+1}}{\underline{w}},
	\end{equation}
	\begin{equation}\label{beq3}
		\sup_{\theta_1, \theta_2\in\Theta_n}\left|\frac{\xi_{1nt}(\theta_2)}{h_t(\theta_1)}\right|\leq\frac{| h_t^{-1}u^\prime z_t|}{\sqrt{n}}\sup_{\theta\in\Theta_n}\frac{h_t}{h_t(\theta)}+ \frac{M}{\underline{w}\sqrt{n}}\sum_{j=1}^{p}|\beta_{\tau0}^{(j)}|\sup_{\theta\in\Theta_n}\left \|\frac{1}{h_{t-j}(\theta)}\frac{\partial h_{t-j}(\theta_0)}{\partial\theta}\right \|,
	\end{equation}
	and by the Taylor expansion,
	\begin{equation}\label{beq13}
		\begin{split}
			\sup_{\theta_1, \theta_2\in\Theta_n}\left|\frac{\xi_{2nt}(\theta_2)}{h_t(\theta_1)}\right|&\leq
			\frac{M}{\underline{w}n}\sum_{j=1}^{p}|u^{(j)}| \sup_{\theta_1,\theta_2 \in \Theta_n} \left \|\frac{1}{h_{t-j}(\theta_1)}\frac{\partial h_{t-j}(\theta_2)}{\partial\theta} \right\|\\
			&\hspace{5mm}+\frac{M^2}{2\underline{w}n}\sum_{j=1}^{p}
			|\beta_{\tau0}^{(j)}| \sup_{\theta_1, \theta_2 \in \Theta_n} \left \|\frac{1}{h_{t-j}(\theta_1)}\frac{\partial^2h_{t-j}(\theta_2)}{\partial\theta\partial\theta^{\prime}} \right\|.
		\end{split}
	\end{equation}
	As a result, by the H\"{o}lder inequality, Lemma \ref{lem2} and \eqref{beq10}-\eqref{beq13}, we have
	\begin{equation*}
	\begin{split}
		E\sup_{\theta\in\Theta_n}\left |\sum_{t=1}^{n}A_{2nt}(\theta)\right |&\leq \sum_{t=1}^{n}\left[E|\psi_{\tau}(\breve{e}_{t,\tau})-\psi_{\tau}(\varepsilon_t-b_{\tau})|\right]^{1/2} \left[E\sup_{\theta\in\Theta_n}\left(\frac{|\xi_{1nt}(\theta)|+|\xi_{2nt}(\theta)|}{h_t(\theta)}\right)^2\right]^{1/2}\\&=o(1),
	\end{split}
	\end{equation*}
	which, together with the fact that $\sqrt{n}(\widetilde{\theta}_n-\theta_0)=O_p(1)$,  implies that
	\begin{equation}\label{A2}
		\sum_{t=1}^{n}A_{2nt}(\widetilde{\theta}_n)=o_p(1).
	\end{equation}
	
	Applying the Taylor expansion to $h^{-1}_t(\theta)$ and $\xi_{2nt}(\theta)$ respectively, we have
	\begin{equation}\label{beq15}
	h_t^{-1}(\theta)\xi_{2nt}(\theta)=\xi_{4nt}(\theta) + \xi_{5nt}(\theta),
	\end{equation}
	where
	\begin{align*}
		\xi_{4nt}(\theta) &= \frac{1}{\sqrt{n}}\sum_{j=1}^{p}\frac{u^{(j)}}{h_t}\frac{\partial h_{t-j}(\theta_0)}{\partial\theta^\prime}(\theta-\theta_0)+\dfrac{1}{2}(\theta-\theta_0)^\prime\sum_{j=1}^{p}\frac{\beta_{\tau0}^{(j)}}{h_t}\frac{\partial^2h_{t-j}(\theta_0)}{\partial\theta\partial\theta^\prime}(\theta-\theta_0),\\
		\xi_{5nt}(\theta)&=-\frac{\xi_{2nt}(\theta)}{h_t^2(\theta^*_1)}\frac{\partial h_t(\theta^*_1)}{\partial\theta^\prime}(\theta-\theta_0)+\frac{(\theta-\theta_0)^\prime}{2\sqrt{n}}\sum_{j=1}^{p}\frac{u^{(j)}}{h_t}\frac{\partial^2h_{t-j}(\theta^*_2)}{\partial\theta\partial\theta^\prime}(\theta-\theta_0)\\
		&\hspace{5mm}+\dfrac{1}{6}\sum_{j=1}^{p}\sum_{i, k, \ell=1}^{p+q+1}\frac{\beta_{\tau0}^{(j)}}{h_t}\frac{\partial^3h_{t-j}(\theta^*_2)}{\partial\theta_i\partial\theta_k\partial\theta_\ell}(\theta_i-\theta_{0i})(\theta_k-\theta_{0k})(\theta_\ell-\theta_{0\ell}),
	\end{align*}
	with $\theta_1^*$ and $\theta_2^*$ both between $\theta$ and $\theta_0$.
	Then, it follows from Lemma \ref{lem2}, the ergodic theorem and $\sqrt{n}(\widetilde{\theta}_n-\theta_0)=O_p(1)$ that
\begin{equation}\label{beq4}
\sum_{t=1}^{n}\psi_{\tau}(\varepsilon_t-b_{\tau})\xi_{4nt}(\widetilde{\theta}_n)=o_p(1)
\end{equation}
	and
\begin{equation}\label{beq6}
E\sup_{\theta\in \Theta_n}\left |\sum_{t=1}^{n}\psi_{\tau}(\varepsilon_t-b_{\tau})\xi_{5nt}(\theta)\right |\leq \sum_{t=1}^{n}E\sup_{\theta\in \Theta_n}|\xi_{5nt}(\theta)|=O(n^{-1/2}),
\end{equation}
	which implies
	\begin{equation}\label{A3}
		\sum_{t=1}^{n}A_{3nt}(\widetilde{\theta}_n)=o_p(1).
	\end{equation}
	By a method similar to that for $\sum_{t=1}^{n}A_{3nt}(\widetilde{\theta}_n)$, we can show that
	\[ \sum_{t=1}^{n}\psi_{\tau}(\varepsilon_t-b_{\tau})[h_t^{-1}(\widetilde{\theta}_n)-h_t^{-1}]\xi_{1nt}(\widetilde{\theta}_n)=o_p(1), \]
	which  implies
	\begin{align}\label{A4}
		 \sum_{t=1}^{n}A_{4nt}(\widetilde{\theta}_n)&=\sum_{t=1}^{n}\psi_{\tau}(\varepsilon_t-b_{\tau})h_t^{-1}\xi_{1nt}(\widetilde{\theta}_n)+o_p(1)=u^\prime T_{1n}+T_{2n}+o_p(1),
	\end{align}
	where
	\[ T_{1n}=\frac{1}{\sqrt{n}}\sum_{t=1}^{n}\psi_{\tau}(\varepsilon_t-b_{\tau})\frac{ z_t}{h_t} \hspace{2mm}\text{and}\hspace{2mm}T_{2n}=\sqrt{n}(\widetilde{\theta}_n-\theta_0)^\prime\frac{1}{\sqrt{n}} \sum_{t=1}^{n}\psi_{\tau}(\varepsilon_t-b_{\tau})\sum_{j=1}^{p}\frac{\beta_{\tau0}^{(j)}}{h_t}\frac{\partial h_{t-j}(\theta_0)}{\partial\theta}. \]
	Combining \eqref{L1a}, \eqref{A1}, \eqref{A2}, \eqref{A3}, and \eqref{A4}, we have
	\begin{equation}\label{L1}
		L_{1n}(u)=u^\prime T_{1n}+T_{2n}+o_p(1).
	\end{equation}
	
	Next we consider $L_{2n}(u)$.  For simplicity, denote
	$I_t^*(s)=I(\breve{e}_{t,\tau}\leq s)-I(\breve{e}_{t,\tau}\leq 0)$.
	From \eqref{beq1}, we have the decomposition
	\begin{equation}\label{L2b}
		 L_{2n}(u)=\sum_{t=1}^{n}B_{1nt}(\widetilde{\theta}_n)+\sum_{t=1}^{n}B_{2nt}(\widetilde{\theta}_n)+\sum_{t=1}^{n}B_{3nt}(\widetilde{\theta}_n)+\sum_{t=1}^{n}B_{4nt}(\widetilde{\theta}_n),
	\end{equation}
	where
	\begin{align*}
		 B_{1nt}(\theta)&=\widetilde{h}_t^{-1}(\theta)\int_{\xi_{1nt}(\theta)+\xi_{2nt}(\theta)}^{\xi_{1nt}(\theta)+\xi_{2nt}(\theta)+\xi_{3nt}(\theta)}I_t^*(s)ds+[\widetilde{h}^{-1}_t(\theta)-h^{-1}_t(\theta)]\int_{0}^{\xi_{1nt}(\theta)+\xi_{2nt}(\theta)}I_t^*(s)ds,\\
		 B_{2nt}(\theta)&=h^{-1}_t(\theta)\int_{\xi_{1nt}(\theta)}^{\xi_{1nt}(\theta)+\xi_{2nt}(\theta)}I_t^*(s)ds,\\
		B_{3nt}(\theta)&=[h^{-1}_t(\theta)-h^{-1}_t]\int_{0}^{\xi_{1nt}(\theta)}I_t^*(s)ds, \hspace{5mm}\text{and}\hspace{5mm}B_{4nt}(\theta)=h^{-1}_t\int_{0}^{\xi_{1nt}(\theta)}I_t^*(s)ds.
	\end{align*}
	
	By a method similar to that for \eqref{A1}, we can show that
	\begin{equation}\label{B1nt}
	\begin{split}
		\sup_{\theta\in\Theta_n}\left |\sum_{t=1}^{n}B_{1nt}(\theta)\right |&\leq \sum_{t=1}^{n}\sup_{\theta\in\Theta_n}\left [ \frac{|\xi_{3nt}(\theta)|}{\widetilde{h}_t(\theta)}+\left |\frac{1}{\widetilde{h}_t(\theta)}-\frac{1}{h_t(\theta)}\right |\left (|\xi_{1nt}(\theta)|+|\xi_{2nt}(\theta)|\right )\right ]\\&=o_p(1),
	\end{split}
	\end{equation}
	which, together with the fact that $\sqrt{n}(\widetilde{\theta}_n-\theta_0)=O_p(1)$, implies
	\begin{equation}\label{B1}
	\sum_{t=1}^{n}	B_{1nt}(\widetilde{\theta}_n)=o_p(1).
	\end{equation}
	
	From \eqref{beq5}, \eqref{beq3}, \eqref{beq13}, Assumption \ref{assum2} and the H\"{o}lder inequality, we have
	\begin{align*}
		& E\sup_{\theta\in\Theta_n}\left |\sum_{t=1}^{n}B_{2nt}(\theta)\right |\\
		&\hspace{5mm}\leq E \sum_{t=1}^{n}	 \sup_{\theta\in\Theta_n}|h^{-1}_t(\theta)\xi_{2nt}(\theta)| I\left ( |\breve{e}_{t,\tau}|\leq \sup_{\theta\in\Theta_n}\left (|\xi_{1nt}(\theta)|+|\xi_{2nt}(\theta)|\right )\right)\\
		& \hspace{5mm}\leq \sqrt{2\sup_{x\in\mathbb{R}}f(x)}\sum_{t=1}^{n} \left [E\sup_{\theta\in\Theta_n}\left |\frac{\xi_{2nt}(\theta)}{h_t(\theta)}\right |^2\right ]^{1/2} \left [E\sup_{\theta\in\Theta_n}\frac{\left (|\xi_{1nt}(\theta)|+|\xi_{2nt}(\theta)|\right )}{h_t}\right ]^{1/2}\\
		&\hspace{5mm}=o(1),
	\end{align*}
	which, combined with the fact that $\sqrt{n}(\widetilde{\theta}_n-\theta_0)=O_p(1)$, yields
	\begin{equation}\label{B2}
		\sum_{t=1}^{n}B_{2nt}(\widetilde{\theta}_n)=o_p(1).
	\end{equation}

	Similarly, it follows from \eqref{beq5}, \eqref{beq3},  Assumption \ref{assum2} and the H\"{o}lder inequality  that
\[E\sup_{\theta\in\Theta_n}\left |\sum_{t=1}^{n}B_{3nt}(\theta)\right |\leq E \sum_{t=1}^{n}	 \sup_{\theta\in\Theta_n} \left |[h^{-1}_t(\theta) -h_t^{-1} ]\xi_{1nt}(\theta)\right | I\left (|\breve{e}_{t,\tau}|\leq \sup_{\theta\in\Theta_n}|\xi_{1nt}(\theta)|\right)=o(1),\]
	and then
	\begin{equation}\label{B3}
		\sum_{t=1}^{n}B_{3nt}(\widetilde{\theta}_n)=o_p(1).
	\end{equation}
	
	Finally, for $\sum_{t=1}^{n}B_{4nt}(\widetilde{\theta}_n)$, denote
	\[B_{4nt}^*(\theta)=h_t^{-1}\int_{0}^{\xi_{1nt}(\theta)}\left [F(b_{\tau}-h_t^{-1}a_t+h_t^{-1}s)-F(b_{\tau}-h_t^{-1}a_t)\right ]ds,\]
	and we first show that
	\begin{equation}\label{B4_1}
		 \sum_{t=1}^{n}B_{4nt}(\widetilde{\theta}_n)=\sum_{t=1}^{n}B_{4nt}^*(\widetilde{\theta}_n)+o_p(1).
	\end{equation}
	For any $v\in\mathbb{R}^{p+q+1}$, let $\eta_t(v) = h_t^{-1}\int_{0}^{\xi_{1nt}(\theta_0+n^{-1/2}v)} I_t^*(s)ds$,
	and denote
	\[S_n(v)=\sum_{t=1}^{n}\left [B_{4nt}(\theta_0+n^{-1/2}v)-B_{4nt}^*(\theta_0+n^{-1/2}v) \right ]= \sum_{t=1}^{n}\left \{\eta_t(v) -E[\eta_t(v) |\mathcal{F}_{t-1}]\right \}.\]
	For any fixed $v$ such that $\|v\|\leq M$, by \eqref{beq3}, Lemma \ref{lem2} and Assumption \ref{assum2}, we have
	\begin{align}
		E\eta_t^2(v)&\leq E\left \{ \frac{|\xi_{1nt}(\theta_0+n^{-1/2}v)|}{h_t^2}\int_{0}^{\xi_{1nt}(\theta_0+n^{-1/2}v)} [F(b_{\tau}-\frac{a_t}{h_t}+\frac{s}{h_t})-F(b_{\tau}-\frac{a_t}{h_t})]ds\right \}\notag\\
		&\leq \dfrac{1}{2}\sup_{x\in\mathbb{R}}f(x)E|h_t^{-1}\xi_{1nt}(\theta_0+n^{-1/2}v)|^3 \leq n^{-3/2}C, \label{forCor1}
	\end{align}
	implying that
	\begin{equation}\label{beq14}
		ES_n^2(v)\leq \sum_{t=1}^{n}E\eta_t^2(v)=o(1).
	\end{equation}
	Note that
	\[ h_t^{-1}\sup_{\|v_1-v_2\|\leq \delta} |\xi_{1nt}(\theta_0+n^{-1/2}v_1)-\xi_{1nt}(\theta_0+n^{-1/2}v_2)| \leq \frac{\delta}{\underline{w}\sqrt{n}}\sum_{j=1}^{p}|\beta_{\tau0}^{(j)}|\left \|\frac{1}{h_{t-j}}	 \frac{\partial h_{t-j}(\theta_0)}{\partial\theta}\right \|.\]
	Then, for any $v_1, v_2 \in \mathbb{R}^{p+q+1}$ such that $\|v_1\|, \|v_2\|\leq M$, in view of \eqref{beq5}, \eqref{beq3}, Lemma \ref{lem2} and Assumption \ref{assum2}, we have
\begin{align*}
& E\sup_{\|v_1-v_2\|\leq \delta} |\eta_{t}(v_1)-\eta_{t}(v_2)| \\
&\hspace{5mm}=E\left [h_t^{-1}\sup_{\|v_1-v_2\|\leq \delta}\left | \int_{\xi_{1nt}(\theta_0+n^{-1/2}v_2)}^{\xi_{1nt}(\theta_0+n^{-1/2}v_1)}I_t^*(s)ds\right |\right]\\
&\hspace{5mm}\leq E \bigg[h_t^{-1}\sup_{\|v_1-v_2\|\leq \delta} |\xi_{1nt}(\theta_0+n^{-1/2}v_1)-\xi_{1nt}(\theta_0+n^{-1/2}v_2)|I\big(|\breve{e}_{t,\tau}|\leq\sup_{\theta\in\Theta_n}|\xi_{1nt}(\theta)|\big) \bigg]\\
&\hspace{5mm}\leq \frac{2\delta}{\underline{w}\sqrt{n}}\sup_{x\in\mathbb{R}}f(x)E\left ( \sup_{\theta\in\Theta_n}\frac{|\xi_{1nt}(\theta)|}{h_t}\sum_{j=1}^{p}|\beta_{\tau0}^{(j)}|\left \|\frac{1}{h_{t-j}}	 \frac{\partial h_{t-j}(\theta_0)}{\partial\theta}\right \| \right ) \leq n^{-1}\delta C,
\end{align*}
	and hence
	\[ E\sup_{\|v_1-v_2\|\leq \delta} |S_n(v_1)-S_n(v_2)| \leq 2\sum_{t=1}^{n} E\sup_{\|v_1-v_2\|\leq \delta} |\eta_{t}(v_1)-\eta_{t}(v_2)|\leq 2\delta C,\]
	which, together with \eqref{beq14} and the finite covering theorem, implies $\sup_{\|v\|\leq M}|S_n(v)|=o_p(1)$, and then \eqref{B4_1} holds.
	
	By  elementary calculation and the Taylor expansion, we have
	\begin{align}\label{B4}
		 \sum_{t=1}^{n}B_{4nt}^*(\theta)&=\sum_{t=1}^{n}h_t^{-1}\int_{0}^{\xi_{1nt}(\theta)}f(b_{\tau}-h_t^{-1}a_t)h_t^{-1}s ds+R_{1n}(\theta)\notag\\
		 &=\frac{1}{2}f(b_{\tau})\sum_{t=1}^{n}h_t^{-2}\xi_{1nt}^2(\theta)+R_{2n}(\theta)+R_{1n}(\theta),
	\end{align}
	where
	 \[R_{1n}(\theta)=\frac{1}{2}\sum_{t=1}^{n}h_t^{-3}\int_{0}^{\xi_{1nt}(\theta)}\dot{f}(b_{\tau,t}^*(s))s^2 ds,\]
	with $b_{\tau,t}^*(s)$ lying between $b_{\tau}-h_t^{-1}a_t$ and $b_{\tau}-h_t^{-1}a_t+h_t^{-1}s$, and
	\[
	 R_{2n}(\theta)=\frac{1}{2}\sum_{t=1}^{n}h_t^{-2}\xi_{1nt}^2(\theta)[f(b_{\tau}-h_t^{-1}a_t)-f(b_{\tau})].
	\]
	Note that
	\[\sup_{\theta\in\Theta_n}|R_{1n}(\theta)|\leq \frac{1}{6}\sup_{x\in\mathbb{R}}|\dot{f}(x)|\sum_{t=1}^{n}\sup_{\theta\in\Theta_n}\left |\frac{\xi_{1nt}(\theta)}{h_t}\right |^3,\]
	and by Lemma \ref{lem3},
	\[\sup_{\theta\in\Theta_n}|R_{2n}(\theta)|\leq  \frac{1}{2}C\sup_{x\in\mathbb{R}}|\dot{f}(x)|\zeta\sum_{t=1}^{n}\rho^t\sup_{\theta\in\Theta_n}\left |\frac{\xi_{1nt}(\theta)}{h_t}\right |^2.\]
	Then, by \eqref{beq3}, Lemma \ref{lem2} and Assumption \ref{assum2}, we have
	\[
	R_{1n}(\widetilde{\theta}_n)=o_p(1) \hspace{5mm}\text{and}\hspace{5mm}R_{2n}(\widetilde{\theta}_n)=o_p(1).
	\]
Hence, by \eqref{L2b}, \eqref{B1}-\eqref{B4_1} and \eqref{B4}, together with the ergodic theorem, we have
	\begin{equation}\label{L2}
		\begin{split}
			 L_{2n}(u)&=\frac{1}{2}f(b_{\tau})\sum_{t=1}^{n}h_t^{-2}\xi_{1nt}^2(\widetilde{\theta}_n)+o_p(1)\\
			&=\frac{1}{2}f(b_{\tau})u^\prime\Omega_2 u+b_{\tau} f(b_{\tau})u^\prime\Gamma_2\sqrt{n}(\widetilde{\theta}_n-\theta_0)+T_{3n}+o_p(1),
		\end{split}
	\end{equation}
	where
	 \[T_{3n}=\dfrac{1}{2}f(b_{\tau})(\widetilde{\theta}_n-\theta_0)^\prime\sum_{t=1}^{n}\sum_{j_1=1}^{p}\sum_{j_2=1}^{p} \beta_{\tau0}^{(j_1)}\beta_{\tau0}^{(j_2)}\frac{1}{h_t^2}\frac{\partial h_{t-j_1}(\theta_0)}{\partial\theta}\frac{\partial h_{t-j_2}(\theta_0)}{\partial\theta^\prime}(\widetilde{\theta}_n-\theta_0).
	\]
	
	Combining  \eqref{L1_L2}, \eqref{L1} and \eqref{L2} yields that
	\begin{align*}
		L_n(\theta_{\tau 0}+n^{-1/2}u)-\breve{L}_n(\theta_{\tau 0})=&-u^\prime \left [T_{1n}-b_{\tau} f(b_{\tau})\Gamma_2\sqrt{n}(\widetilde{\theta}_n-\theta_0) \right ] +\frac{1}{2}f(b_{\tau})u^\prime\Omega_2 u\\
		&-T_{2n}+T_{3n}+o_p(1),
	\end{align*}
	where
	\begin{equation}\label{qmle}
		 \sqrt{n}(\widetilde{\theta}_{n}-\theta_{0})=-\frac{J^{-1}}{\sqrt{n}}\sum_{t=1}^{n}\frac{1-|\varepsilon_t|}{h_t}\frac{\partial h_t(\theta_0)}{\partial\theta}+o_p(1);
	\end{equation}
	see \cite{Francq_Zakoian2004}.
	Applying the central limit theorem and Corollary 2 in \cite{Knight1998}, together with the convexity of $L_n(\cdot)$, we have
	\begin{equation}\label{bahadur2}
	\sqrt{n}(\widehat{\theta}_{\tau n}-\theta_{\tau 0})=\frac{\Omega_2^{-1}}{f(b_{\tau})}T_{1n}- b_{\tau}\Omega_2^{-1} \Gamma_2 \sqrt{n}(\widetilde{\theta}_n-\theta_0)+o_p(1)\rightarrow_d N(0, \Sigma_1),
	\end{equation}
where $T_{1n}=n^{-1/2}\sum_{t=1}^{n}\psi_{\tau}(\varepsilon_t-b_{\tau}) z_t/h_t$.
	The proof is complete.
\end{proof}

\begin{proof}[Proof of Theorem \ref{thm3}]
	Let $L_n^*(\theta)=\sum_{t=1}^{n}\omega_t\widetilde h_{t}^{-1}\rho_{\tau}(y_t-\theta^{\prime}\widetilde{z}_t^*)$ and
	$\breve{L}_n^*(\theta)=\sum_{t=1}^{n}\omega_t\widetilde h_{t}^{-1}\rho_{\tau}(y_t-\theta^{\prime}\breve{z}_t)$.  For any fixed $u\in\mathbb{R}^{p+q+1}$, similar to \eqref{L1_L2},  it holds that
	\begin{equation}\label{L1_L2b}
	L_n^*(\theta_{\tau 0}+n^{-1/2}u)-\breve{L}_n^*(\theta_{\tau 0})=-L_{1n}^*(u)+L_{2n}^*(u),
	\end{equation}
	where
	\begin{align*}
	L_{1n}^*(u)&= \sum_{t=1}^{n}\omega_t\psi_{\tau}(\breve{e}_{t,\tau})\widetilde{h}^{-1}_t \left[(\theta_{\tau 0}+n^{-1/2}u)^{\prime}\widetilde{z}_t^*-\theta_{\tau 0}^{\prime}\breve{z}_t\right], \\
	L_{2n}^*(u)&=\sum_{t=1}^{n}\omega_t\widetilde{h}^{-1}_t\int_{0}^{(\theta_{\tau 0}+n^{-1/2}u)^{\prime}\widetilde{z}_t^*-\theta_{\tau 0}^{\prime}\breve{z}_t}\left[I(\breve{e}_{t,\tau}\leq s)-I(\breve{e}_{t,\tau}\leq 0)\right]ds,
	\end{align*}
	and
	\[
	(\theta_{\tau 0}+n^{-1/2}u)^{\prime}\widetilde{z}_t^*-\theta_{\tau 0}^{\prime}\breve{z}_t =\xi_{1nt}(\widetilde{\theta}_n^*)+\xi_{2nt}(\widetilde{\theta}_n^*) +\xi_{3nt}(\widetilde{\theta}_n^*).
	\]
	From the proof of Theorem \ref{thm1}, we have $\widetilde{J}=J+o_p(1)$, which together with \eqref{boots} implies
	\begin{equation}\label{beq}
	 \sqrt{n}(\widetilde{\theta}_n^*-\widetilde{\theta}_n)=-\frac{J^{-1}}{\sqrt{n}}\sum_{t=1}^{n}(\omega_t-1)\left (1-\frac{|y_t|}{h_t}\right )\frac{1}{h_t}\frac{\partial h_t(\theta_0)}{\partial\theta}+o_p^*(1),
	\end{equation}
	and
	\begin{equation}\label{proof1}
	\sqrt{n}(\widetilde{\theta}_n^*-\theta_0)=\sqrt{n}(\widetilde{\theta}_n^*-\widetilde{\theta}_n) +\sqrt{n}(\widetilde{\theta}_n-\theta_0) =O_p^*(1).
	\end{equation}
	
	Without any confusion, we redefine the functions $A_{int}$ with $1\leq i\leq 4$ from the proof of Theorem \ref{thm1} as follows,
	\begin{align*}
	A_{1nt}(\theta_1, \theta_2)&=\psi_{\tau}(\breve{e}_{t,\tau})\widetilde{h}_t^{-1}(\theta_1)\xi_{3nt}(\theta_2)+ \psi_{\tau}(\breve{e}_{t,\tau}) [\widetilde{h}_t^{-1}(\theta_1)-h_t^{-1}(\theta_1)][\xi_{1nt}(\theta_2)+\xi_{2nt}(\theta_2)],\\
	A_{2nt}(\theta_1,\theta_2)&= [\psi_{\tau}(\breve{e}_{t,\tau})-\psi_{\tau}(\varepsilon_t-b_{\tau})] h_t^{-1}(\theta_1)[\xi_{1nt}(\theta_2)+\xi_{2nt}(\theta_2)],\\A_{3nt}(\theta_1,\theta_2)&=\psi_{\tau}(\varepsilon_t-b_{\tau})h_t^{-1}(\theta_1) \xi_{2nt}(\theta_2),\hspace{2mm}\text{and}\hspace{2mm} A_{4nt}(\theta_1,\theta_2)=\psi_{\tau}(\varepsilon_t-b_{\tau})h_t^{-1}(\theta_1)\xi_{1nt}(\theta_2),
	\end{align*}
	as well as $B_{int}$ with $1\leq i\leq 3$ as follows,
	\begin{align*}
	 B_{1nt}(\theta_1,\theta_2)&=\widetilde{h}_t^{-1}(\theta_1)\int_{\xi_{1nt}(\theta_2)+\xi_{2nt}(\theta_2)}^{\xi_{1nt}(\theta_2)+\xi_{2nt}(\theta_2)+\xi_{3nt}(\theta_2)}I_t^*(s)ds\\
	 &\hspace{5mm}+[\widetilde{h}^{-1}_t(\theta_1)-h^{-1}_t(\theta_1)]\int_{0}^{\xi_{1nt}(\theta_2)+\xi_{2nt}(\theta_2)}I_t^*(s)ds,\\
	 B_{2nt}(\theta_1,\theta_2)&=h^{-1}_t(\theta_1)\int_{\xi_{1nt}(\theta_2)}^{\xi_{1nt}(\theta_2)+\xi_{2nt}(\theta_2)}I_t^*(s)ds,\hspace{5mm}\text{and}\\
	B_{3nt}(\theta_1,\theta_2)&=[h^{-1}_t(\theta_1)-h^{-1}_t]\int_{0}^{\xi_{1nt}(\theta_2)}I_t^*(s)ds,
	\end{align*}
	while the definition of $B_{4nt}(\cdot)$ is the same as in the proof of Theorem \ref{thm1}.
	
	By methods similar to \eqref{A1}, \eqref{A2}, \eqref{A3} and \eqref{A4} respectively, together with Assumption \ref{assum2}, Lemma \ref{lem3}, \eqref{beq11}, \eqref{beq12} and \eqref{beq}, we can show that
	\[
	\sum_{t=1}^{n}\omega_tA_{int}(\widetilde{\theta}_n, \widetilde{\theta}_n^*)=o_p^*(1),\hspace{5mm} 1\leq i\leq 3,
	\]
	and
	\[
	\sum_{t=1}^{n}\omega_tA_{4nt}(\widetilde{\theta}_n,\widetilde{\theta}_n^*) =\sum_{t=1}^{n}\omega_t\psi_{\tau}(\varepsilon_t-b_{\tau})h_t^{-1}\xi_{1nt}(\widetilde{\theta}_n^*) +o_p^*(1) =u^\prime T_{1n}^*+T_{2n}^*+o_p^*(1),
	\]
	where $T_{1n}^* =n^{-1/2}\sum_{t=1}^{n}\omega_t\psi_{\tau}(\varepsilon_t-b_{\tau})z_t/h_t$ and
	\[
	T_{2n}^*=\sqrt{n}(\widetilde{\theta}_n^*-\theta_0)^\prime\frac{1}{\sqrt{n}} \sum_{t=1}^{n}\omega_t\psi_{\tau}(\varepsilon_t-b_{\tau}) \sum_{j=1}^{p}\frac{\beta_{\tau0}^{(j)}}{h_t}\frac{\partial h_{t-j}(\theta_0)}{\partial\theta},
	\]
	where $\beta_{\tau0}^{(j)}=b_{\tau}\beta_{0j}$, $j=1,\dots, p$, is defined as in the proof of Theorem \ref{thm1}.
	As a result,
	\begin{align}
	\begin{split}\label{L1b}
	L_{1n}^*(u) &=\sum_{t=1}^{n}\omega_tA_{1nt}(\widetilde{\theta}_n, \widetilde{\theta}_n^*) +\sum_{t=1}^{n}\omega_tA_{2nt}(\widetilde{\theta}_n, \widetilde{\theta}_n^*) +\sum_{t=1}^{n}\omega_tA_{3nt}(\widetilde{\theta}_n, \widetilde{\theta}_n^*) +\sum_{t=1}^{n}\omega_tA_{4nt}(\widetilde{\theta}_n, \widetilde{\theta}_n^*)\\
	&=u^\prime T_{1n}^*+T_{2n}^*+o_p^*(1).
	\end{split}
	\end{align}
	Moreover, by methods similar to \eqref{B1}-\eqref{B3}, we can verify that
	\[
	\sum_{t=1}^{n}(\omega_t-1)B_{int}(\widetilde{\theta}_n, \widetilde{\theta}_n^*)=o_p^*(1),\hspace{2mm} 1\leq i\leq3,\hspace{2mm}\text{and}\hspace{2mm}\sum_{t=1}^{n}(\omega_t-1)B_{4nt}(\widetilde{\theta}_n^*) =o_p^*(1),
	\]
	which implies
	\[L_{2n}^*(u)=\sum_{t=1}^{n}B_{1nt}(\widetilde{\theta}_n, \widetilde{\theta}_n^*) +\sum_{t=1}^{n}B_{2nt}(\widetilde{\theta}_n, \widetilde{\theta}_n^*)+\sum_{t=1}^{n}B_{3nt}(\widetilde{\theta}_n, \widetilde{\theta}_n^*) +\sum_{t=1}^{n}B_{4nt}(\widetilde{\theta}_n^*)+o_p^*(1),\]
	and hence, similar to the proof of \eqref{L2}, it can be further verified that
	\begin{equation}\label{L2b2}
	L_{2n}^*(u)=\frac{1}{2}f(b_{\tau})u^\prime\Omega_2 u+b_{\tau} f(b_{\tau})u^\prime\Gamma_2\sqrt{n}(\widetilde{\theta}_n^*-\theta_0)+T_{3n}^*+o_p^*(1),
	\end{equation}
	where
	\[
	 T_{3n}^*=\dfrac{1}{2}f(b_{\tau})(\widetilde{\theta}_n^*-\theta_0)^\prime\sum_{t=1}^{n}\sum_{j_1=1}^{p}\sum_{j_2=1}^{p} \beta_{\tau0}^{(j_1)}\beta_{\tau0}^{(j_2)}\frac{1}{h_t^2}\frac{\partial h_{t-j_1}(\theta_0)}{\partial\theta}\frac{\partial h_{t-j_2}(\theta_0)}{\partial\theta^\prime}(\widetilde{\theta}_n^*-\theta_0).
	\]
	
	Therefore, combining \eqref{L1_L2b}, \eqref{L1b} and \eqref{L2b2}, we have
	\begin{align*}
	L_n^*(\theta_{\tau 0}+n^{-1/2}u)-\breve{L}_n^*(\theta_{\tau 0})=& -u^\prime\left [T_{1n}^*-b_{\tau} f(b_{\tau})\Gamma_2\sqrt{n}(\widetilde{\theta}_n^*-\theta_0)\right ]+\frac{1}{2}f(b_{\tau})u^\prime\Omega_2 u\\&-T_{2n}^*+T_{3n}^*+o_p^*(1),
	\end{align*}
	where $T_{1n}^* =n^{-1/2}\sum_{t=1}^{n}\omega_t\psi_{\tau}(\varepsilon_t-b_{\tau})z_t/h_t$.
	
	Denote $X_t=n^{-1/2}(\omega_t-1)\psi_{\tau}(\varepsilon_t-b_{\tau})z_t/h_t$, and then $T_{1n}^*-T_{1n}=\sum_{t=1}^{n}X_t $. For any constant vector $c\in\mathbb{R}^{p+q+1}$, let $\mu_t=E^*(c^\prime X_t)$ and $\sigma_n^2 = \sum_{t=1}^{n}E^*(c^\prime X_t X_t^\prime c)$.
	Then, $\mu_t=0$, and by \eqref{beq2} we have
	\begin{align*}
	\left (\sum_{t=1}^{n}E^*|c^\prime X_t-\mu_t|^{2+\delta}\right )^{\frac{1}{2+\delta}}& =\frac{1}{\sqrt{n}} \left [\sum_{t=1}^{n}\left |\psi_{\tau}(\varepsilon_t-b_{\tau})\frac{c^\prime z_t}{h_t}\right |^{2+\delta}\right ]^{\frac{1}{2+\delta}}(E^*|\omega_t-1|^{2+\delta})^{\frac{1}{2+\delta}}\\
	&=o_p(1),
	\end{align*}
	as long as $0<\delta\leq \kappa_0$, since $E^*|\omega_t|^{2+\kappa_0}<\infty$ from the assumptions of this theorem. Moreover, by the ergodic theorem, $\sigma_n^2 = c^\prime n^{-1}\sum_{t=1}^{n}[\psi_{\tau}(\varepsilon_t-b_{\tau})]^2h_t^{-2}z_t z_t^\prime c=\tau (1-\tau) c^\prime \Omega_2 c +o_p(1)$.  Thus, we can show that the Liapounov's condition,
	$\sum_{t=1}^{n}E^*|c^\prime X_t-\mu_t|^{2+\delta}=o_p(\sigma_n^{2+\delta})$,
	holds for $0<\delta\leq \kappa_0$. This, together with the Cram\'{e}r-Wold device and the Lindeberg's central limit theorem, implies that conditional on $\mathcal{F}_n$,
	\[T_{1n}^*-T_{1n}=\sum_{t=1}^{n}X_t \rightarrow_{d} N(0,\tau(1-\tau)\Omega_2)\]
	in probability as $n\rightarrow\infty$.
	
	Since $L_n^*(\cdot)$ is convex, by Lemma 2.2 of \cite{Davis_Knight_Liu1992} and Corollary 2 of \cite{Knight1998}, it holds that
	\begin{equation}\label{proof2}	
	\sqrt{n}(\widehat{\theta}_{\tau n}^{*}-\theta_{\tau 0})=\frac{\Omega_2^{-1}}{f(b_{\tau})}T_{1n}^*- b_{\tau} \Omega_2^{-1} \Gamma_2 \sqrt{n}(\widetilde{\theta}_n^*-\theta_0)+o_p^*(1),
	\end{equation}
	which, in conjunction with \eqref{bahadur2}, yields the Bahadur representation of the corrected bootstrap estimator $\widehat{\theta}_{\tau n}^{*}$,
	\[  \sqrt{n}(\widehat{\theta}_{\tau n}^{*}-\widehat{\theta}_{\tau n}) =\frac{\Omega_2^{-1}}{f(b_{\tau})}\left (T_{1n}^*-T_{1n}\right )+  \frac{b_{\tau} \Omega_2^{-1} \Gamma_2 J^{-1}}{\sqrt{n}}\sum_{t=1}^{n}(\omega_t-1) \frac{1-|\varepsilon_t|}{h_t}\frac{\partial h_t(\theta_0)}{\partial \theta}+o_p^*(1).\]
	Denote $X_t^\dag=n^{-1/2}(\omega_t-1)d_t$, with $d_t=(\psi_{\tau}(\varepsilon_t-b_{\tau}) z_t^\prime/h_t, (1-|\varepsilon_t|)h_t^{-1}\partial h_t(\theta_0)/\partial \theta^\prime)^\prime$.
	Note that by \eqref{beq2} and $E|\varepsilon_t|^{2+\nu_0}<\infty$ for $\nu_0>0$, we have $E|d_t|^{2+\nu_0}<\infty$.  Then,  for $0<\delta\leq \min(\kappa_0, \nu_0)$, we can similarly verify the  Liapounov's condition,
	$\sum_{t=1}^{n}E^*|c^\prime X_t^\dag-\mu_t^\dag|^{2+\delta}=o_p(\sigma_n^{\dag 2+\delta})$,  where $\mu_t^\dag=E^*(c^\prime X_t^\dag)$ and $\sigma_n^{\dag2} = \sum_{t=1}^{n}E^*(c^\prime X_t^\dag X_t^{\dag \prime} c)$.  Applying the Lindeberg's central limit theorem and the Cram\'{e}r-Wold device, we accomplish the proof of the theorem.
\end{proof}	

\begin{proof}[Proof of Theorem \ref{thm2}]
	Observe that
	\begin{align}\label{deq1}
		\begin{split}
			\frac{1}{\sqrt{n}}&\sum_{t=k+1}^{n} \psi_{\tau}(\widehat{\varepsilon}_{t, \tau})|\widehat{\varepsilon}_{t-k, \tau}|
			\\
			&=\frac{1}{\sqrt{n}}\sum_{t=k+1}^{n}\psi_{\tau}(\varepsilon_{t, \tau})|\varepsilon_{t-k, \tau}|+\sum_{t=k+1}^{n}\mathcal{E}_{1nt}+\sum_{t=k+1}^{n}\mathcal{E}_{2nt}+\sum_{t=k+1}^{n}\mathcal{E}_{3nt},
		\end{split}
	\end{align}	
	where
	\begin{equation*}
		\begin{split}
			\mathcal{E}_{1nt}&=n^{-1/2}[\psi_{\tau}(\widehat{\varepsilon}_{t, \tau})-\psi_{\tau}(\varepsilon_{t,\tau})]|\varepsilon_{t-k, \tau}|,\hspace{3mm}
			 \mathcal{E}_{2nt}=n^{-1/2}\psi_{\tau}(\varepsilon_{t,\tau})(|\widehat{\varepsilon}_{t-k,\tau}|-|\varepsilon_{t-k, \tau}|),\hspace{3mm}\text{and}\\
			 \mathcal{E}_{3nt}&=n^{-1/2}[\psi_{\tau}(\widehat{\varepsilon}_{t,\tau})-\psi_{\tau}(\varepsilon_{t, \tau})](|\widehat{\varepsilon}_{t-k, \tau}|-|\varepsilon_{t-k, \tau}|).
		\end{split}
	\end{equation*}
	To derive the asymptotic result for the quantity on the left-hand side of \eqref{deq1}, we shall begin by proving that
	\begin{equation}\label{E1}
		\sum_{t=k+1}^{n}\mathcal{E}_{1nt}=-f(b_{\tau})\left [d_{1k}^\prime \sqrt{n}(\widehat{\theta}_{\tau n}-\theta_{\tau0})+b_{\tau} d_{2k}^\prime \sqrt{n}(\widetilde{\theta}_{n}-\theta_0)\right ]+o_p(1),
	\end{equation}
	where $d_{1k}=E(h_t^{-1}|\varepsilon_{t-k,\tau}|z_t)$ and $d_{2k}=E(h_t^{-1}|\varepsilon_{t-k, \tau}|\sum_{j=1}^{p}\beta_{0j}{\partial h_{t-j}(\theta_0)}/{\partial\theta})$.
	For any $u, v\in\mathbb{R}^{p+q+1}$, define
	\[\widetilde{b}_t(u,v)=(\theta_{\tau 0}+n^{-1/2}u)^\prime\widetilde{z}_t(\theta_0+n^{-1/2}v)h_t^{-1}.\]
	Since $\sqrt{n}(\widehat{\theta}_{\tau n}-\theta_{\tau0})=O_p(1)$, $\sqrt{n}(\widetilde{\theta}_{n}-\theta_0)=O_p(1)$, and
	\[\sum_{t=k+1}^{n}\mathcal{E}_{1nt}= \frac{1}{\sqrt{n}}\sum_{t=k+1}^{n}\left [I(\varepsilon_t<b_{\tau})-I(\varepsilon_t<\widehat{\theta}_{\tau n}^\prime\widetilde{z}_t h_t^{-1} )\right ]|\varepsilon_{t-k,\tau}|,\]
	to prove \eqref{E1}, it suffices to show that for any $M>0$,
	\begin{equation}\label{deq2}
		\sup_{\|u\|, \|v\|\leq M}\left |\frac{1}{\sqrt{n}}\sum_{t=k+1}^{n}\phi_t(u,v) +f(b_{\tau})\left (d_{1k}^\prime u+b_{\tau} d_{2k}^\prime v\right )\right |=o_p(1),
	\end{equation}
	where
	$\phi_t(u,v)=\{I(\varepsilon_t<b_{\tau})-I[\varepsilon_t<\widetilde{b}_t(u,v) ]\}|\varepsilon_{t-k,\tau}|.$
	
	Let $S_n(u,v)=n^{-1/2}\sum_{t=k+1}^{n} \{\phi_t(u,v)-E[\phi_t(u,v)|\mathcal{F}_{t-1}]\}$, and we shall first show that
	\begin{equation}\label{deq3}
		\sup_{\|u\|, \|v\|\leq M}|S_n(u, v)|=o_p(1).
	\end{equation}
	For any $u, v\in\mathbb{R}^{p+q+1}$, define
	\[b_t(u,v)=(\theta_{\tau 0}+n^{-1/2}u)^\prime z_t(\theta_0+n^{-1/2}v)h_t^{-1}.\]
	Note that for any $u_i, v_i \in \mathbb{R}^{p+q+1}$, $i=1, 2$, since
	\begin{align*}
		b_t(u_1, v_1)&-b_t(u_2, v_2)\\
		=&\sum_{j=1}^{p}\frac{\beta_{\tau 0}^{(j)} +n^{-1/2}u_1^{(j)}}{h_t}[h_{t-j}(\theta_0+n^{-1/2}v_1)-h_{t-j}(\theta_0+n^{-1/2}v_2)]\\
		 &+\frac{1}{\sqrt{n}}\sum_{j=1}^{p}\frac{u_1^{(j)}-u_2^{(j)}}{h_t}[h_{t-j}(\theta_0+n^{-1/2}v_2)-h_{t-j}]+\frac{ h_t^{-1}z_t^\prime(u_1-u_2)}{\sqrt{n}},
	\end{align*}
	by the Taylor expansion and \eqref{beq2},  where $\beta_{\tau0}^{(j)}=b_{\tau}\beta_{0j}$ for $j=1,\dots, p$, we can readily show that if $\|u_i\|,  \|v_i\|\leq M$, then
	\begin{align}\label{deq7}
		\begin{split}
			&|b_t( u_1, v_1)-b_t(u_2, v_2)|\\
			\leq &\frac{C(M)}{\sqrt{n}}\bigg[ \bigg (\|v_1-v_2\|+\frac{\|u_1-u_2\|}{\sqrt{n}}\bigg )\sum_{j=1}^{p}\sup_{\theta\in\Theta_n}\left \|\frac{1}{h_{t-j}}\frac{\partial h_{t-j}(\theta)}{\partial\theta}\right \| +\|u_1-u_2\|\bigg].
		\end{split}
	\end{align}
	
	For any $u, v\in\mathbb{R}^{p+q+1}$ such that $\|u\|, \|v\|\leq M$, by the H\"{o}lder inequality and the fact that $E|\varepsilon_t|^{2+\nu_0}<\infty$ for $\nu_0>0$,  we have
	\begin{align}\label{deq5}
		\begin{split}
			\sum_{t=k+1}^{n}E\phi_t^2(u,v)&\leq \sum_{t=k+1}^{n}\left \{ E\left |I(\varepsilon_t<b_{\tau})-I[\varepsilon_t<\widetilde{b}_t(u,v) ] \right | \right\}^{\frac{\nu_0}{2+\nu_0}} \left (E|\varepsilon_{t-k, \tau}|^{2+\nu_0}\right )^{\frac{2}{2+\nu_0}}\\
			&=C\sum_{t=k+1}^{n}\left[E\left |F(\widetilde{b}_t(u,v))-F(b_{\tau})\right | \right]^{\frac{\nu_0}{2+\nu_0}} \\
			&\leq C \bigg\{  \sum_{t=k+1}^{n}\left[E\left |F(\widetilde{b}_t(u,v))-F(b_t(u,v))\right | \right]^{\frac{\nu_0}{2+\nu_0}} \\
			& \hspace{15mm}+ \sum_{t=k+1}^{n}\left[E\left |F(b_t(u,v))-F(b_{\tau})\right | \right]^{\frac{\nu_0}{2+\nu_0}} \bigg\},
		\end{split}
	\end{align}
	where the last inequality follows from the fact that $(x+y)^a\leq x^a+y^a$ for any $x, y\geq 0$ and $0<a<1$.
	Note that by Lemma \ref{lem3}, we have
	\begin{align}\label{deq4}
		\begin{split}
			\sup_{\|u\|, \|v\|\leq M}|\widetilde{b}_t(u,v)-b_t(u,v)|&\leq \sum_{j=1}^{p}\frac{|\beta_{\tau 0}^{(j)}|+n^{-1/2}M}{\underline{w}}\sup_{\theta\in\Theta} |\widetilde{h}_{t-j}(\theta) -h_{t-j}(\theta)|\\
			&\leq C(M)\rho^t\zeta.
		\end{split}
	\end{align}
	Then, by Assumption \ref{assum2} and a method similar to that  for \eqref{beq10},  we can show that
	\begin{equation}\label{deq6}
		E\left |F(\widetilde{b}_t(u,v))-F(b_t(u,v))\right |\leq \rho^{t/2}+C(M)\rho^{{\delta_0} t/2}.
	\end{equation}
	Moreover, since $b_{\tau}=b_t(0, 0)$, it follows from \eqref{deq7}, Lemma \ref{lem2} and Assumption \ref{assum2} that
	\begin{equation}\label{deq8}
		E|F(b_t(u,v))-F(b_{\tau})|\leq \sup_{x\in\mathbb{R}}f(x)E|b_t(u,v)-b_{\tau}| \leq n^{-1/2}C(M).
	\end{equation}
	In view of \eqref{deq5}, \eqref{deq6} and \eqref{deq8}, for any $u, v\in\mathbb{R}^{p+q+1}$ with $\|u\|, \|v\|\leq M$,
	\begin{equation}\label{deq9}
		ES_n^2(u,v) \leq \frac{1}{n}\sum_{t=k+1}^{n}E\phi_t^2(u,v)=o(1).
	\end{equation}

	For any $\delta>0$, let $U(\delta)$ be the set of all four-tuples $(u_1, u_2, v_1, v_2)$ of column vectors in $\mathbb{R}^{p+q+1}$ such that $\|u_i\|, \|v_i\|\leq M$, $i=1,2$,  and $\|u_1-u_2\|, \|v_1-v_2\|\leq \delta$,  and denote by $\upsilon$ an element of $U(\delta)$.  Moreover, for simplicity, denote $\widetilde{b}_{ti}=\widetilde{b}_t(u_i,v_i)$ and $b_{ti}=b_t(u_i,v_i)$ for $i=1, 2$. Let $\widetilde{\Delta}_t=\sup_{\upsilon\in U(\delta)}|\widetilde{b}_{t1} -\widetilde{b}_{t2}|$ and $\Delta_t=\sup_{\upsilon\in U(\delta)}|b_{t1} -b_{t2}|$.
	Notice that
	\begin{align*}
		\begin{split}
			\sup_{\upsilon\in U(\delta)} |\phi_t(u_1,v_1)-\phi_t(u_2,v_2)| &\leq \sup_{\upsilon\in U(\delta)}  \big|I(\varepsilon_t<\widetilde{b}_{t2} )-I(\varepsilon_t<\widetilde{b}_{t1})\big||\varepsilon_{t-k, \tau}|\\
			&\leq  I\big( |\varepsilon_t-\widetilde{b}_{t2}|<\widetilde{\Delta}_t\big) |\varepsilon_{t-k, \tau}|.
		\end{split}
	\end{align*}
	Then, applying the H\"{o}lder inequality, together with $E|\varepsilon_t|^{2+\nu_0}<\infty$ for $\nu_0>0$ and  the fact that  $(x+y)^a\leq x^a+y^a$ for any $x, y\geq 0$ and $0<a<1$,  we have
	\begin{align}\label{deq10}
		\begin{split}
			& E \sup_{\upsilon\in U(\delta)} |\phi_t(u_1,v_1)-\phi_t(u_2,v_2)|
			\\
			&\hspace{5mm} \leq \left[E\big |F\big(\widetilde{b}_{t2}+\widetilde{\Delta}_t\big) - F\big(\widetilde{b}_{t2}-\widetilde{\Delta}_t\big) \big|\right ]^{1/2} (E\varepsilon_{t-k,\tau}^2)^{1/2}
			\\
			& \hspace{5mm}\leq  C \bigg\{\left [E\big|F\big(\widetilde{b}_{t2}+\widetilde{\Delta}_t\big) - F\big(\widetilde{b}_{t2}+\Delta_t\big)\big|\right ]^{1/2}+\left [E\big|F\big(\widetilde{b}_{t2}-\widetilde{\Delta}_t\big) - F\big(\widetilde{b}_{t2}-\Delta_t\big)\big|\right ]^{1/2}\\
			&\hspace{17mm}+\left [E\big|F\big(\widetilde{b}_{t2}+\Delta_t\big) - F\big(\widetilde{b}_{t2}-\Delta_t\big)\big|\right ]^{1/2}\bigg\}.
		\end{split}
	\end{align}
	Since
	$|\widetilde{\Delta}_t-\Delta_t|  \leq \sup_{\upsilon\in U(\delta)} \big|(\widetilde{b}_{t1}-\widetilde{b}_{t2})-(b_{t1}-b_{t2})\big| \leq 2\sup_{\|u\|, \|v\|\leq M}|\widetilde{b}_t(u,v)-b_t(u,v)|$,
	by \eqref{deq4} and a method similar to that for \eqref{beq10}, we can verify that
	\begin{equation}\label{deq11}
		E\big|F\big(\widetilde{b}_{t2}\pm\widetilde{\Delta}_t\big) - F\big(\widetilde{b}_{t2}\pm\Delta_t\big)\big| \leq \rho^{t/2}+C(M)\rho^{{\delta_0} t/2}.
	\end{equation}
	Furthermore, it follows from Assumption \ref{assum2}, \eqref{deq7} and Lemma \ref{lem2} that
	\begin{equation}\label{deq12}
		E\big|F\big(\widetilde{b}_{t2}+\Delta_t\big) - F\big(\widetilde{b}_{t2}-\Delta_t\big)\big|\leq 2\sup_{x\in\mathbb{R}}f(x) E(\Delta_t)\leq n^{-1/2}\delta C(M).
	\end{equation}
	As a result of \eqref{deq10}-\eqref{deq12}, we have
	\[ E\sup_{\upsilon\in U(\delta)}|S_n(u_1, v_1)-S_n(u_2, v_2)| \leq \frac{2}{\sqrt{n}}\sum_{t=k+1}^{n}E \sup_{\upsilon\in U(\delta)} |\phi_t(u_1,v_1)-\phi_t(u_2,v_2)| \leq \delta C(M),\]
	which, together with \eqref{deq9} and the finite covering theorem, implies \eqref{deq3}.

	Since $E[\phi_t(u,v)|\mathcal{F}_{t-1}] =[F(b_{\tau})- F(\widetilde{b}_t(u,v))]|\varepsilon_{t-k, \tau}|$,
	to prove \eqref{deq2}, it remains to show that
	\begin{equation}\label{deq14}
		\sup_{\|u\|, \|v\|\leq M}\left |\frac{1}{\sqrt{n}}\sum_{t=k+1}^{n}\big[F(b_{\tau})- F(\widetilde{b}_t(u,v))\big]|\varepsilon_{t-k, \tau}|+f(b_{\tau})\left (d_{1k}^\prime u+b_{\tau} d_{2k}^\prime v\right )\right |=o_p(1).
	\end{equation}
	By \eqref{deq4}, Assumption \ref{assum2} and a method similar to that for \eqref{beq10},  we can show that
	\begin{equation*}
		E\bigg(\sup_{\|u\|, \|v\|\leq M}\left |F(\widetilde{b}_t(u,v))-F(b_t(u,v))\right |\bigg)^2\leq \rho^{t}+C(M)\rho^{{\delta_0} t/2},
	\end{equation*}
	which, in conjunction with the H\"{o}lder inequality and $E|\varepsilon_t|^{2+\nu_0}<\infty$ for $\nu_0>0$, yields
	\begin{align*}
		& E\sup_{\|u\|, \|v\|\leq M}\left |\frac{1}{\sqrt{n}}\sum_{t=k+1}^{n}\big[F(\widetilde{b}_t(u,v))-F(b_t(u,v))\big]|\varepsilon_{t-k, \tau}| \right |\\
		& \hspace{5mm}\leq \frac{1}{\sqrt{n}}\sum_{t=k+1}^{n}\bigg[E\bigg(\sup_{\|u\|, \|v\|\leq M}\left |F(b_t(u,v))-F(\widetilde{b}_t(u,v))\right |\bigg)^2 \bigg]^{1/2}(E\varepsilon_{t-k,\tau}^2)^{1/2}=o(1),
	\end{align*}
	and hence,
	\begin{equation}\label{deq13}
		\sup_{\|u\|, \|v\|\leq M}\left |\frac{1}{\sqrt{n}}\sum_{t=k+1}^{n}\big[F(b_t(u,v))-F(\widetilde{b}_t(u,v))\big]|\varepsilon_{t-k, \tau}| \right |=o_p(1).
	\end{equation}
	Note that by the Taylor expansion,
	\[b_{\tau}-b_t(u,v)=-\frac{h_t^{-1} z_t^\prime u}{\sqrt{n}}-\frac{v^\prime}{\sqrt{n}}\sum_{j=1}^{p}\frac{\beta_{\tau 0}^{(j)} }{h_t}\frac{\partial h_{t-j}(\theta_0)}{\partial\theta}-R_t(u,v),\]
	where
	\[R_t(u,v)=\frac{v^\prime}{n}\sum_{j=1}^{p}\frac{u^{(j)}}{h_t}\frac{\partial h_{t-j}(\theta_0)}{\partial\theta}+\frac{v^\prime}{2n}\sum_{j=1}^{p}\frac{\beta_{\tau 0}^{(j)} +n^{-1/2}u^{(j)}}{h_t}\frac{\partial^2 h_{t-j}(\theta^*)}{\partial\theta\partial\theta^\prime}v,\]
	with $\theta^*$  between $\theta_0$ and $\theta_0+n^{-1/2}v$. Then, by \eqref{deq7},  Assumption \ref{assum2}, Lemma \ref{lem2} and the ergodic theorem, we can show that
	\begin{align*}
		&\sup_{\|u\|, \|v\|\leq M}\left |\frac{1}{\sqrt{n}}\sum_{t=k+1}^{n}\big[F(b_{\tau})- F(b_t(u,v))\big]|\varepsilon_{t-k, \tau}|+f(b_{\tau})\left (d_{1k}^\prime u+b_{\tau} d_{2k}^\prime v\right )\right |\\
		&\hspace{5mm}\leq f(b_{\tau})\sup_{\|u\|, \|v\|\leq M}\left |\frac{1}{\sqrt{n}}\sum_{t=k+1}^{n}[b_{\tau}- b_t(u,v)]|\varepsilon_{t-k, \tau}|+d_{1k}^\prime u+b_{\tau} d_{2k}^\prime v\right | \\
		&\hspace{10mm}+\frac{1}{2}\sup_{x\in \mathbb{R}}|\dot{f}(x)| \frac{1}{\sqrt{n}}\sum_{t=k+1}^{n}\sup_{\|u\|, \|v\|\leq M}|b_{\tau}- b_t(u,v)|^2|\varepsilon_{t-k, \tau}|\\
		&\hspace{5mm}=o_p(1).
	\end{align*}
	This together with \eqref{deq13} implies \eqref{deq14}, and therefore, \eqref{E1} holds.
	
	Next, we consider $\sum_{t=k+1}^{n}\mathcal{E}_{2nt}$. Observe that
	\[ \varepsilon_{t,\tau}-\widehat{\varepsilon}_{t,\tau}=\zeta_{1nt}(\widehat{\theta}_{\tau n}, \widetilde{\theta}_n)+\zeta_{2nt}(\widehat{\theta}_{\tau n}, \widetilde{\theta}_n), \]
	where
	\[\zeta_{1nt}(\theta_{\tau}, \theta)=\frac{y_t-\theta_{\tau 0}^\prime z_t}{h_t}-\frac{y_t-\theta_{\tau}^\prime z_t(\theta)}{h_t(\theta)}\hspace{5mm}\text{and}\hspace{5mm}
	\zeta_{2nt}(\theta_{\tau}, \theta)=\frac{y_t-\theta_{\tau}^\prime z_t(\theta)}{h_t(\theta)}-\frac{y_t-\theta_{\tau}^\prime \widetilde{z}_t(\theta)}{\widetilde{h}_t(\theta)}.\]
	Then, similar to the decompositions in \eqref{L1_L2}, \eqref{L1a} and \eqref{L2b}, by using the identity in \eqref{identity}, it can be verified that
	\begin{equation}\label{E2}
		\sum_{t=k+1}^{n}\mathcal{E}_{2nt}=\sum_{t=1}^{n-k}Z_{1nt}(\widehat{\theta}_{\tau n}, \widetilde{\theta}_n)+\sum_{t=1}^{n-k}Z_{2nt}(\widehat{\theta}_{\tau n}, \widetilde{\theta}_n)+\sum_{t=1}^{n-k}Z_{3nt}(\widehat{\theta}_{\tau n}, \widetilde{\theta}_n),
	\end{equation}
	where
	\begin{equation*}
		\begin{split}
			Z_{1nt}(\theta_{\tau}, \theta)&=\frac{\psi_{\tau}(\varepsilon_{t+k, \tau})}{\sqrt{n}}\bigg \{-\zeta_{2nt}(\theta_{\tau}, \theta)[1-2I(\varepsilon_{t}<b_{\tau})]+2\int_{\zeta_{1nt}(\theta_{\tau}, \theta)}^{\zeta_{1nt}(\theta_{\tau}, \theta)+\zeta_{2nt}(\theta_{\tau}, \theta)}
			I_t(s)ds\bigg\},\\
			Z_{2nt}(\theta_{\tau}, \theta)&= -\frac{ \psi_{\tau}(\varepsilon_{t+k, \tau})}{\sqrt{n}}\zeta_{1nt}(\theta_{\tau}, \theta)[1-2I(\varepsilon_{t}<b_{\tau})], \hspace{5mm}\text{and}\\
			Z_{3nt}(\theta_{\tau}, \theta)&=\frac{2\psi_{\tau}(\varepsilon_{t+k, \tau})}{\sqrt{n}}\int_{0}^{\zeta_{1nt}(\theta_{\tau}, \theta)}
			I_t(s)ds,
		\end{split}
	\end{equation*}
	with $I_t(s)=I(\varepsilon_{t,\tau}\leq s)-I(\varepsilon_{t,\tau}\leq 0)$. For any $M>0$, let $\Theta_{\tau n}=\Theta_{\tau n}(M)=\{\theta_{\tau}:  \|\theta_{\tau}-\theta_{\tau 0}\|\leq n^{-1/2}M, \theta_{\tau}/b_{\tau}\in\Theta\}$.
	Note that $\zeta_{2nt}(\theta_{\tau}, \theta)=\widetilde{h}_t^{-1}(\theta)\theta_{\tau}^\prime[\widetilde{z}_t(\theta)-z_t(\theta)]+[h_t^{-1}(\theta)-\widetilde{h}_t^{-1}(\theta)][y_t-\theta_{\tau}^\prime z_t(\theta)]$. Then,
	similar to \eqref{beq12}, \eqref{A1nt} and \eqref{B1nt},  by Lemma \ref{lem3}, it can be shown that
	\begin{align*}
		\sup_{\theta_{\tau}\in \Theta_{\tau n}, \,\theta \in \Theta_n}|\zeta_{2nt}(\theta_{\tau}, \theta)|\leq&\; \frac{1}{\underline{w}}\sum_{j=1}^{p}\sup_{\theta_{\tau}\in \Theta_{\tau n}}|\beta_{\tau}^{(j)}|\sup_{\theta \in \Theta} |\widetilde{h}_{t-j}(\theta) - h_{t-j}(\theta)|\\
		&+\frac{1}{\underline{w}^2}\sup_{\theta \in \Theta} |\widetilde{h}_{t}(\theta) - h_{t}(\theta)|\sup_{\theta_{\tau}\in \Theta_{\tau n}, \,\theta \in \Theta_n}|y_t-\theta_{\tau}^\prime z_t(\theta)|\\
		\leq &\;  C(M)\rho^t\zeta \bigg[1+ \sup_{\theta_{\tau}\in \Theta_{\tau n}, \,\theta \in \Theta_n}|y_t-\theta_{\tau}^\prime z_t(\theta)| \bigg].
	\end{align*}
	Consequently, it follows from Lemma \ref{lem2} that
	\[ \sup_{\theta_{\tau}\in \Theta_{\tau n}, \,\theta \in \Theta_n} \left |\sum_{t=1}^{n-k}Z_{1nt}(\theta_{\tau}, \theta)\right | \leq \frac{3}{\sqrt{n}}\sum_{t=1}^{n-k}\sup_{\theta_{\tau}\in \Theta_{\tau n}, \,\theta \in \Theta_n}|\zeta_{2nt}(\theta_{\tau}, \theta)|=o_p(1),\]
	which, together with  $\sqrt{n}(\widehat{\theta}_{\tau n}-\theta_{\tau0})=O_p(1)$ and $\sqrt{n}(\widetilde{\theta}_{n}-\theta_0)=O_p(1)$, yields
	\begin{equation}\label{z1}
		\sum_{t=1}^{n-k}Z_{1nt}(\widehat{\theta}_{\tau n}, \widetilde{\theta}_n)=o_p(1).
	\end{equation}
	Applying the second-order Taylor expansion to $h_t^{-1}(\theta)$, and the first and second-order Taylor expansions to $\theta_{\tau}^\prime z_t(\theta)$ respectively,  similar to \eqref{beq15}, it can be verified that
	\begin{equation}\label{deq16}
		\zeta_{1nt}(\theta_{\tau}, \theta)=\zeta_{3nt}(\theta_{\tau}, \theta) + \zeta_{4nt}(\theta_{\tau}, \theta),
	\end{equation}
	where
	\begin{align*}
		\zeta_{3nt}(\theta_{\tau}, \theta)=&(\theta_{\tau}-\theta_{\tau 0})^\prime\frac{ z_t}{h_t}+(\theta-\theta_0)^\prime \sum_{j=1}^{p}\frac{\beta_{\tau 0}^{(j)}}{h_t} \frac{\partial h_{t-j}(\theta_0)}{\partial\theta}+ (\theta-\theta_0)^\prime \frac{\varepsilon_t-b_{\tau}}{h_t}\frac{\partial h_{t}(\theta_0)}{\partial\theta},\\
		\zeta_{4nt}(\theta_{\tau}, \theta)=&(\theta-\theta_0)^\prime \sum_{j=1}^{p}\frac{\beta_{\tau }^{(j)}-\beta_{\tau 0}^{(j)}}{h_t} \frac{\partial h_{t-j}(\theta_2^*)}{\partial\theta}+\frac{(\theta-\theta_0)^\prime}{2}\sum_{j=1}^{p}\frac{\beta_{\tau2}^{*(j)}}{h_t} \frac{\partial^2 h_{t-j}(\theta^*_2)}{\partial\theta \partial\theta^\prime}(\theta-\theta_0)\\
		&-\frac{(\theta-\theta_0)^\prime}{h_t}\frac{\partial h_{t}(\theta_0)}{\partial\theta}\bigg[ \frac{z_t^\prime(\theta_1^*)}{h_t}(\theta_{\tau}-\theta_{\tau 0})+\sum_{j=1}^{p}\frac{\beta_{\tau 1}^{*(j)}}{h_t} \frac{\partial h_{t-j}(\theta_1^*)}{\partial\theta^\prime}(\theta-\theta_0)  \bigg]\\
		&-\frac{y_t-\theta_{\tau}^\prime z_t(\theta)}{h_t(\theta_3^*)}\frac{(\theta-\theta_0)^\prime}{2} \bigg[ \frac{2}{h_t^2(\theta_3^*)}\frac{\partial h_{t}(\theta_3^*)}{\partial\theta} \frac{\partial h_{t}(\theta_3^*)}{\partial\theta^\prime} -\frac{1}{h_t(\theta_3^*)}\frac{\partial^2 h_{t}(\theta_3^*)}{\partial\theta\partial\theta^\prime} \bigg](\theta-\theta_0),
	\end{align*}
	with $\theta_1^*, \theta_2^*$ and $\theta_3^*$ all lying between $\theta_0$ and $\theta$, and $\beta_{\tau 1}^{*(j)}$ and $\beta_{\tau 2}^{*(j)}$  both between $\beta_{\tau0}^{(j)}$ and $\beta_{\tau}^{(j)}$.  Then, similar to \eqref{beq4} and \eqref{beq6}, by Lemma \ref{lem2} and the ergodic theorem, together with $\sqrt{n}(\widehat{\theta}_{\tau n}-\theta_{\tau0})=O_p(1)$ and $\sqrt{n}(\widetilde{\theta}_{n}-\theta_0)=O_p(1)$, it can be shown that
	\[\frac{1}{\sqrt{n}}\sum_{t=1}^{n-k}\psi_{\tau}(\varepsilon_{t+k, \tau})\zeta_{3nt}(\widehat{\theta}_{\tau n}, \widetilde{\theta}_n) [1-2I(\varepsilon_{t}<b_{\tau})]=o_p(1),\]
	and
	\begin{align*}
		E \sup_{\theta_{\tau}\in \Theta_{\tau n}, \,\theta \in \Theta_n} & \left |\frac{1}{\sqrt{n}}\sum_{t=1}^{n-k}\psi_{\tau}(\varepsilon_{t+k, \tau})\zeta_{4nt}(\theta_{\tau}, \theta)[1-2I(\varepsilon_{t}<b_{\tau})]\right |\\
		\leq&\; \frac{1}{\sqrt{n}}\sum_{t=1}^{n-k}E\sup_{\theta_{\tau}\in \Theta_{\tau n}, \,\theta \in \Theta_n}|\zeta_{4nt}(\theta_{\tau}, \theta)|=O(n^{-1/2}),
	\end{align*}
	which implies
	\begin{equation}\label{z2}
		\sum_{t=1}^{n-k}Z_{2nt}(\widehat{\theta}_{\tau n}, \widetilde{\theta}_n)=o_p(1).
	\end{equation}
	Similarly, using the Taylor expansion in \eqref{deq16}, together with Lemma \ref{lem2} and Assumption \ref{assum2},   we can show that
	\begin{align*}
		& E\sup_{\theta_{\tau}\in \Theta_{\tau n}, \,\theta \in \Theta_n}  \left |\sum_{t=1}^{n-k}Z_{3nt}(\theta_{\tau}, \theta)  \right | \\
		& \hspace{5mm} \leq \frac{2}{\sqrt{n}}E\sum_{t=1}^{n-k} \sup_{\theta_{\tau}\in \Theta_{\tau n}, \,\theta \in \Theta_n} |\zeta_{1nt}(\theta_{\tau}, \theta)| I\bigg(|\varepsilon_t-b_{\tau}|\leq\sup_{\theta_{\tau}\in \Theta_{\tau n}, \,\theta \in \Theta_n} |\zeta_{1nt}(\theta_{\tau}, \theta)|\bigg)\\
		& \hspace{5mm} \leq \frac{4\sup_{x\in\mathbb{R}}f(x)}{\sqrt{n}}\sum_{t=1}^{n-k}E\bigg(\sup_{\theta_{\tau}\in \Theta_{\tau n}, \,\theta \in \Theta_n} |\zeta_{1nt}(\theta_{\tau}, \theta)|\bigg)^2 = O(n^{-1/2}),
	\end{align*}
	and as a result,
	\begin{equation}\label{z3}
		\sum_{t=1}^{n-k}Z_{3nt}(\widehat{\theta}_{\tau n}, \widetilde{\theta}_n)=o_p(1).
	\end{equation}
	Combining \eqref{E2}, \eqref{z1}, \eqref{z2} and \eqref{z3}, we have
	\begin{equation}\label{E2b}
		\sum_{t=k+1}^{n}\mathcal{E}_{2nt}=o_p(1).
	\end{equation}
	
	Now we consider $\sum_{t=k+1}^{n}\mathcal{E}_{3nt}$. Similar to the proof of \eqref{E1}, for any $u, v\in\mathbb{R}^{p+q+1}$, define $\varphi_t(u,v)=\{I(\varepsilon_t<b_{\tau})-I[\varepsilon_t<\widetilde{b}_t(u,v) ]\}\left [|\widetilde{\varepsilon}_{t-k,\tau}(u,v)|-|\varepsilon_{t-k,\tau}|\right ]$, where
	$\widetilde{\varepsilon}_{t,\tau}(u,v)=\left [y_t-(\theta_{\tau 0}+n^{-1/2}u)^\prime\widetilde{z}_t(\theta_0+n^{-1/2}v)\right ]\widetilde{h}_t^{-1}(\theta_0+n^{-1/2}v)$.  Then, for any $M>0$, we can readily verify that
	\[ \sup_{\|u\|, \|v\|\leq M}\left | \frac{1}{\sqrt{n}}\sum_{t=k+1}^{n} \{\varphi_t(u,v)-E[\varphi_t(u,v)|\mathcal{F}_{t-1}]\} \right |=o_p(1)\]
	and
	\[ \sup_{\|u\|, \|v\|\leq M}\left | \frac{1}{\sqrt{n}}\sum_{t=k+1}^{n} E[\varphi_t(u,v)|\mathcal{F}_{t-1}] \right |=o_p(1),\]
	which yields
	\begin{equation}\label{E3}
		\sum_{t=k+1}^{n}\mathcal{E}_{3nt}=o_p(1).
	\end{equation}
	Therefore, combining \eqref{deq1}, \eqref{E1}, \eqref{E2b} and \eqref{E3}, we have
	\begin{equation}\label{deq17}
		\begin{split}
			\frac{1}{\sqrt{n}}\sum_{t=k+1}^{n} \psi_{\tau}(\widehat{\varepsilon}_{t, \tau})|\widehat{\varepsilon}_{t-k, \tau}|
			=& \; \frac{1}{\sqrt{n}}\sum_{t=k+1}^{n}\psi_{\tau}(\varepsilon_{t, \tau})|\varepsilon_{t-k, \tau}|\\
			& -f(b_{\tau})\left [d_{1k}^\prime \sqrt{n}(\widehat{\theta}_{\tau n}-\theta_{\tau0})+b_{\tau} d_{2k}^\prime \sqrt{n}(\widetilde{\theta}_{n}-\theta_0)\right ]+o_p(1).
		\end{split}
	\end{equation}
	
	Finally, by the law of large numbers and a proof similar to that for \eqref{E2}, we can show that
	\[ |\widehat{\mu}_{a,\tau}-\mu_{a,\tau}|=\left |\frac{1}{n}\sum_{t=1}^{n}\left (|\widehat{\varepsilon}_{t, \tau}|-|\varepsilon_{t, \tau}|\right )\right | +o_p(1) \leq  \frac{1}{n}\sum_{t=1}^{n}|\widehat{\varepsilon}_{t, \tau}-\varepsilon_{t, \tau}| +o_p(1) =o_p(1),\]
	and then,
	\begin{align*}
		\widehat{\sigma}^2_{a,\tau}&=\frac{1}{n}\sum_{t=1}^{n}(|\widehat{\varepsilon}_{t, \tau}|-\widehat{\mu}_{a,\tau})^2=\frac{1}{n}\sum_{t=1}^{n}\widehat{\varepsilon}_{t, \tau}^2 - \mu_{a,\tau}^2+o_p(1)\\
		&=\frac{1}{n}\sum_{t=1}^{n}(\widehat{\varepsilon}_{t, \tau}^2-\varepsilon_{t,\tau}^2)+\sigma_{a,\tau}^2+o_p(1)\\
		&=\sigma_{a,\tau}^2+o_p(1),
	\end{align*}
	which, together with \eqref{deq17},  \eqref{qmle} and \eqref{bahadur2}, yields
	\begin{equation}\label{}
		\begin{split}
			r_{k, \tau}=&\; \frac{1}{\sqrt{(\tau-\tau^2)\sigma_{a, \tau}^2}}\cdot \frac{1}{n}\sum_{t=k+1}^{n}
			\bigg\{\psi_{\tau}(\varepsilon_{t, \tau})\left (|\varepsilon_{t-k, \tau}| -d_{1k}^\prime \Omega_2^{-1}\frac{z_t}{h_t}\right )\\
			&+ b_{\tau} f(b_{\tau})\left (d_{2k}^\prime-d_{1k}^\prime\Omega_2^{-1}\Gamma_2\right )J^{-1}\frac{1-|\varepsilon_t|}{h_t}\frac{\partial h_t(\theta_0)}{\partial\theta}\bigg\}+o_p(n^{-1/2}).
		\end{split}
	\end{equation}
	Consequently, for $R=(r_{1,\tau}, \dots, r_{K,\tau})^\prime$, we have
	\begin{equation}\label{}
		\begin{split}
			R=&\; \frac{1}{\sqrt{(\tau-\tau^2)\sigma_{a, \tau}^2}}\cdot \frac{1}{n}\sum_{t=k+1}^{n}
			\bigg\{\psi_{\tau}(\varepsilon_{t, \tau})\left (\epsilon_{t-1} -D_{1}\Omega_2^{-1}\frac{z_t}{h_t}\right )\\
			&+ b_{\tau} f(b_{\tau})\left (D_{2}-D_{1}\Omega_2^{-1}\Gamma_2\right )J^{-1}\frac{1-|\varepsilon_t|}{h_t}\frac{\partial h_t(\theta_0)}{\partial\theta}\bigg\}+o_p(n^{-1/2}),
		\end{split}
	\end{equation}
	where $\epsilon_{t-1}=( |\varepsilon_{t-1,\tau}|, \dots, |\varepsilon_{t-K,\tau}|)^\prime$ and $D_i=(d_{i1}, \dots, d_{iK})^\prime$ for $i=1$ and 2.  Thus, we complete the proof by applying the central limit theorem and the Cram\'{e}r-Wold device.
\end{proof}

\begin{proof}[Proof of Theorem \ref{thm4}]	
Similar to \eqref{deq1}, we have
\begin{align}\label{deq1add}
\begin{split}
\frac{1}{\sqrt{n}}&\sum_{t=k+1}^{n} \omega_t \psi_{\tau}(\widehat{\varepsilon}_{t, \tau}^*)|\widehat{\varepsilon}_{t-k, \tau}^*| \\
 &=\frac{1}{\sqrt{n}}\sum_{t=k+1}^{n}\omega_t\psi_{\tau}(\varepsilon_{t, \tau})|\varepsilon_{t-k, \tau}| +\sum_{t=k+1}^{n}\mathcal{E}_{1nt}^* +\sum_{t=k+1}^{n}\mathcal{E}_{2nt}^*+\sum_{t=k+1}^{n}\mathcal{E}_{3nt}^*,
\end{split}
\end{align}	
where
\begin{equation*}
\begin{split}
\mathcal{E}_{1nt}^*&=n^{-1/2}\omega_t[\psi_{\tau}(\widehat{\varepsilon}_{t, \tau}^*)-\psi_{\tau}(\varepsilon_{t,\tau})]|\varepsilon_{t-k, \tau}|,\hspace{3mm}
			 \mathcal{E}_{2nt}^*=n^{-1/2}\omega_t \psi_{\tau}(\varepsilon_{t,\tau})(|\widehat{\varepsilon}_{t-k,\tau}^*|-|\varepsilon_{t-k, \tau}|),\hspace{3mm}\text{and}\\
			 \mathcal{E}_{3nt}^*&=n^{-1/2}\omega_t[\psi_{\tau}(\widehat{\varepsilon}_{t,\tau}^*) -\psi_{\tau}(\varepsilon_{t, \tau})](|\widehat{\varepsilon}_{t-k, \tau}^*|-|\varepsilon_{t-k, \tau}|).
		\end{split}
	\end{equation*}
Note that, from \eqref{proof1} and \eqref{proof2}, $\sqrt{n}(\widetilde{\theta}_n^{*}-\theta_0)=O_p^*(1)$ and $\sqrt{n}(\widehat{\theta}_{\tau n}^{*}-\theta_{\tau 0})=O_p^*(1)$.
As a result, by methods similar to \eqref{E1}, \eqref{E2b} and \eqref{E3}, respectively, we can show that
\begin{equation*}
		\sum_{t=k+1}^{n}\mathcal{E}_{1nt}^*=-f(b_{\tau})\left [d_{1k}^\prime \sqrt{n}(\widehat{\theta}_{\tau n}^*-\theta_{\tau0})+b_{\tau} d_{2k}^\prime \sqrt{n}(\widetilde{\theta}_{n}^*-\theta_0)\right ]+o_p^*(1),
\end{equation*}
and
\[
\sum_{t=k+1}^{n}\mathcal{E}_{int}^*=o_p^*(1),\hspace{5mm} i=2 \text{ and } 3,
\]
	where $d_{1k}=E(h_t^{-1}|\varepsilon_{t-k,\tau}|z_t)$ and $d_{2k}=E(h_t^{-1}|\varepsilon_{t-k, \tau}|\sum_{j=1}^{p}\beta_{0j}{\partial h_{t-j}(\theta_0)}/{\partial\theta})$ are defined as in \eqref{E1}.
This, in conjunction with \eqref{deq1add} and \eqref{deq17}, yields the Bahadur representation of
\begin{align*}
\frac{1}{\sqrt{n}}\sum_{t=k+1}^{n} &\omega_t \psi_{\tau}(\widehat{\varepsilon}_{t, \tau}^*)|\widehat{\varepsilon}_{t-k, \tau}^*| -\frac{1}{\sqrt{n}}\sum_{t=k+1}^{n} \psi_{\tau}(\widehat{\varepsilon}_{t, \tau})|\widehat{\varepsilon}_{t-k, \tau}|\\
&=\frac{1}{\sqrt{n}}\sum_{t=k+1}^{n}(\omega_t-1)\psi_{\tau}(\varepsilon_{t, \tau})|\varepsilon_{t-k, \tau}| \\
&\hspace{9mm}-f(b_{\tau})\left [d_{1k}^\prime \sqrt{n}(\widehat{\theta}_{\tau n}^*-\widehat{\theta}_{\tau n})+b_{\tau} d_{2k}^\prime \sqrt{n}(\widetilde{\theta}_{n}^*-\widetilde{\theta}_{n})\right ]+o_p^*(1),
\end{align*}
and hence
	\begin{equation*}
		\begin{split}
			R^*-R=&\; \frac{1}{\sqrt{(\tau-\tau^2)\sigma_{a, \tau}^2}}\cdot \frac{1}{n}\sum_{t=k+1}^{n}
			(\omega_t-1)\bigg\{\psi_{\tau}(\varepsilon_{t, \tau})\left (\epsilon_{t-1} -D_{1}\Omega_2^{-1}\frac{z_t}{h_t}\right )\\
			&+ b_{\tau} f(b_{\tau})\left (D_{2}-D_{1}\Omega_2^{-1}\Gamma_2\right )J^{-1}\frac{1-|\varepsilon_t|}{h_t}\frac{\partial h_t(\theta_0)}{\partial\theta}\bigg\}+o_p^*(n^{-1/2}),
		\end{split}
	\end{equation*}
	where $\epsilon_{t-1}=( |\varepsilon_{t-1,\tau}|, \dots, |\varepsilon_{t-K,\tau}|)^\prime$ and $D_i=(d_{i1}, \dots, d_{iK})^\prime$ for $i=1$ and 2.  Thus, we complete the proof by applying Lindeberg's central limit theorem and the Cram\'{e}r-Wold device.
\end{proof}
\begin{proof}[Proof of Corollary \ref{cor1}] 	
	The proof follows the same lines as that of Theorem \ref{thm1}, while the corresponding $L_{1n}(u)$ and  $L_{2n}(u)$ are defined with $\widetilde{h}_t^{-1}$ replaced by one; consequently, all the $A_{int}(\theta)$'s and $B_{int}(\theta)$'s are defined with all $\widetilde{h}_t^{-1}(\theta)$, $h_t^{-1}(\theta)$ and $h_t^{-1}$ replaced by one. Note that without these denominators, Lemma \ref{lem2} cannot be applied as in the proof of Theorem \ref{thm1} in some intermediate steps, and additional moment conditions on $x_t$ will be needed. The highest moment condition,  $E|x_t|^{4+\iota_0}$ for some $\iota_0>0$, is required for the proof of the counterpart of \eqref{forCor1}, where, correspondingly, $\eta_t(v)=\int_{0}^{\xi_{1nt}(\theta_0+n^{-1/2}v)} I_t^*(s)ds$, with $\xi_{1nt}$ and $ I_t^*(s)$ defined as in the proof of Theorem \ref{thm1}. The corresponding proof is straightforward by the H\"{o}lder inequality.
\end{proof}
\begin{proof}[Proof of Corollary \ref{cor2} and Equation \eqref{quantx}] Since $\sqrt{n}(\widetilde{\theta }_{n}-\theta _{0})=O_p(1)$ and $\sqrt{n}(\widehat{\theta }_{\tau n}-\theta _{\tau 0})=O_p(1)$, Corollary \ref{cor2} follows directly from Lemma \ref{lem3} and the Taylor expansion.

Moreover, it can be readily shown that the sequence $\{X_n\}$ with $X_n=u_{n+1}^{\prime }\sqrt{n}(\widetilde{\theta }_{n}-\theta _{0})+z_{n+1}^{\prime}\sqrt{n}(\widehat{\theta }_{\tau n}-\theta _{\tau 0})$ is uniformly tight, which, combined with Corollary \ref{cor2}, implies that $o_p(|\widehat{Q}_{\tau}(y_{n+1}|\mathcal{F}_{n})-Q_{\tau}(y_{n+1}|\mathcal{F}_{n})|)=o_p(n^{-1/2})$. Note that $b_{\tau}\neq 0$ if and only if $Q_{\tau}(y_{n+1}|\mathcal{F}_{n})=\theta_{\tau 0}^\prime z_{n+1}=b_{\tau}h_{n+1}\neq0$, since $h_{n+1}\geq\underline{w}>0$. If $b_{\tau}\neq 0$, then $T^{-1}(\cdot)$ is differentiable at $Q_{\tau}(y_{n+1}|\mathcal{F}_{n})$, and hence
\begin{align*}
& T^{-1}[\widehat{Q}_{\tau}(y_{n+1}|\mathcal{F}_{n})]-T^{-1}[Q_{\tau}(y_{n+1}|\mathcal{F}_{n})]\\
&\hspace{5mm}=\frac{dT^{-1}(x)}{dx}\bigg|_{x=Q_{\tau}(y_{n+1}|\mathcal{F}_{n})}\left [\widehat{Q}_{\tau}(y_{n+1}|\mathcal{F}_{n})-Q_{\tau}(y_{n+1}|\mathcal{F}_{n})\right ]+o_p(n^{-1/2})\\
&\hspace{5mm}= \frac{1}{2\sqrt{|b_{\tau}h_{n+1}|}}\left [u_{n+1}^{\prime }(\widetilde{\theta }_{n}-\theta _{0})+z_{n+1}^{\prime}(\widehat{\theta }_{\tau n}-\theta _{\tau 0})\right ]+o_p(n^{-1/2}).
\end{align*}	
Since $\widehat{Q}_{\tau}(x_{n+1}|\mathcal{F}_{n})=T^{-1}[\widehat{Q}_{\tau}(y_{n+1}|\mathcal{F}_{n})]$ and $Q_{\tau}(x_{n+1}|\mathcal{F}_{n})=T^{-1}[Q_{\tau}(y_{n+1}|\mathcal{F}_{n})]$, we complete the proof of \eqref{quantx}.
\end{proof}	
\begin{proof}[Proof of Corollary \ref{cor3}] By methods similar to the proofs of Theorem \ref{thm3} and Corollary \ref{cor2}, this corollary follows.
\end{proof}	

\bibliographystyle{plain}
\bibliography{Quantile}

\clearpage
\newpage
\begin{landscape}
\begin{table}
\begin{center}
\caption{\label{table1a}Biases ($\times10$) and MSEs for in-sample and out-of-sample conditional quantile estimates at $\tau=0.05$, for $\alpha_0=0.1$, $\alpha_1=0.8$, $\beta_1=0.15$, and normally or Student's $t_5$ distributed innovations.}\vspace{5mm}
\begin{tabularx}{0.85\linewidth}{c@{\hskip 8mm}l@{\hskip 4mm}*{2}{Y}c*{2}{Y}c@{\hskip 3mm}*{2}{Y}c*{2}{Y}}			 
\hline\hline
&&\multicolumn{5}{c}{Normal distribution}&&\multicolumn{5}{c}{Student's $t_5$ distribution}\\
\cline{3-7}\cline{9-13}
&&\multicolumn{2}{c}{Bias}&&\multicolumn{2}{c}{MSE}&&\multicolumn{2}{c}{Bias}&&\multicolumn{2}{c}{MSE}\\
\cline{3-4}\cline{6-7}\cline{9-10}\cline{12-13}
$n$ &  & \multicolumn{1}{c}{In} & \multicolumn{1}{c}{Out} & & \multicolumn{1}{c}{In} & \multicolumn{1}{c}{Out} & & \multicolumn{1}{c}{In} & \multicolumn{1}{c}{Out} & & \multicolumn{1}{c}{In} & \multicolumn{1}{c}{Out}\\\hline
\noalign{\vskip 1mm}
200 & Hybrid & -0.028 & -0.020 &  & 0.121 & 0.088 &  & -0.231 & -0.094 &  & 0.194 & 0.175\\
& $\text{QGARCH}_1$ & 0.293 & 0.130 &  & 0.390 & 0.275 &  & 0.131 & 0.115 &  & 0.472 & 0.417\\
& $\text{QGARCH}_2$ & 0.300 & 0.134 &  & 0.368 & 0.319 &  & 0.137 & 0.066 &  & 0.475 & 0.638\\
& CAViaR & 0.165 & 0.060 &  & 0.162 & 0.147 &  & -0.060 & -0.035 &  & 0.291 & 0.270\\
& RiskM & -1.266 & -1.572 &  & 1.633 & 1.261 &  & -1.491 & -1.818 &  & 1.338 & 1.324\vspace{1.5mm}\\
500 & Hybrid & -0.017 & 0.004 &  & 0.064 & 0.046 &  & -0.079 & -0.070 &  & 0.092 & 0.049\\
& $\text{QGARCH}_1$ & 0.201 & 0.205 &  & 0.354 & 0.139 &  & 0.132 & 0.077 &  & 0.430 & 0.134\\
& $\text{QGARCH}_2$ & 0.205 & 0.219 &  & 0.358 & 0.137 &  & 0.148 & 0.060 &  & 0.447 & 0.134\\
& CAViaR & 0.059 & 0.043 &  & 0.128 & 0.066 &  & 0.009 & 0.014 &  & 0.273 & 0.070\\
& RiskM & -1.591 & -1.585 &  & 2.282 & 1.467 &  & -1.615 & -1.745 &  & 1.603 & 1.162\vspace{1.5mm}\\
1000 & Hybrid & -0.001 & -0.007 &  & 0.028 & 0.023 &  & -0.040 & -0.047 &  & 0.048 & 0.032\\
& $\text{QGARCH}_1$ & 0.153 & 0.090 &  & 0.279 & 0.173 &  & 0.127 & 0.557 &  & 0.414 & 12.911\\
& $\text{QGARCH}_2$ & 0.152 & 0.110 &  & 0.271 & 0.147 &  & 0.130 & 0.500 &  & 0.422 & 10.190\\
& CAViaR & 0.037 & 0.026 &  & 0.075 & 0.039 &  & 0.001 & 0.057 &  & 0.198 & 0.205\\
& RiskM & -1.566 & -1.700 &  & 1.951 & 1.472 &  & -1.637 & -1.492 &  & 1.931 & 2.897\\[1mm]
\hline
\end{tabularx}
\end{center}
\end{table}

\begin{table}
\begin{center}
\caption{\label{table1b}Biases ($\times10$) and MSEs for in-sample and out-of-sample conditional quantile estimates at $\tau=0.05$, for $\alpha_0=0.1$, $\alpha_1=0.15$, $\beta_1=0.8$, and normally or Student's $t_5$ distributed innovations.}\vspace{5mm}	
\begin{tabularx}{0.85\linewidth}{c@{\hskip 8mm}l@{\hskip 4mm}*{2}{Y}c*{2}{Y}c@{\hskip 3mm}*{2}{Y}c*{2}{Y}}	 
\hline\hline
&&\multicolumn{5}{c}{Normal distribution}&&\multicolumn{5}{c}{Student's $t_5$ distribution}\\
\cline{3-7}\cline{9-13}
&&\multicolumn{2}{c}{Bias}&&\multicolumn{2}{c}{MSE}&&\multicolumn{2}{c}{Bias}&&\multicolumn{2}{c}{MSE}\\
\cline{3-4}\cline{6-7}\cline{9-10}\cline{12-13}
$n$ &  & \multicolumn{1}{c}{In} & \multicolumn{1}{c}{Out} & & \multicolumn{1}{c}{In} & \multicolumn{1}{c}{Out} & & \multicolumn{1}{c}{In} & \multicolumn{1}{c}{Out} & & \multicolumn{1}{c}{In} & \multicolumn{1}{c}{Out}\\\hline
\noalign{\vskip 1mm}
200 & Hybrid & -0.193 & -0.268 &  & 0.193 & 0.207 &  & -0.593 & -0.726 &  & 0.401 & 0.461\\
& $\text{QGARCH}_1$ & -0.103 & -0.112 &  & 0.392 & 0.471 &  & -0.417 & -0.533 &  & 0.741 & 0.866\\
& $\text{QGARCH}_2$ & -0.075 & -0.012 &  & 0.350 & 0.422 &  & -0.333 & -0.360 &  & 0.660 & 0.835\\
& CAViaR & 0.129 & 0.218 &  & 0.157 & 0.194 &  & -0.143 & -0.079 &  & 0.317 & 0.365\\
& RiskM & 0.466 & -0.061 &  & 0.150 & 0.142 &  & -0.460 & -1.017 &  & 0.270 & 0.272\vspace{1.5mm}\\
500 & Hybrid & -0.027 & 0.034 &  & 0.078 & 0.082 &  & -0.166 & -0.105 &  & 0.145 & 0.166\\
&  $\text{QGARCH}_1$  & -0.061 & 0.071 &  & 0.231 & 0.266 &  & -0.166 & -0.102 &  & 0.435 & 0.561\\
&  $\text{QGARCH}_2$ & -0.017 & 0.085 &  & 0.173 & 0.191 &  & -0.129 & -0.076 &  & 0.342 & 0.613\\
& CAViaR & 0.099 & 0.181 &  & 0.069 & 0.078 &  & 0.006 & 0.110 &  & 0.131 & 0.156\\
& RiskM & 0.249 & 0.167 &  & 0.132 & 0.128 &  & -0.580 & -0.581 &  & 0.236 & 0.207\vspace{1.5mm}\\
1000 & Hybrid & 0.002 & -0.006 &  & 0.038 & 0.041 &  & -0.084 & -0.172 &  & 0.077 & 0.132\\
& $\text{QGARCH}_1$ & -0.068 & -0.020 &  & 0.146 & 0.155 &  & -0.156 & -0.348 &  & 0.361 & 1.334\\
& $\text{QGARCH}_2$ & -0.020 & 0.010 &  & 0.097 & 0.103 &  & -0.100 & -0.298 &  & 0.259 & 1.254\\
& CAViaR & 0.066 & 0.073 &  & 0.034 & 0.038 &  & -0.001 & -0.001 &  & 0.092 & 0.085\\
& RiskM & 0.175 & 0.090 &  & 0.129 & 0.128 &  & -0.627 & -0.597 &  & 0.247 & 0.287\\[1mm]
\hline
\end{tabularx}
\end{center}
\end{table}
\end{landscape}

\begin{landscape}
\begin{table}
\begin{center}
\caption{\label{table2}Biases, ESDs and ASDs for the weighted estimator $\widehat{\theta}_{\tau n}$ at $\tau=0.1$ or 0.25, for normally or Student's $t_5$ distributed innovations, where $\text{ASD}_i$ corresponds to random weight $W_i$ for $i=1,2$ and 3. The notations $\alpha_0$, $\alpha_1$ and $\beta_1$ represent the corresponding elements of  $\widehat{\theta}_{\tau n}$.}\vspace{5mm}		
\begin{tabularx}{0.85\linewidth}{c@{\hskip 8mm}c@{\hskip 4mm}*{5}{Y}c*{5}{Y}}	
\hline\hline
&&\multicolumn{5}{c}{Normal distribution}&&\multicolumn{5}{c}{Student's $t_5$ distribution}\\
\cline{3-7}\cline{9-13}
$n$&&\multicolumn{1}{c}{Bias}&\multicolumn{1}{c}{ESD}&\multicolumn{1}{c}{$\text{ASD}_1$}&\multicolumn{1}{c}{$\text{ASD}_2$}&\multicolumn{1}{c}{$\text{ASD}_3$}&&\multicolumn{1}{c}{Bias}&\multicolumn{1}{c}{ESD}&\multicolumn{1}{c}{$\text{ASD}_1$}&\multicolumn{1}{c}{$\text{ASD}_2$}&\multicolumn{1}{c}{$\text{ASD}_3$}\\\hline
\noalign{\vskip 1mm}
& &\multicolumn{11}{c}{$\tau=0.1$}\\
500  & $\alpha_0$ & 0.000 & 0.447 & 0.507 & 0.514 & 0.509 &  & -0.019 & 0.426 & 0.646 & 0.626 & 0.589\\
&  $\alpha_1$  & 0.008 & 0.258 & 0.275 & 0.271 & 0.273 &  & -0.032 & 0.268 & 0.292 & 0.283 & 0.287\\
&  $\beta_1$   & -0.018 & 0.349 & 0.379 & 0.388 & 0.382 &  & 0.001 & 0.344 & 0.482 & 0.516 & 0.456\\
1000  & $\alpha_0$ & 0.001 & 0.329 & 0.344 & 0.346 & 0.345 &  & -0.014 & 0.291 & 0.351 & 0.332 & 0.332\\
&  $\alpha_1$  & 0.004 & 0.185 & 0.193 & 0.192 & 0.192 &  & -0.011 & 0.183 & 0.199 & 0.195 & 0.197\\
&  $\beta_1$   & -0.011 & 0.258 & 0.265 & 0.266 & 0.265 &  & -0.001 & 0.238 & 0.289 & 0.286 & 0.280\\
2000  & $\alpha_0$ & 0.004 & 0.229 & 0.241 & 0.242 & 0.241 &  & 0.000 & 0.203 & 0.240 & 0.220 & 0.220\\
&  $\alpha_1$  & 0.006 & 0.131 & 0.135 & 0.134 & 0.134 &  & -0.007 & 0.131 & 0.137 & 0.136 & 0.136\\
&  $\beta_1$   & -0.011 & 0.180 & 0.187 & 0.187 & 0.187 &  & -0.007 & 0.176 & 0.198 & 0.189 & 0.188\\[1mm]
& &\multicolumn{11}{c}{$\tau=0.25$}\\
500  & $\alpha_0$ & -0.005 & 0.199 & 0.216 & 0.216 & 0.216 &  & 0.001 & 0.145 & 0.214 & 0.212 & 0.190\\
&  $\alpha_1$  & -0.003 & 0.106 & 0.112 & 0.111 & 0.112 &  & -0.012 & 0.087 & 0.090 & 0.088 & 0.089\\
&  $\beta_1$   & -0.005 & 0.147 & 0.159 & 0.160 & 0.159 &  & -0.005 & 0.115 & 0.160 & 0.176 & 0.148\\
1000  & $\alpha_0$ & -0.004 & 0.145 & 0.148 & 0.148 & 0.148 &  & 0.000 & 0.099 & 0.114 & 0.110 & 0.109\\
&  $\alpha_1$  & -0.002 & 0.080 & 0.080 & 0.079 & 0.079 &  & -0.003 & 0.060 & 0.063 & 0.062 & 0.063\\
&  $\beta_1$   & -0.003 & 0.108 & 0.111 & 0.110 & 0.111 &  & -0.005 & 0.081 & 0.093 & 0.092 & 0.090\\
2000  & $\alpha_0$ & 0.000 & 0.103 & 0.103 & 0.103 & 0.103 &  & 0.001 & 0.071 & 0.081 & 0.073 & 0.073\\
&  $\alpha_1$  & -0.001 & 0.055 & 0.056 & 0.056 & 0.056 &  & -0.003 & 0.043 & 0.044 & 0.044 & 0.044\\
&  $\beta_1$   & -0.003 & 0.076 & 0.078 & 0.078 & 0.078 &  & -0.003 & 0.060 & 0.065 & 0.061 & 0.061\\[1mm]				 
\hline
\end{tabularx}
\end{center}
\end{table}

\begin{table}
\begin{center}
\caption{\label{table3}Biases ($\times 10$), ESDs ($\times 10$) and ASDs ($\times 10$) for the residual QACF $r_{k,\tau}$ at $\tau=0.1$ or 0.25 and $k=2, 4$ or 6, for normally or Student's $t_5$ distributed innovations, where $\text{ASD}_i$ corresponds to random weight $W_i$ for $i=1,2$ and 3.}\vspace{5mm}		 
\begin{tabularx}{0.85\linewidth}{c@{\hskip 8mm}c@{\hskip 4mm}*{5}{Y}c*{5}{Y}}	
\hline\hline
&&\multicolumn{5}{c}{Normal distribution}&&\multicolumn{5}{c}{Student's $t_5$ distribution}\\
\cline{3-7}\cline{9-13}
$n$& $k$ & \multicolumn{1}{c}{Bias}&\multicolumn{1}{c}{ESD}&\multicolumn{1}{c}{$\text{ASD}_1$}&\multicolumn{1}{c}{$\text{ASD}_2$}&\multicolumn{1}{c}{$\text{ASD}_3$}&& \multicolumn{1}{c}{Bias}&\multicolumn{1}{c}{ESD}&\multicolumn{1}{c}{$\text{ASD}_1$}&\multicolumn{1}{c}{$\text{ASD}_2$}&\multicolumn{1}{c}{$\text{ASD}_3$}  \\\hline
\noalign{\vskip 1mm}
&&\multicolumn{11}{c}{$\tau=0.1$}\\					
500 & 2 & 0.047 & 0.433 & 0.539 & 0.526 & 0.533 &  & 0.024 & 0.429 & 0.493 & 0.492 & 0.490\\
 & 4 & 0.057 & 0.453 & 0.541 & 0.532 & 0.536 &  & 0.032 & 0.426 & 0.482 & 0.485 & 0.483\\
 & 6 & 0.047 & 0.468 & 0.545 & 0.536 & 0.540 &  & 0.040 & 0.452 & 0.474 & 0.476 & 0.473\\
1000 & 2 & 0.016 & 0.304 & 0.342 & 0.338 & 0.340 &  & 0.005 & 0.304 & 0.323 & 0.326 & 0.324\\
 & 4 & 0.013 & 0.322 & 0.353 & 0.349 & 0.351 &  & 0.019 & 0.301 & 0.317 & 0.319 & 0.318\\
 & 6 & 0.021 & 0.320 & 0.356 & 0.353 & 0.354 &  & 0.000 & 0.321 & 0.324 & 0.327 & 0.325\\
2000 & 2 & 0.014 & 0.214 & 0.229 & 0.228 & 0.228 &  & 0.003 & 0.216 & 0.217 & 0.218 & 0.217\\
 & 4 & -0.003 & 0.215 & 0.237 & 0.236 & 0.237 &  & 0.005 & 0.220 & 0.220 & 0.221 & 0.220\\
 & 6 & 0.011 & 0.217 & 0.239 & 0.238 & 0.239 &  & 0.006 & 0.220 & 0.222 & 0.224 & 0.223\\[1mm]
&&\multicolumn{11}{c}{$\tau=0.25$}\\	
500 & 2 & 0.004 & 0.373 & 0.429 & 0.423 & 0.426 &  & -0.011 & 0.388 & 0.440 & 0.437 & 0.438\\
 & 4 & 0.030 & 0.421 & 0.465 & 0.463 & 0.465 &  & 0.008 & 0.438 & 0.460 & 0.461 & 0.461\\
 & 6 & 0.029 & 0.430 & 0.474 & 0.472 & 0.473 &  & 0.029 & 0.439 & 0.459 & 0.460 & 0.459\\
1000 & 2 & 0.004 & 0.267 & 0.288 & 0.286 & 0.287 &  & -0.013 & 0.284 & 0.302 & 0.301 & 0.301\\
 & 4 & 0.018 & 0.303 & 0.319 & 0.318 & 0.318 &  & 0.006 & 0.307 & 0.318 & 0.319 & 0.319\\
 & 6 & 0.022 & 0.313 & 0.325 & 0.325 & 0.326 &  & 0.006 & 0.321 & 0.321 & 0.322 & 0.322\\
2000 & 2 & 0.006 & 0.192 & 0.197 & 0.197 & 0.197 &  & -0.003 & 0.204 & 0.208 & 0.207 & 0.207\\
 & 4 & 0.002 & 0.208 & 0.220 & 0.220 & 0.220 &  & -0.001 & 0.223 & 0.221 & 0.221 & 0.221\\
 & 6 & 0.007 & 0.220 & 0.227 & 0.227 & 0.227 &  & 0.008 & 0.228 & 0.224 & 0.224 & 0.224\\[1mm]					 
\hline
\end{tabularx}
\end{center}
\end{table}
\end{landscape}

\begin{table}[tbp]
\caption{Rejection rates ($\times100$) of the test statistic $Q(K)$ for $K=6$
at the 5\% significance level, for normally or Student's $t_5$ distributed
innovations, where $\text{Q}_i$ denotes the test statistic based on random
weight $W_i$ for $i=1,2$ and 3.}
\label{table4}
\begin{center}
\vspace{5mm}
\begin{tabularx}{0.8\textwidth}{c@{\hskip 8mm}l@{\hskip 4mm}*{3}{Y}c*{3}{Y}}	
\hline\hline
&&\multicolumn{3}{c}{Normal distribution}&&\multicolumn{3}{c}{Student's $t_5$ distribution}\\
\cline{3-5}\cline{7-9}
$n$ & $d$ & \multicolumn{1}{c}{$Q_1$} & \multicolumn{1}{c}{$Q_2$} & \multicolumn{1}{c}{$Q_3$} && \multicolumn{1}{c}{$Q_1$} & \multicolumn{1}{c}{$Q_2$} & \multicolumn{1}{c}{$Q_3$} \\\hline
\noalign{\vskip 1mm}
&&\multicolumn{7}{c}{$\tau=0.1$}\\
500 & 0 & 4.0 & 3.2 & 3.6 &    & 3.7 & 3.1 & 3.7\\
 & 0.3 & 5.7 & 5.0 & 5.9 &    & 9.9 & 7.7 & 9.2\\
 & 0.6 & 20.9 & 18.5 & 20.7 &    & 29.5 & 28.6 & 29.4\\
1000 & 0 & 4.7 & 4.4 & 4.3 &    & 5.2 & 4.8 & 5.3\\
 & 0.3 & 17.3 & 16.0 & 17.2 &    & 22.2 & 20.6 & 21.4\\
 & 0.6 & 57.0 & 54.7 & 56.2 &    & 61.7 & 61.4 & 62.4\\
2000 & 0 & 4.9 & 4.4 & 4.5 &    & 5.5 & 5.2 & 5.2\\
 & 0.3 & 37.9 & 37.1 & 38.1 &    & 46.8 & 45.1 & 45.9\\
 & 0.6 & 89.4 & 89.3 & 89.9 &    & 91.4 & 90.9 & 91.1\\[1mm]
 &  & \multicolumn{7}{c}{$\tau=0.25$}\\
500 & 0 & 3.4 & 3.7 & 3.7 &    & 3.3 & 3.1 & 3.1\\
 & 0.3 & 6.5 & 5.9 & 6.2 &    & 5.7 & 5.3 & 5.7\\
 & 0.6 & 20.2 & 20.0 & 20.2 &    & 15.5 & 15.5 & 15.9\\
1000 & 0 & 4.3 & 4.2 & 4.3 &    & 4.6 & 4.6 & 4.3\\
 & 0.3 & 16.2 & 15.8 & 16.0 &    & 10.8 & 10.9 & 10.8\\
 & 0.6 & 46.6 & 47.2 & 46.9 &    & 32.3 & 32.0 & 32.1\\
2000 & 0 & 4.1 & 4.2 & 4.1 &    & 4.7 & 4.6 & 4.6\\
 & 0.3 & 36.6 & 36.5 & 35.5 &    & 29.0 & 29.1 & 28.9\\
 & 0.6 & 83.3 & 83.3 & 83.0 &    & 69.7 & 69.9 & 69.6\\[1mm]									
\hline
\end{tabularx}
\end{center}
\end{table}

\begin{landscape}
\begin{table}
\begin{center}
\caption{\label{table7}Empirical coverage rates ($\times 100$) for various conditional quantile estimation methods for $1$\% VaR and $5$\% VaR.}\vspace{5mm}		
\begin{tabularx}{0.85\linewidth}{l@{\hskip 8mm}*{2}{Z}c*{2}{Z}c*{2}{Z}c*{2}{Z}c*{2}{Z}}		
\hline\hline
& \multicolumn{2}{c}{2010 - 2011} & & \multicolumn{2}{c}{2012 - 2013} &  & \multicolumn{2}{c}{2014 - 2015} &  & \multicolumn{2}{c}{2016 - end} &  &  \multicolumn{2}{c}{Overall}\\
\cline{2-3}\cline{5-6}\cline{8-9}\cline{11-12}\cline{14-15}& \multicolumn{1}{c}{1\%} & \multicolumn{1}{c}{5\%} &  & \multicolumn{1}{c}{1\%} & \multicolumn{1}{c}{5\%} &  & \multicolumn{1}{c}{1\%} & \multicolumn{1}{c}{5\%} &  & \multicolumn{1}{c}{1\%} & \multicolumn{1}{c}{5\%} &  & \multicolumn{1}{c}{1\%} & \multicolumn{1}{c}{5\%}\\\hline
\noalign{\vskip 1mm}& \multicolumn{14}{c}{S\&P 500}\\
Hybrid & 1.19 & 4.76 &  & 0.60 & 3.39 &  & 1.19 & 4.37 &  & 0.80 & 3.20 &  & 0.98 & 4.10\\
$\text{QGARCH}_1$ & 0.79 & 3.18 &  & 0.00 & 1.20 &  & 0.40 & 3.57 &  & 0.80 & 3.20 &  & 0.43 & 2.69\\
$\text{QGARCH}_2$ & 0.60 & 4.37 &  & 0.00 & 1.20 &  & 0.60 & 3.97 &  & 0.80 & 3.20 &  & 0.43 & 3.18\\
CAViaR & 0.79 & 4.17 &  & 0.20 & 2.59 &  & 0.60 & 4.37 &  & 0.80 & 2.40 &  & 0.55 & 3.61\\
RiskM & 2.98 & 6.94 &  & 1.99 & 5.18 &  & 3.18 & 6.75 &  & 0.80 & 4.00 &  & 2.57 & 6.12\\[1mm]
& \multicolumn{14}{c}{Dow 30}\\
Hybrid & 1.39 & 4.76 &  & 0.40 & 2.79 &  & 0.79 & 5.15 &  & 0.80 & 4.80 &  & 0.86 & 4.28\\
$\text{QGARCH}_1$ & 0.79 & 2.78 &  & 0.00 & 1.20 &  & 0.20 & 3.56 &  & 1.60 & 4.00 &  & 0.43 & 2.63\\
$\text{QGARCH}_2$ & 0.79 & 2.18 &  & 0.00 & 1.79 &  & 0.79 & 4.55 &  & 0.80 & 4.00 &  & 0.55 & 2.94\\
CAViaR & 0.79 & 4.17 &  & 0.00 & 2.59 &  & 0.79 & 5.16 &  & 0.80 & 4.00 &  & 0.55 & 3.98\\
RiskM & 3.17 & 6.35 &  & 2.19 & 4.78 &  & 2.97 & 6.73 &  & 0.80 & 4.80 &  & 2.63 & 5.87\\[1mm]
& \multicolumn{14}{c}{HSI}\\
Hybrid & 1.39 & 4.56 &  & 0.99 & 3.17 &  & 1.01 & 4.44 &  & 0.82 & 7.38 &  & 1.11 & 4.31\\
$\text{QGARCH}_1$ & 0.79 & 3.37 &  & 0.00 & 2.38 &  & 0.40 & 2.62 &  & 0.82 & 5.74 &  & 0.43 & 3.01\\
$\text{QGARCH}_2$ & 0.99 & 2.78 &  & 0.60 & 2.98 &  & 1.21 & 4.84 &  & 1.64 & 5.74 &  & 0.98 & 3.69\\
CAViaR & 0.79 & 4.17 &  & 0.79 & 3.37 &  & 1.01 & 4.03 &  & 1.64 & 6.56 &  & 0.92 & 4.06\\
RiskM & 1.98 & 7.34 &  & 2.18 & 6.15 &  & 2.22 & 5.65 &  & 4.92 & 7.38 &  & 2.34 & 6.46
\\[1mm]
\hline
\end{tabularx}
\end{center}
\end{table}
\end{landscape}

\begin{landscape}
	\begin{figure}[tbp]
		\begin{center}
			\includegraphics[scale=0.85]{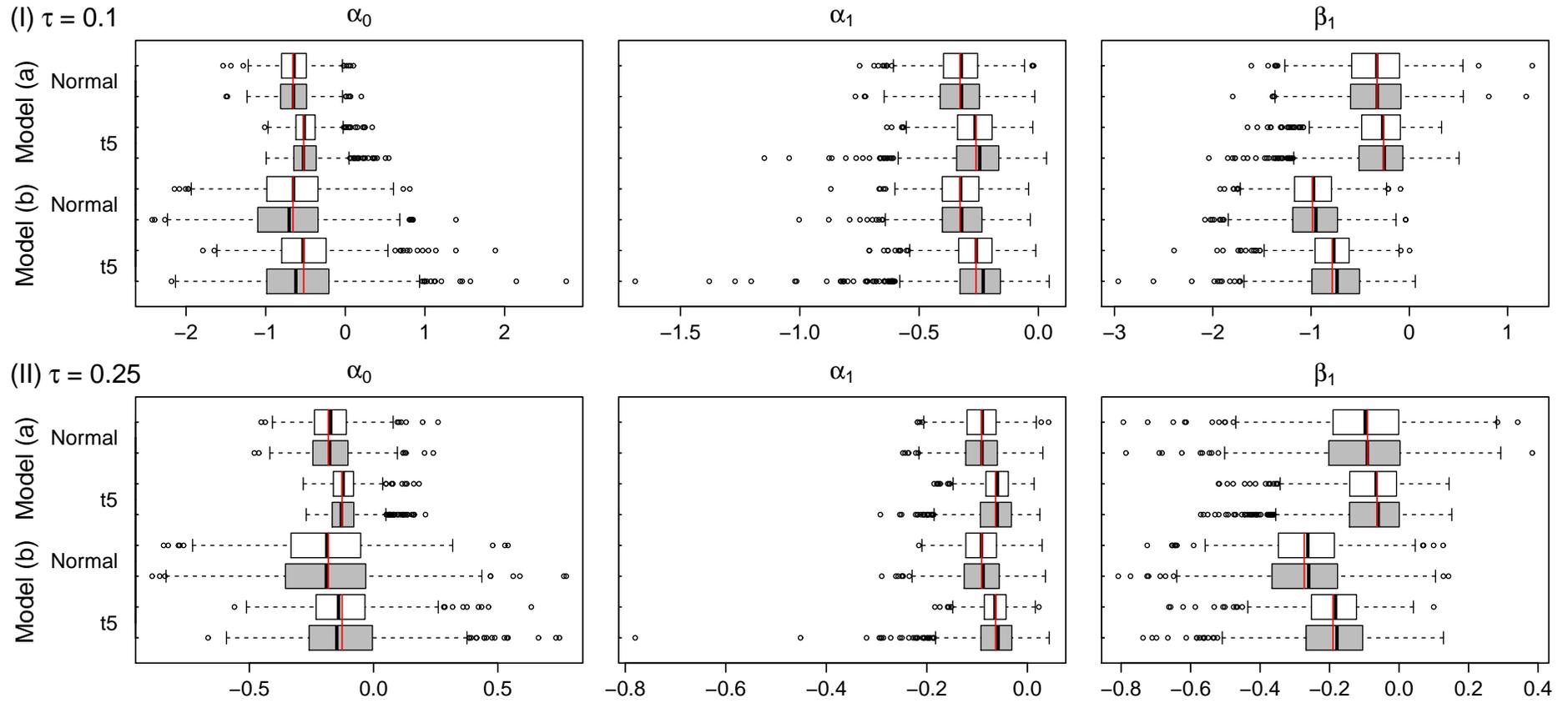}
		\end{center}
		\caption{Box plots for the weighted estimator $\widehat{\theta}_{\tau n}$ (white boxes) and the unweighted estimator $\widecheck{\theta}_{\tau n}$ (grey boxes), at $\tau=0.1$ or 0.25, for two models with normally or Student's $t_5$ distributed innovations. Model (a): $(\alpha_0, \alpha_1, \beta_1)=(0.4,0.2,0.2)$; Model (b): $(\alpha_0, \alpha_1, \beta_1)=(0.4,0.2,0.6)$. The thick black line in the center of the box indicates the sample median, and the thin red line indicates the value of the corresponding element of the true parameter vector $\theta_{\tau 0}$. The notations $\alpha_0$, $\alpha_1$ and $\beta_1$ represent the corresponding elements of  $\widehat{\theta}_{\tau n}$ and $\widecheck{\theta}_{\tau n}$.}
		\label{fig3}
	\end{figure}
\end{landscape}

\begin{figure}[tbp]
\begin{center}
\includegraphics[scale=0.3]{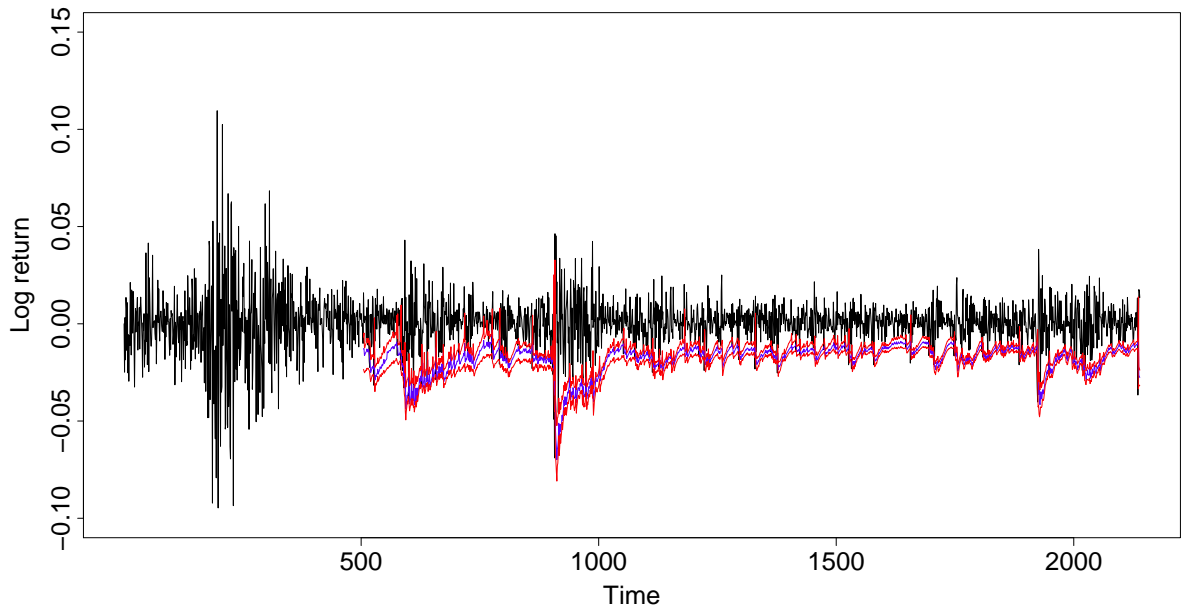}
\end{center}
\caption{Time plot for daily log returns (black line) of S\&P 500 from
January 2, 2008 to June 30, 2016, with rolling forecasts of conditional
quantiles (blue line) at $\protect\tau=0.05$ from January 4, 2010 to June
30, 2016 and corresponding $95\%$ confidence bounds (red lines), using the
proposed conditional quantile estimation and bootstrap method.}
\label{fig1}
\end{figure}

\begin{figure}[tbp]
\begin{center}
\includegraphics[scale=0.6]{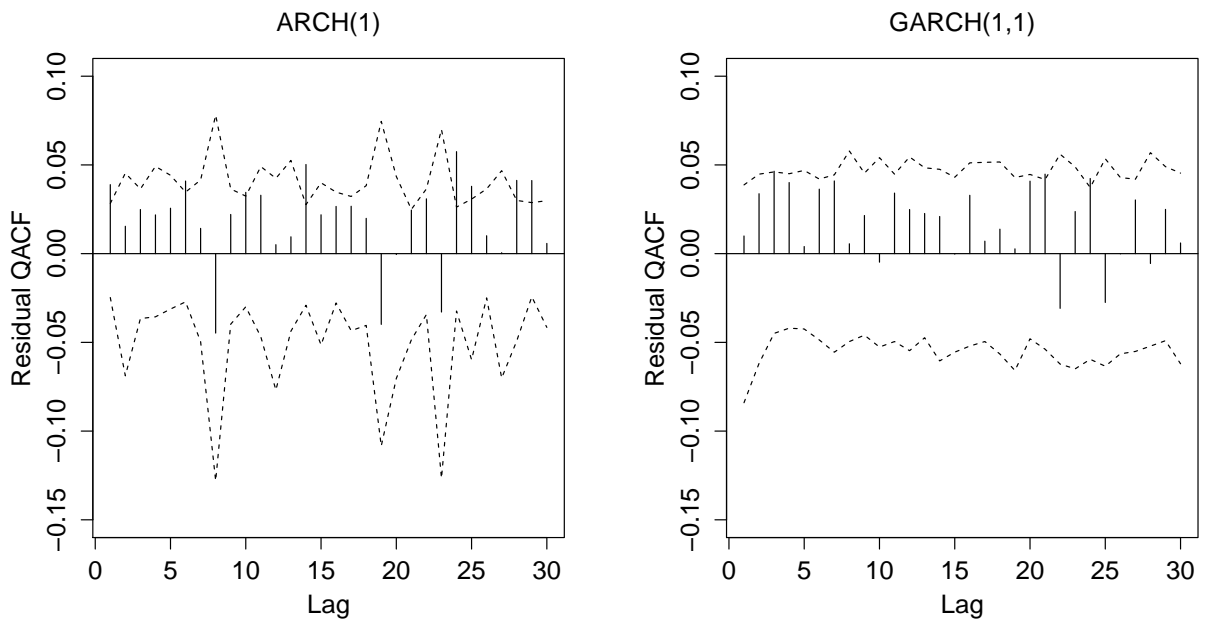}
\end{center}
\caption{Residual QACFs for the fitted conditional quantiles at $\protect\tau%
=0.05$, with corresponding 95\% confidence bounds, for daily log returns of
S\&P 500.}
\label{fig2}
\end{figure}

\end{document}